\documentclass[conference]{IEEEtran}
\usepackage{cite}
\usepackage{csquotes}
\usepackage{amsmath,amssymb,amsfonts}
\usepackage{graphicx}
\usepackage{textcomp}
\usepackage{xcolor}

\usepackage{amsthm}

\usepackage{proof}
\usepackage{stmaryrd}
\usepackage{MnSymbol}
\usepackage{algorithm}
\usepackage{algpseudocode}
\usepackage{hyperref}
\usepackage{wrapfig}

\usepackage{mathtools}
\usepackage{tikz}
\usepackage{tikz-network}
\usepackage{xspace}
\usepackage{xargs}
\usepackage{fontenc}

\usepackage{subcaption}
\usepackage{caption}

\usepackage{tikz-cd}
\usepackage{tikzit}
\usepackage{adjustbox}

\usepackage{xfrac}

\usepackage[final]{microtype}

\newtheorem{theorem}{Theorem}[section]
\newtheorem{definition}[theorem]{Definition}
\newtheorem{lemma}[theorem]{Lemma}
\newtheorem{example}[theorem]{Example}
\newtheorem{proposition}[theorem]{Proposition}
\newtheorem{remark}[theorem]{Remark}

\newcommand\HypI[1]{\textbf{HypI}(#1)}
\newcommand\MdaCospans{\textbf{MHypI}(\Sigma)}
\newcommand{\Ecospans}{{\catname{EHypI(\Sigma)}}}
\newcommand{\MdaEcospans}{{\catname{MEHypI(\Sigma)}}}
\newcommand{\WellTypedMdaEcospans}{{\catname{MEHypI(\Sigma)}}}
\newcommand{\hashtag}{{\#}}
\newcommand{\consistency}{{\smile}}

\newcommand{\missing}[1]{{\color{red}\bfseries [TODO]}}

\newcommand{\defTodo}[2]{%
  \expandafter\newcommand\csname #1\endcsname[1]{%
    \todo[linecolor=#2,backgroundcolor=#2!25,bordercolor=#2,inline,size=\tiny]{\textbf{#1}: ##1}}}

\newcommand{\defTODO}[2]{%
  \expandafter\newcommand\csname #1\endcsname[1]{%
    \todo[linecolor=#2,backgroundcolor=#2!25,bordercolor=#2,inline,size=\tiny,caption={\textbf{(#1 LONG TODO)}}]{##1}}}

\newcommand{\AlekseiColor}{RoyalBlue}

\defTodo{Aleksei}{\AlekseiColor}
\defTODO{ALEKSEI}{\AlekseiColor}
\newcommand{\catname}[1]{\mathbf{#1}}

\newcommand\id{\textsf{id}}
\newcommand\sym{\textsf{sym}}

\tikzstyle{termbox}=[draw=term, fill={term!10}, rounded corners, minimum size=20pt]
\tikzstyle{tallbox}=[draw=term, fill={term!10}, rounded corners, minimum width=20pt, minimum height=40pt]
\tikzstyle{termpic}=[x=1pt, y=1pt, inner sep=0pt, outer sep=0pt, thick]
\tikzstyle{box}=[shape=rectangle, text height=1.5ex, text depth=0.25ex, yshift=0.5mm, fill=white, draw=black, minimum height=12.5mm, yshift=-0.5mm, minimum width=7.5mm, font={\small}]
\tikzstyle{Z dot}=[inner sep=0mm, minimum size=2mm, shape=circle, draw=black, fill={rgb,255: red,160; green,255; blue,160}]
\tikzstyle{Z phase dot}=[minimum size=5mm, font={\footnotesize\boldmath}, shape=rectangle, rounded corners=2mm, inner sep=0.2mm, outer sep=-2mm, scale=0.8, tikzit shape=circle, draw=black, fill={rgb,255: red,160; green,255; blue,160}, tikzit draw=blue]
\tikzstyle{X dot}=[Z dot, shape=circle, draw=black, fill={rgb,255: red,220; green,0; blue,0}]
\tikzstyle{X phase dot}=[Z phase dot, tikzit shape=circle, tikzit draw=blue, fill={rgb,255: red,220; green,0; blue,0}, font={\footnotesize\color{white}\boldmath}]
\tikzstyle{hadamard}=[fill=yellow, draw=black, shape=rectangle, inner sep=0.6mm, minimum height=1.5mm, minimum width=1.5mm]
\tikzstyle{small hadamard}=[hadamard]
\tikzstyle{vertex}=[inner sep=0mm, minimum size=3pt, shape=circle, draw=black, fill=black]
\tikzstyle{vertex set}=[inner sep=0mm, minimum size=1mm, shape=circle, draw=black, fill=white, font={\footnotesize\boldmath}]
\tikzstyle{new style 0}=[draw=term, fill={term!10}, rounded corners, minimum width=25pt, minimum height=30pt]
\tikzstyle{new style 2}=[draw=term, fill={term!10}, rounded corners, minimum width=25pt, minimum height=30pt]
\tikzstyle{round box}=[fill=white, draw=black, shape=rectangle, rounded corners=5pt, font={\small}, minimum height=6mm, minimum width=6mm]
\tikzstyle{e-box}=[fill=white, draw=black, shape=rectangle, minimum width=20mm, minimum height=20mm]
\tikzstyle{empty diag string}=[fill=white, draw={rgb,255: red,165; green,165; blue,165}, shape=rectangle, minimum size=1.2 cm, dashed, thick]

\tikzstyle{new edge style 0}=[-, draw=black, line width=0.5pt]
\tikzstyle{blue 0}=[-, draw=blue, line width=0.8pt]
\tikzstyle{red 0}=[-, draw=red, line width=0.8pt]
\tikzstyle{black dash}=[-, color=black, dashed, dash pattern=on 1.5pt off 1.5pt, draw=black, line width=0.8pt]
\tikzstyle{red dash}=[-, color=red, dashed, dash pattern=on 1pt off 1.5pt, draw=red, line width=0.8pt, line cap=round]
\tikzstyle{blue dash}=[-, color=blue, dashed, dash pattern=on 1pt off 1.5pt, draw=blue, line width=0.8pt, line cap=round]
\tikzstyle{red term dash}=[-, dotted, draw={rgb,255: red,210; green,0; blue,0}, dash pattern=on 1pt off 2pt, line width=0.8pt, line cap=round]
\tikzstyle{brace edge}=[-, tikzit draw=blue, decorate, decoration={brace,amplitude=1mm,raise=-1mm}, line width=0.8pt]
\tikzstyle{gray}=[-, draw={rgb,255: red,191; green,191; blue,191}]
\tikzstyle{arrow}=[<-, draw={rgb,255: red,128; green,128; blue,128}]
\tikzstyle{double-arrow}=[draw={rgb,255: red,128; green,128; blue,128}, <->]
\tikzstyle{thick line}=[-, line width=0.8pt]
\tikzstyle{new edge style 1}=[draw=black, fill=none, ->, >=latex, line width=0.5pt]
\tikzstyle{new edge style 2}=[-, fill=gray]
\tikzstyle{grey identity}=[-, draw={rgb,255: red,191; green,191; blue,191}, color=gray, dashed, dash pattern=on 1.5pt off 1.5pt, draw=black, line width=0.8pt]
\tikzstyle{lambda_unit_string}=[-, draw={rgb,255: red,0; green,0; blue,0},fill={rgb,255: red,238; green,238; blue,255}, thick]
\tikzstyle{hedge}=[fill=white, draw=black, shape=rectangle, rounded corners=2mm, inner sep=0.2mm, outer sep=-2mm, scale=0.8, minimum height=8mm, minimum width=8mm, tikzit category=hypergraph]
\tikzstyle{hedge interface}=[fill=white, draw=black, shape=rectangle, rounded corners=1.5mm, inner sep=0.2mm, outer sep=-2mm, scale=0.8, minimum height=5mm, minimum width=5mm, tikzit category=hypergraph]
\tikzstyle{hedge blue}=[hedge, fill={rgb,255: red,102; green,204; blue,255}, draw=black, shape=rectangle, tikzit category=hypergraph]
\tikzstyle{node}=[fill=black, draw=black, shape=circle, minimum size=1.5mm, inner sep=0mm, tikzit category=hypergraph]
\tikzstyle{red node}=[fill=red, draw=black, shape=circle, minimum size=1.5mm, inner sep=0mm, tikzit category=hypergraph]
\tikzstyle{blue node}=[fill=blue, draw=black, shape=circle, minimum size=1.5mm, inner sep=0mm, tikzit category=hypergraph]
\tikzstyle{green node}=[fill={rgb,255: red,33; green,140; blue,33}, draw=black, shape=circle, minimum size=1.5mm, inner sep=0mm, tikzit category=hypergraph]
\tikzstyle{node highlight}=[fill=black, draw=blue, thick, shape=circle, minimum size=1.5mm, inner sep=0mm, tikzit category=hypergraph]
\tikzstyle{red node highlight}=[fill=red, draw=blue, thick, shape=circle, minimum size=1.5mm, inner sep=0mm, tikzit category=hypergraph]
\tikzstyle{yellow hedge}=[hedge, fill=yellow, draw=black, shape=rectangle, tikzit category=hypergraph]
\tikzstyle{green hedge}=[hedge, fill=green, draw=black, shape=rectangle, tikzit category=hypergraph]
\tikzstyle{small box}=[fill=white, draw=black, shape=rectangle, minimum height=6mm, minimum width=6mm, tikzit category=string diagram]
\tikzstyle{vsmall box}=[fill=black, draw=black, shape=rectangle, minimum height=4mm, minimum width=1mm, tikzit category=string diagram, inner sep=0]
\tikzstyle{medium box}=[fill=white, draw=black, shape=rectangle, minimum height=11mm, minimum width=6mm, tikzit category=string diagram]
\tikzstyle{semilarge box}=[fill=white, draw=black, shape=rectangle, minimum height=16mm, minimum width=6mm, tikzit category=string diagram]
\tikzstyle{large box}=[fill=white, draw=black, shape=rectangle, minimum height=21mm, minimum width=6mm, tikzit category=string diagram]
\tikzstyle{black dot}=[fill=black, draw=black, shape=circle, minimum size=2mm, inner sep=0mm, tikzit category=string diagram]
\tikzstyle{white dot}=[fill=white, draw=black, shape=circle, minimum size=2mm, inner sep=0mm, tikzit category=string diagram]
\tikzstyle{red dot}=[fill=red, draw=black, shape=circle, minimum size=2mm, inner sep=0mm, tikzit category=string diagram]
\tikzstyle{wlabel}=[fill=none, draw=none, shape=rectangle, tikzit category=string diagram, font={\footnotesize}, inner sep=0pt, tikzit fill={rgb,255: red,102; green,204; blue,255}, tikzit draw={rgb,255: red,102; green,204; blue,255}, yshift=0.3mm]
\tikzstyle{BRchange}=[draw=black, shape=diamond, tikzit shape=circle, tikzit fill={rgb,255: red,96; green,0; blue,0}, diamond split part fill={black,red}, inner sep=-5mm, minimum width=2.7mm, minimum height=1.7mm]
\tikzstyle{RBchange}=[draw=black, shape=diamond, tikzit shape=circle, tikzit fill={rgb,255: red,165; green,0; blue,0}, diamond split part fill={red,black}, inner sep=0, minimum width=2.7mm, minimum height=1.7mm]
\tikzstyle{dummy}=[fill=none, draw=none, shape=circle, font={\small}, inner sep=1pt, tikzit draw=blue, tikzit fill=white]
\tikzstyle{node label}=[fill=none, draw=none, shape=rectangle, tikzit fill=cyan, tikzit draw=cyan, font={\scriptsize}, tikzit shape=circle, inner sep=0pt]
\tikzstyle{empty diag}=[fill=white, draw={rgb,255: red,165; green,165; blue,165}, shape=rectangle, minimum size=1.2 cm, dashed, thick]
\tikzstyle{large horizontal box}=[fill=white, draw=black, shape=rectangle, minimum height=6mm, minimum width=21mm, tikzit category=string diagram]
\tikzstyle{medium horizontal box}=[fill=white, draw=black, shape=rectangle, minimum height=6mm, minimum width=11mm, tikzit category=string diagram]
\tikzstyle{semilarge horizontal box}=[fill=white, draw=black, shape=rectangle, minimum height=6mm, minimum width=16mm, tikzit category=string diagram]
\tikzstyle{e-box-interface}=[inner sep=0mm, minimum size=1mm, shape=circle, draw=black, fill=white, font={\footnotesize\boldmath}]
\tikzstyle{very large horizontal box}=[fill=white, draw=black, shape=rectangle, minimum height=6mm, minimum width=41mm, tikzit category=string diagram]

\tikzstyle{dashed edge}=[-, dashed, very thick]
\tikzstyle{alt sort}=[-, dashed, dash pattern=on 2pt off 0.5pt, thick, draw=red]
\tikzstyle{diredge}=[->, >={Latex[length=1.5mm]}]
\tikzstyle{diredge highlight}=[->, >={Latex[length=1.5mm]}, draw=blue, thick]
\tikzstyle{boundary frame}=[-, draw={rgb,255: red,170; green,170; blue,255}, dashed, fill={rgb,255: red,238; green,238; blue,255}, thick, dash pattern=on 2pt off 0.5pt]
\tikzstyle{graph frame}=[-, draw={rgb,255: red,191; green,191; blue,191}, dashed, fill={rgb,255: red,238; green,238; blue,238}, thick, dash pattern=on 2pt off 0.5pt]
\tikzstyle{def sort}=[-]
\tikzstyle{component}=[-, draw=red, thick]
\tikzstyle{map edge}=[{|->}, >=latex, shorten <=0.5mm, shorten >=0.5mm]
\tikzstyle{hypergraph map edge}=[{|->}, draw=red, shorten <=1mm, shorten >=1mm]
\tikzstyle{cdedge}=[->]
\tikzstyle{big cdedge}=[->, very thick, >=latex]
\tikzstyle{pointer edge}=[->, draw=gray, thick]
\tikzstyle{vertical_delimiter}=[-, dashed, very thick, fill=none, draw={rgb,255: red,170; green,170; blue,255}]
\tikzstyle{e_hyperedge}=[-, draw=black, dashed, fill={rgb,255: red,255; green,255; blue,255}, thick, dash pattern=on 3pt off 1.5pt]
\tikzstyle{lambda box}=[-, draw=black, fill={rgb,255: red,255; green,255; blue,255}, thick]
\tikzstyle{lambda_unit}=[-, draw={rgb,255: red,0; green,0; blue,0},fill={rgb,255: red,238; green,238; blue,255}, thick]
\input{hypergraph.tikzdefs}

\makeatletter
\newcommand{\linebreakand}{%
  \end{@IEEEauthorhalign}
  \hfill\mbox{}\par
  \mbox{}\hfill\begin{@IEEEauthorhalign}
}
\makeatother

\def\BibTeX{{\rm B\kern-.05em{\sc i\kern-.025em b}\kern-.08em
    T\kern-.1667em\lower.7ex\hbox{E}\kern-.125emX}}
\begin{document}

\title{Equivalence Hypergraphs: DPO Rewriting for Monoidal E-Graphs}

% \and
% }

\author{\IEEEauthorblockN{\href{https://orcid.org/0000-0003-2724-5211}{Aleksei Tiurin}}
\IEEEauthorblockA{\textit{University of Birmingham}\\
United Kingdom\\
axt257@student.bham.ac.uk}
\and
\IEEEauthorblockN{\href{https://orcid.org/0000-0003-1708-3554}{Chris Barrett}}
\IEEEauthorblockA{\textit{University of Birmingham}\\
United Kingdom\\
c.r.barrett.1@bham.ac.uk
}
\and
\IEEEauthorblockN{\href{https://orcid.org/0000-0002-4003-8893}{Dan R. Ghica}}
\IEEEauthorblockA{\textit{Huawei Central Software Institute}\\ 
\textit{and}\\ 
\textit{University of Birmingham}\\
United Kingdom\\
d.r.ghica@bham.ac.uk}
\linebreakand 
\IEEEauthorblockN{\href{https://orcid.org/0000-0002-4783-9757}{Nick Hu}}
\IEEEauthorblockA{\textit{University of Birmingham}\\
United Kingdom\\
n.hu@bham.ac.uk}
}

\maketitle

\begin{abstract}
	% The technique of equipping graphs with an equivalence relation, called equality saturation, has recently proved both powerful and practical in program optimisation, particularly for satisfiability modulo theory solvers. 
The technique of \emph{equality saturation}, which equips graphs with an equivalence relation, has proven effective for program optimisation.
We give a categorical semantics to these structures, called \emph{e-graphs}, in terms of Cartesian categories enriched over the category of semilattices.
This approach generalises to monoidal categories, which opens the door to new applications of e-graph techniques, from algebraic to monoidal theories.
Finally, we present a sound and complete combinatorial representation of morphisms in such a category,  based on a generalisation of hypergraphs which we call \emph{e-hypergraphs}.
They have the usual advantage that many of their structural equations are absorbed into a general notion of isomorphism.
This new principled approach to e-graphs enables double-pushout (DPO) rewriting for these structures, which constitutes the main contribution of this paper.
   
\end{abstract}

%-------------------- INTRODUCTION

\section{Introduction}\label{sec:introduction}

Rewrite-driven program optimisation consists of the application of a sequence of semantics-preserving rewrites which may improve the cost of execution, \emph{e.g.}, in terms of running time or memory.
However, applying some rewrites may enable or disable the application of subsequent rewrites, so the choice of application order significantly impacts the quality of the resulting optimisation.  
This long-standing issue is known as the \textit{phase-ordering problem}, with a \textit{phase} referring to a particular set of rewrites.
The optimal solution to this problem is often intractably hard to compute, so typical approaches in practice use heuristics to determine a pseudooptimal ordering instead.

A recently proposed alternative solution is \textit{equality saturation} 
\cite{10.1145/1594834.1480915}: instead of maintaining a \textit{single},  putatively optimised program which is rewritten \textit{destructively} (\textit{i.e.}, at each step, forgetting the part of the program that is replaced by the rewrite), a \textit{set} of equivalent programs is maintained instead, where each rewrite step \textit{non-destructively} grows the set.  
Upon reaching a fixed point (\textit{saturation}),  a \textit{globally} optimal program can then be extracted from this set.
A na\"ive approach for this is unfeasible, since the size of this set grows exponentially with the number of rewrites. 
\textit{Equality graphs (e-graphs)} \cite{EggPaper} are a data structure that can represent this set compactly, thus making this technique computationally tractable in practical applications.

Although equality saturation is already a state-of-the-art algorithmic optimisation technique, there is an opportunity to better understand its mathematical foundations and the foundations of e-graphs in particular.
We address this by giving a categorical axiomatisation for e-graphs and their rewriting.
Considering programs as represented by terms of an arbitrary algebraic theory, we demonstrate a correspondence between e-graphs and morphisms of Cartesian categories enriched over the category of semilattices (\textit{semilattice-enriched}).
% The morphisms of the latter can also be depicted using \emph{string diagrams}~\cite{Selinger_2010,joyal_geometry_1991, mellies_functorial_2006}: a two-dimensional syntax that encodes equivalence classes of morphism terms.
% Enriched categories~\cite{Borceux_1994,Kelly2022BASICCO} generalise plain categories by imposing more structure making the categorical language more expressive~\cite{Borceux_1994,Kelly2022BASICCO} which will be useful for modelling equivalence between terms as encountered in e-graphs.
Our approach is guided by the established methodology of creating a correspondence between string diagrams and graph rewriting, for which there are existing precedents in the literature~\cite{bonchi_string_2022-1,bonchi_string_2022-2,bonchi_string_2022-3,fscd,ghica-zanassi2023string}.
The key ingredient of our categorical approach is the semilattice-enrichment.
% \footnote{We could alternatively use the language of idempotent commutative semigroups, which form an equivalent category, but stick to the \enquote{semilattice} terminology for brevity.}
By using the binary operator of a given semilattice we can operate with (non-empty) \textit{sets} of terms (string diagrams, morphisms); the rest of the paper is about working out the mathematical implications of introducing this enrichment to string diagrams and their concrete combinatorial representation.

\emph{Monoidal} category theory has seen applications in digital~\cite{ghica_compositional_2023} and quantum circuits~\cite{coecke_interacting_2011,ZX}, functional programs~\cite{ghica-zanassi2023string}, and computational linguistics~\cite{wazni_quantum_2022,coecke_lambek_2013}.
These topics all have \emph{computation} at their core, where a typical goal is to transform one expression into a \enquote{better} version of itself, with respect to some sound theory of rewriting, analogously to program optimisation.
Our work does not presuppose any particular monoidal theory, but rather provides a general categorical framework, which we hope will serve as a foundation for a principled approach to efficient rewriting theories in these settings.

Finally, we give a concrete combinatorial representation of string diagrams for semilattice-enriched categories in terms of (hierarchical)  hypergraphs, which we call \textit{e-hypergraphs} or monoidal e-graphs, generalising e-graphs for monoidal theories.
Additionally, we also give a specification of e-hypergraph rewriting via a suitable extension of the double pushout (DPO) rewriting framework 
\cite{dpo, bonchi_string_2022-1,bonchi_string_2022-2,bonchi_string_2022-3},  proving the representation \textit{sound and complete} with respect to our categorical semantics.  
As a corollary, this also provides a formalisation of the rewriting theory of e-graphs~\cite{EggPaper}.

\subsection{E-graphs}

\ifdefined \ONECOLUMN
\begin{figure}
	\[
		\scalebox{0.55}{
		\input{./figures/egraph-translation.tikz}
		}
	\]
	\caption{Example translation of acyclic e-graphs into string diagrams for semilattice-enriched Cartesian categories. }
	\label{fig:e-graph-example}
	\end{figure}
\else
\begin{figure*}
\[
	\hspace{1.3cm}
    \scalebox{0.65}{
    \input{./figures/egraph-translation.tikz}
    }
\]
\captionsetup{skip=0pt, belowskip=-5mm}
\caption{Example translation of acyclic e-graphs into string diagrams for semilattice-enriched Cartesian categories. }
\label{fig:e-graph-example}
\end{figure*}
\fi

E-graphs are a data structure which can efficiently represent many equivalent terms of an algebraic theory simultaneously.
They generalise the typical \emph{directed acyclic graph} (DAG) representation of terms to include equivalence classes of subterms (subgraphs).
The key concepts of e-graphs can be intuitively grasped from some simple examples, as shown in Figure~\ref{fig:e-graph-example}.
This example, adapted from~\cite{EggPaper}, illustrates how an optimiser using e-graphs may non-destructively apply a set of rewrites to initial term $(a * 2) / 2$.

Each column represents a term (or term rewrite rule), the conventional e-graph representation, and the equivalent string diagram representation: an e-hypergraph. 
The initial term, $(a*2)/2$ (which corresponds to subfigure (a)), is already represented efficiently in the e-graph by sharing the node 2.
The string diagram version is slightly different to the conventional e-graph version, as it represents the sharing of the node 2 explicitly, using a sharing (contraction) node
\scalebox{0.3}{
	\begin{tikzpicture}
		\begin{pgfonlayer}{nodelayer}
			\node [style=vertex] (0) at (1.5, 3.5) {};
			\node [style=none] (1) at (1.5, 3) {};
			\node [style=none] (2) at (1, 4) {};
			\node [style=none] (3) at (2, 4) {};
		\end{pgfonlayer}
		\begin{pgfonlayer}{edgelayer}
			\draw (0) to (1.center);
			\draw [bend right=45] (0) to (3.center);
			\draw [bend left=45] (0) to (2.center);
		\end{pgfonlayer}
	\end{tikzpicture}
}
. 

Nodes are encapsulated in dashed boxes indicating equivalence classes of nodes, or, more generally, of subgraphs rooted at these boxes.
Initially, each equivalence class has a single node (subgraph).
The first rewrite rule, replacing multiplication by the more efficient \emph{shift-left} operator, creates the first non-unitary (containing more than one node or subgraph) equivalence class, which includes the multiplication ($*$) and shift-left ($<\!\!<$) operator nodes, with a new sharing: that of node $a$.
This second e-graph illustrates the other difference, besides sharing, between conventional e-graphs and e-hypergraphs: in the former, edges connect directly to nodes inside the equivalence class, whereas in the latter, edges connect to the equivalence class itself.
Explicit discarding (delete) nodes, depicted with black dots with a single wire inside the dashed boxes 
\scalebox{0.4}{
	\begin{tikzpicture}
		\begin{pgfonlayer}{nodelayer}
			\node [style=vertex] (0) at (1.5, 3.5) {};
			\node [style=none] (1) at (1.5, 3) {};
		\end{pgfonlayer}
		\begin{pgfonlayer}{edgelayer}
			\draw (0) to (1.center);
		\end{pgfonlayer}
	\end{tikzpicture}
}, indicate which of the class-level edges are connected to which nodes inside the class.

The other steps, (c) to (d), represent further elaborations of the e-graph, or e-hypergraph, via the application of more rules.
If we were to apply another rule, $x * 1 \to x$, we would identify the e-class (dashed box) containing $a$ with the e-class containing $a * 1$ (the topmost dashed box) which would require a cycle to be introduced: the left argument of node $*$ would be the same e-class that $*$ belongs to.
Then, an optimal term $a$ could be extracted from the e-graph by traversing all the e-classes, where its optimality is justified by its minimal size as a graph.
To support such cycles we need a richer categorical structure for our string diagrams.
This is obtained by combining this work with the results of~\cite{ghica_rewriting_2023}, and for simplicity we will limit ourselves to the case of acyclic e-graphs, although this is not a limitation of the general theory.
% While equipping the string diagrams which such structure is straightforward, it is out of the scope of this work, and thus we limit ourselves to acyclic e-graphs and e-string diagrams for the purposes of presentation.
Note that while the transformations of the e-graph are guided by the algorithm, the transformations of the string diagrammatic version are guided by the equations of the underlying theory.
A detailed step-by-step application of the equations for the rewrite rules from $(a)$ to $(b)$ and $(b)$ to $(c)$ is given in Figures~\ref{fig:e-graph-example-a-b} and~\ref{fig:e-graph-example-b-c}, respectively.
The equations applied are the ones from Figure~\ref{fig:string-equations} along with the naturality of delete and copy.
Note that the degree of sharing in the eventual string diagram is comparable to the degree of sharing in the corresponding e-graph.
% Note that in the final e-graph, (e), a cycle is introduced, making it the representation of infinitely many terms: $a, a*1, (a*1)*1, \ldots$. 
% From the saturated graph the optimal term consisting of only the node $a$ can be then extracted, discarding everything else. 

\subsection{Semilattice Enriched Categories}

\begin{figure}
\[
	\scalebox{0.65}{
	\input{./figures/egraph-strings.tikz}
	}
\]
\captionsetup{skip=0pt, belowskip=-4ex}
\caption{String diagrams for semilattice-enriched symmetric monoidal categories.}
\label{fig:egraph-strings}
\end{figure}

Recall that there is a correspondence between DAG representations of algebraic terms and  string-diagrams for morphisms of a suitable Cartesian category (see \textit{e.g.}~\cite{ghica-zanassi2023string}).
% Recalling the correspondence between DAG representations of algebraic terms and  string-diagrams for morphisms of a suitable Cartesian category (see \textit{e.g.}~\cite{ghica-zanassi2023string})
It makes perspicuous the correspondence between (acyclic) e-graphs and morphisms of a semilattice-enriched Cartesian category, taken as string diagrams.
We extend the usual vocabulary of string diagrams for symmetric monoidal categories with an additional generator:
\[
\infer{f_1 + \ldots + f_n: A \to B}{\{f_i: A \to B\}_{i \in \{1,\ldots,n\}}}
\]
Taking formal (non-empty) joins of morphisms in this way is used to model the equivalence class structure of e-graphs.

String diagrams are read from bottom-to-top, and we implicitly consider additional generators for the duplication and deletion transformations of a (semilattice-enriched) Cartesian category, and we call the string diagrams for this enriched setting \emph{enriched string diagrams}.
Note that in the Cartesian case we can restrict generating morphisms $c$ to have a single output,  without loss of generality,  by applying the universal properties of the Cartesian product.  
This means the (informal) translation given next is fully general for acyclic e-graphs. 

The translation from e-graphs to enriched string diagrams is illustrated informally in Figure~\ref{fig:e-graph-to-string},
noting how the typing constraints of the $+$ constructor are satisfied in the image of the translation by discarding, in each component, all unnecessary inputs.
Examples of this translation are given in the second row of Figure \ref{fig:e-graph-example}.

\begin{figure}
\[
    \input{./figures/egraph-translation-1.tikz}
\]
\captionsetup{skip=0pt}
\caption{E-graph e-class translation to string diagram.}
\label{fig:e-graph-to-string}
\end{figure}
% The subfigure on the right is slightly simplified: if we denote the arity of each $c_i$ as $a_{c_{i}}$, then the dashed box on the right should contain $\sum\limits_{i} a_{c_{i}}$ inputs and each component $c_j$ has $\sum\limits_{i \not = j} a_{c_{i}}$ weakening nodes inside.
% Examples of this translation are given in the second row of Figure \ref{fig:e-graph-example}.
% In particular, note that the translation of the \textit{cyclic} e-graph (e) requires the use of a categorical \textit{trace}, generating a cycle. 
% In this paper, we focus on \textit{acyclic} e-graphs and their corresponding categories; but the above example illustrates that the extension to the cyclic case seems to be routine. As it would involve needless additional presentational complexity we leave it as further work. 

The motivating example of our work, e-graphs, requires the Cartesian structure to express arbitrary sharing (see Figure~\ref{fig:e-graph-example}).
However, we will see how all the core of the theory of e-hypergraph rewriting only requires the monoidal structure. 
Thus, we can readily generalise e-graphs from algebraic to monoidal theories, giving rise to a  host of new potential applications, which we briefly outline in the conclusion of this paper.   

\subsection{Combinatorial Representation of Enriched String Diagrams}

In order to \textit{implement} generalised e-graphs, rewriting must be performed on concrete representations of the corresponding string diagrams.  
String diagrams can be considered equivalently as either topological objects (\textit{i.e.}, taken modulo ``connectivity'') or as a 2-dimensional syntax quotiented by the equations of a \textit{symmetric monoidal category (SMC)}.
For example, we have the following equivalences between diagrams:
%\[(f_1 \otimes f_2) ; (g_1 \otimes g_2) = (f_1 ; g_1) \otimes (f_2 ; g_2)\] 
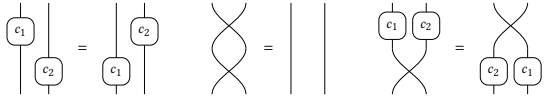
\[
	\scalebox{0.6}{
	\input{./figures/interchange.tikz}
	}
\]
The latter representation makes string diagrams unamenable to efficient implementation due to the difficulty of calculating the quotient.
\emph{Hypergraphs} are a combinatorial representation designed to address this issue~\cite{bonchi_string_2022-1,bonchi_string_2022-2,bonchi_string_2022-3}.
% A large body of research work shows how this issue can be solved by taking a different approach, representing string diagrams as hypergraphs~\cite{bonchi_string_2022-1,bonchi_string_2022-2,bonchi_string_2022}.  
Here, the generators $c_i$ become hyper-edges,  and wires become vertices.  
Thus, the expected quotient becomes simply hypergraph isomorphism. 
With appropriate restrictions on the form these hypergraphs can take, this approach can be used to \textit{characterise} the free symmetric monoidal category generated by some~$c_i$. 

We are interested not only in the free SMC over a set of generators, but also in \textit{(symmetric) monoidal theories} with extra equations, such as the rewrite rules (b)--(d) of Figure \ref{fig:e-graph-example}. 
These can be seen as equations between string diagrams, and thus rewrites of their hypergraph representations.  
Because the generating equations can be applied in any context, we are led to the notion of \textit{(hyper)graph rewriting}: given an equation $l=r$, we require some way to identify (a sub-hypergraph corresponding to) $l$ within a hypergraph and replace it with (a sub-hypergraph corresponding to) $r$.
The standard methodology to giving specifications of such graph rewriting, known as \textit{double pushout (DPO) rewriting} \cite{dpo, bonchi_string_2022-1}, is still applicable.
However, these concepts must be now generalised from hypergraphs to \textit{e-hypergraphs} --- hypergraphs with two additional relations denoting the hierarchical structure introduced by so-called \textit{e-boxes}: the generator for semilattice enrichment, and the separation of the components of each e-box.
The technical aspects of finding the correct definitions are complex, and constitutes the main body of this paper. 

\subsection{Soundness and Completeness}
The main technical result of the paper is a proof of soundness and completeness of this representation with respect to morphisms of a semilattice-enriched symmetric monoidal categories and their respective rewriting theory.
This extends the results of \cite{bonchi_string_2022-2} for plain SMCs. 
While the structural equalities of SMCs are factored out in the representation,  those arising from the enrichment (see the equations of Figure \ref{fig:string-equations}) are not, and should not be. 
They represent the un/sharing of subdiagrams with respect to the e-box structure,  which is precisely what allows for the compact representation of equivalence classes.  
Instead, we consider DPO-rewrites implementing both the structural equalities for enrichment (which involve the e-box structure) and the equalities arising from the generating monoidal theory (which do not).  

Precisely,  our soundness and completeness result is the statement that morphisms in an appropriate free semilattice-enriched SMC are equal  if and only if there exists a sequence of DPO-rewrites --- each induced by a structural equality or the monoidal theory --- between their combinatorial representations. 
In the particular case of a monoidal signature including natural copy and delete maps,  we thus recover a sound and complete representation of semilattice-enriched Cartesian categories,  with the translation described earlier in Figure~\ref{fig:e-graph-to-string} justifying our claim of developing a mathematical theory of e-graphs. 

%\subsection{E-matching and E-rewriting}
%
%Having defined e-hypergraphs and shown them to correctly model the rewriting of string diagrams for semilattice-enriched SMCs, we note that the naive implementation of DPO-rewriting is inefficient: to find a redex within an e-hypergraph involves finding a structurally equivalent e-hypergraph which contains the redex as a subgraph. In general, this involves unsharing the e-box structure before a redex becomes available.  Thus, we define a notion of \textit{e-matching} for e-hypergraphs: that is, to find redexes working modulo the e-box structure. In particular, we wish to locate the smallest subgraph $G' \subseteq G$ "containing" (in an appropriate sense) the redex $L$. Given this,  \textit{e-rewriting} $L \to R$, intended to achieve equality saturation, is simple to define due to its non-destructive nature: we rewrite $G' \to G' + R$ in $G$.  
%\begin{figure}
%\[
%	\tikzfig{combinatorial_semantics/example-rewriting}
%\]
%\caption{Example application of e-rewriting for rewrite $(x*y)/z \to x*(y/z)$: find an appropriate minimal sub-e-hypergraph containing the redex $(x*y)/z$ and non-destructively add the reduct $x*(y/z)$.  REDO THIS}
%\end{figure}

\subsection{Related Work}

% E-graphs
Although e-graphs were first developed in the 80s \cite{nelson1980techniques}, there has recently been an explosion of interest in them for the purpose of building program optimisers and synthesisers, especially due to recent work on the equality saturation technique \cite{10.1145/1594834.1480915, griggio_proceedings_2022, EggPaper,flatt_small_2022}.  This includes interest in various extensions of the e-graph formalism, for example to account for conditional rewriting \cite{singher2023colored},  to reason about rewriting for the $\lambda$-calculus \cite{koehler2022sketchguided},  and to combine techniques of e-graphs and database queries (``relational e-matching'') \cite{zhang_relational_2022}. 
We are hopeful that our novel perspective on e-graphs can be extended to give uniform foundations to these investigations.  
We elaborate on some potential applications in the \nameref{sec:conclusion}.

% String diagrams
Our use of string diagrams is also central to the intuition behind our work \cite{Selinger_2010, joyal_geometry_1991}.  
There is a breadth of work on similar string diagrammatic formalisms,  including hierarchical string diagrams \cite{ghica-zanassi2023string} and functorial boxes \cite{mellies_functorial_2006},  proof nets for compact closed categories with biproducts \cite{duncan_generalised_2009}, and recent work on tape diagrams for rig categories with biproducts \cite{bonchi_tape_nodate}. 
% Hierarchical hypergraphs
Hierarchical (hyper-)graphs are also a well-studied area of research \cite{plump:hierarchical-graphs, montanari:gs-lambda, palacz:hierarchical-transform, Gaducci:hierarchical-graphs, Ghica:hierarchical}. 
% String diagram rewrites
Our work builds heavily on string diagram rewrite theory as developed in \cite{bonchi_string_2022-1,bonchi_string_2022-2, bonchi_string_2022-3}. 

\section{Categorical Semantics of E-Graphs}
In this section,  we introduce some preliminaries on semilattice-enriched symmetric monoidal categories generated by monoidal theories,  and the string diagrammatic formalism we will use to represent them.  

Given a category $\mathbb{C}$  with objects $A,B \in \mathbb{C}$ we denote by $\mathbb{C}(A,B)$ the corresponding hom-set.  
We write the identity morphism on $A$ as $\id_A$.  
We commonly write $f;g$ for composition in diagrammatic order.  
Composition in the usual order is written $g \circ f$.  
We denote the tensor product of an SMC $\mathbb{C}$ by $\otimes$,  its unit by $I$, and its symmetry natural transformation as $\sym$ \cite{maclane}.  
We adopt the convention that $\otimes$ binds more tightly than $(;\!)$.  
We elide all associativity and unit isomorphisms associated with monoidal categories,  and often omit subscripts on identities and natural transformations where it can be inferred.  
We denote by $\catname{SLat}$ the closed monoidal category of semilattices which is defined in section~\ref{sec:appendix:slat} of the Appendix.
Given two adjoint functors $F : \mathbb{A} \to \mathbb{B}$ and $G : \mathbb{B} \to \mathbb{A}$ with $F$ being left-adjoint, we write $F \dashv G : \mathbb{B} \to \mathbb{A}$ for this adjunction.

\subsection{Symmetric Monoidal Theories and PROPs.}

A presentation of an algebraic theory is traditionally given by a \textit{signature} of $n$-ary operations and a set of \textit{equations}---formally,  pairs of terms freely generated over the signature which are identified.  We are interested here in symmetric monoidal theories, which are generalisations of algebraic theories where the operations of the signature may have arbitrary \textit{co-arities}.  
First,  we generalise the notion of a signature. 
\begin{definition}[Monoidal signature]
A \textit{(monoidal) signature} $\Sigma$ is given by a set of \textit{generators} $c: n \to m$,  with \textit{arity} $n$ and \textit{co-arity} $m$.  %A \textit{signature homomorphism} $h: \Sigma \to \Gamma$ is a function from $\Sigma$ to $\Gamma$ which preserves the (co-)arities of elements. 
\end{definition}

% \begin{remark}
%     Below we will assume that for all $c_1,c_2 \in \Sigma$, if arity (respectively, co-arity) of $c_1$ is 0, then co-arity (respectively, arity) of $c_2$ must be at least 1.
% 	This is to ensure we do not have terms of type $0 \to 0$.
% \end{remark}

A \textit{symmetric monoidal theory} is then defined using $\Sigma$-terms.
A set of such terms is the set obtained by taking well-typed combinations of $c \in \Sigma, \textsf{id}_{I}, \textsf{id}_{1}, \text{sym}_{1,1}, (;), \text{ and } \otimes$.
A categorical presentation of a set of $\Sigma$-terms is given by a freely generated products and permutations category.

\begin{definition}[Products and permutations category]
A \textit{products and permutations category (PROP)} is a strict symmetric monoidal category with natural numbers as objects,  and such that $n \otimes m = n+m$.  
A \textit{PROP functor} is a strict SMC-functor which is additionally identity-on-objects.
We denote the free \textit{products and permutations category} over a monoidal signature $\Sigma$ by $\catname{S}(\Sigma)$.
Such PROPs and functors between them define the category $\catname{PROP}$ and, in particular, $\catname{S}(\Sigma) \in \catname{PROP}$.
\end{definition}
The free category $\textbf{S}(\Sigma)$ can be syntactically constructed by quotienting the set of $\Sigma$-terms by the axioms of a symmetric monoidal category.
Thus, $\textbf{S}$ stands for \textit{syntactic} as another way of saying free.

An equation associated with a monoidal signature is a pair of parallel $\Sigma$-terms, leading to the following definition of a symmetric monoidal theory. 
\begin{definition}[Symmetric Monoidal Theory (Definition 2.1~\cite{bonchi_string_2022-1})]
A \textit{presentation of a symmetric monoidal theory} $(\Sigma, \mathcal{E})$ is given by a pair of a monoidal signature $\Sigma$ and a set $\mathcal{E}$ of pairs of parallel $\Sigma$-terms $f,g: n \to m$.
A \textit{symmetric monoidal theory} $\textbf{SMT}(\Sigma,\mathcal{E})$ generated by a presentation $(\Sigma, \mathcal{E})$ is given by a set of $\Sigma$-terms quotiented by the least congruence including $\mathcal{E}$.
\end{definition}

\begin{definition}(Definition 2.4~\cite{bonchi_string_2022-1})
Let $\catname{S}(\Sigma, \mathcal{E})$ a free PROP presented by $\catname{SMT}(\Sigma, \mathcal{E})$.
That is, it is given by $\catname{S}(\Sigma)$ additionally quotiented by the least congruence including $\mathcal{E}$.
\end{definition}

It will often be convenient to give presentations of SMTs in terms of string diagrams.
\begin{example}
The SMT of \textit{commutative comonoids} is given by the following generators, ${\Delta, !}$, depicted in string diagrammatic form:
\[
	\scalebox{0.8}{
  	 \input{./figures/Cartesian-equipment.tikz}
	}
\]
together with the following associativity, commutativity and unitality equations, given in terms of equations of string diagrams: 
\[
	\scalebox{0.6}{
	\input{./figures/Cartesian-theory.tikz}	
	}
\]
A Cartesian SMT is given by a set of generators $\Sigma_C$ which additionally contains the generators of commutative comonoids, together with the equations $\mathcal{E}$ that include the SMT of commutative comonoids, and the following additional naturality of copy-delete equations; for every $c \in \Sigma_C$:
\[
	\scalebox{0.6}{
	\input{./figures/Cartesian-naturality.tikz}
	}
\]
By Fox's Theorem~\cite{fox},  $\textbf{S}(\Sigma_C, \mathcal{E}_C)$ is equivalent to the free Cartesian category over $\Sigma_C$. 
\end{example}

\subsection{Free Semilattice Enrichment: Formal Joins}

We will use enrichment to axiomatize the e-box (e-class) structure as the ``join" of two morphisms $f,g: A \to B$ as $f + g: A \to B$. 
More precisely our hom-sets will be given the structure of a semilattice.
% Note, in particular,  that we do not require semilattices to have a unit for $+$. 
% This is important, since, in any commutative monoid enriched category,  if there exists a categorical product, then it must be a categorical biproduct \cite{maclane}---a degeneracy we wish to avoid.  
Generally,  we may define a category \textit{enriched} in any other (closed) monoidal category $M$, but we restrict ourselves to the concrete definition required for our application, where $M = \catname{SLat}$.
%  (\textit{e.g.},  $\catname{SLat}$).
% In this case,  hom-sets are generalised to hom-objects of $M$ (\textit{e.g.},  a semilattice) and composition 
% \[
% 	\circ: Hom(B,C) \otimes Hom(A,B) \to Hom(A,C)
% \]
% is defined as an $M$-morphism which satisfies certain axioms.
% Composition being an $M$-morphism implies it respects the additional structure on hom-sets.
%  In other words,  composition must respect the additional structure on hom-sets.  
% In our case,  the category we wish to enrich is also monoidal,  so we additionally ask that the monoidal structure also respects this structure.  Technically, this amounts to asking that the tensor product is an enriched functor. 
% We omit further details of the general case of enrichment here, and instead work with the concrete definition required for our application, where $M = \catname{SLat}$.

\begin{definition}[Semilattice-enriched category]
A semilattice-enriched category $\mathbb{C}$ is defined by the following data
\begin{itemize}
	\item a set of objects $\mathbb{C}$
	\item for every pair of objects $A,B \in \mathbb{C}$ --- a hom-object (or, hom-semilattice) $\mathbb{C}(A,B) \in \catname{SLat}$
	\item for every triple of objects $A,B,C \in \mathbb{C}$ --- a composition morphism
	\[
		\circ: \mathbb{C}(B,C) \otimes \mathbb{C}(A,B) \to \mathbb{C}(A,C)
	\]
	\item for every object $A \in C$ --- a unit morphism
	\[
	u_{A} : I \to \mathbb{C}(A,A), \text{ where $I$ is the monoidal unit for $\catname{SLat}$}
	\]
\end{itemize}
such that particular coherence diagrams commute~\cite{Borceux_1994}.
We will also call such categories $\catname{SLat}$-categories.
\end{definition}

Intuitively such enrichment equips every hom-set with a commutative, associative and idempotent operator $+$ that is respected by the composition in the following sense
\[
	f ; (g+h) = f;g + f;h \qquad (f+g) ; h = f;h + g;h
\]
for all appropriately typed morphisms. 
We take $(;\!)$ to bind more tightly than $+$.

In particular,  that we do not require semilattices to have a unit for $+$, \textit{i.e.,} we work with \textit{unbounded} semilattices.
% We also will not utilize the ordering that is induced by the operator of a semilattice in any meaningful way, so a semilattice can be thought of as a commutative idempotent semigroup throughout the paper.
% We stick to the term semilattice as it is shorter.

\begin{definition}[$\catname{SLat}$-functor]
	Let $\mathbb{C}$ and $\mathbb{D}$ be two $\catname{SLat}$-categories.
	An $\catname{SLat}$-functor $F : \mathbb{C} \to \mathbb{D}$ is defined by the following data.
	\begin{enumerate}
	  \item A mapping $F : \mathbb{C} \to \mathbb{D}$
	  \item An object-indexed family of morphisms in $\catname{SLat}$ $F_{A,B} : \mathbb{C}(A,B) \to \mathbb{D}(FA,FB)$
	\end{enumerate}
	such that certain coherence diagrams commute~\cite{Borceux_1994}.
	$F_{A,B}$ being a morphism in $\catname{SLat}$ has a property $F_{A,B}(f + g) = F_{A,B}(f) + F_{A,B}(g)$.
	An $\catname{SLat}$-functor $F$ defines an equivalence between two $\catname{SLat}$-categories if it is \textit{full}, \textit{faithful}, and \textit{essentially surjective}~\cite{Kelly2022BASICCO}.
\end{definition}

\begin{definition}[$\catname{SLat}$-SMC]\label{def:enriched-prop}
A \textit{semilattice-enriched strict SMC}  $\mathbb{C}$ is given by an $\catname{SLat}$-category $\mathbb{C}$ as above that additionally has
\begin{itemize}
\item a unit object $I_{C} \in \mathbb{C}$
\item for every pair of objects $A,B \in \mathbb{C}$ --- an object $A \otimes B \in \mathbb{C}$
\item for all $A,B,C,D$ --- a tensor morphism $ - \otimes_{C} - : \mathbb{C}(A,C) \otimes \mathbb{C}(B,D) \to \mathbb{C}(A \otimes B, C \otimes D)$
\end{itemize}
such that particular coherence diagrams commute~\cite{enriched_monoidal}.
\end{definition}
The latter morphism being in $\catname{SLat}$ implies the following equations that $\otimes$ satisfy
\[
f \otimes (g+h) = f \otimes g + f \otimes h \qquad (f+g) \otimes h = f \otimes h + g \otimes h
\]

Again, we take $\otimes$ to bind more tightly than $+$.

\begin{proposition}(A specialized case of Proposition 6.4.3~\cite{Borceux_1994}. See also Appendix~\ref{sec:appendix:slat})
There is a 2-adjunction 
% https://q.uiver.app/#q=WzAsMixbMCwwLCJcXGNhdG5hbWV7U0xhdHR9LVxcY2F0bmFtZXtDYXR9Il0sWzIsMCwiXFxjYXRuYW1le0NhdH0iXSxbMSwwLCJGIiwyLHsiY3VydmUiOjR9XSxbMCwxLCJHIiwyLHsiY3VydmUiOjR9XSxbMiwzLCIiLDAseyJsZXZlbCI6MSwic3R5bGUiOnsibmFtZSI6ImFkanVuY3Rpb24ifX1dXQ==
% \[\begin{tikzcd}
% 	\catname{SLatt}\text{--}\catname{Cat}\arrow[rr, "\mathcal{U}", bend left] & \hspace{-1em}\top & \catname{Cat} \arrow[ll, "\mathcal{F}", bend left]
% 	\end{tikzcd}
% \]
\vspace{-1mm}
\[
\mathcal{F} \dashv \mathcal{U} : \catname{SLat}\text{--}\catname{Cat} \to \catname{Cat}
\vspace{-1mm}
\]
that is induced by a usual free-forgetful adjunction 

% \[\begin{tikzcd}
% 	\catname{SLatt} \arrow[rr, "U", bend left] & \top & \catname{Set} \arrow[ll, "F", bend left]
% 	\end{tikzcd}
% \]
\vspace{-1mm}
\[
F \dashv U : \catname{SLat} \to \catname{Set}
\vspace{-1mm}
\]
\end{proposition}
In particular, 2-functor $\mathcal{F}$ turns every category $\mathbb{C'} \in \catname{Cat}$ into a free $\catname{SLat}$-category $\mathbb{C} \in \catname{SLat}\text{--}\catname{Cat}$ by making every hom-set of $\mathbb{C'}$ into a free semilattice on this set.
% SMC $(\mathbb{C}, \otimes, 1,+)$ where every hom-set additionally has the structure of a semilattice $(\mathbb{C}(A,B), +)$ -- that is,  a set equipped with an associative,  commutative,  and idempotent operation -- which respects the composition and tensor product in the following ways,  for all appropriately typed morphisms.
% \begin{align*}
% f ; (g+h) &= f;g + f;h &
% (f+g) ; h &= f;h + g;h \\
% f \otimes (g+h) &= f \otimes g + f \otimes h & 
% (f+g) \otimes h &= f \otimes h + g \otimes h,
% %\item A set $ob(\mathbb{C})$ of objects;
% %\item For every pair $A,B \in ob(\mathbb{C})$, a hom-semigroup $Hom(A,B)$;
% %\item An element $\textsf{id}_{A,B} \in Hom(A,B)$; 
% \end{align*}
% Note,  we take $(;\!)$ and $\otimes$ to bind more tightly than $+$.
% A \textit{semilattice-enriched strict SMC functor} $F: (\mathbb{C}, \otimes, 1,+) \to (\mathbb{D}, \otimes, 1,+)$ is given by a strict SMC functor on the underlying SMCs which additionally satisfies $F(f+g) = F(f)+F(g)$.
% All definitions extend to \textit{PROP} when the manipulated categories are PROPs rather than SMCs. 
% We denote the \textit{freely enriched PROP (SMT)} over a monoidal signature $\Sigma$ as $\textbf{PROP}^+(\Sigma)$ ($\textbf{SMT}^+(\Sigma, \mathcal{E})$,  respectively). 
% \end{definition}
\begin{definition}
We define $\catname{S}(\Sigma)^{+}$ and $\catname{S}(\Sigma, \mathcal{E})^{+}$ to be $\catname{SLat}$-SMCs obtained as $\mathcal{F}(\catname{S}(\Sigma))$ and $\mathcal{F}(\catname{S}(\Sigma, \mathcal{E}))$ respectively.
Both of these $\catname{SLat}$-categories are in $\catname{PROP}^{+}$ which we define to be the image of $\catname{PROP}$ via the same adjunction.
\end{definition}
Equivalently, they can be constructed from taking (typed) $\Sigma^+$-terms  -- namely, those freely constructed from generators $c \in \Sigma$, $\textsf{id}_1$, $\sym_{1,1}$, $(;\!)$ and $\otimes$ (and $f+g$, respectively) -- quotiented by the axioms of an enriched symmetric monoidal category (and the least congruence including $\mathcal{E}$ in the case of $\catname{S}(\Sigma, \mathcal{E})^{+}$).
% The free category $\textbf{PROP}^+(\Sigma)$ can be syntactically constructed from taking (typed) $\Sigma^+$-terms  -- namely, those freely constructed from generators $c \in \Sigma$, $\textsf{id}_1$, $\sym{1,1}$, $(;\!)$ and $\otimes$ (and $f+g$, respectively) -- quotiented by the axioms of an enriched symmetric monoidal category.

To aid reasoning, we introduce a new language of string diagrams for $\catname{SLat}$-SMCs, using a hierarchical ``box'' structure to capture the join operation on morphisms.  This is  used in the translation of equivalence classes from the e-graph to the string diagrammatic setting. 

Figure \ref{fig:egraph-strings} displays the generators of this language which is the usual string diagrammatic syntax~\cite{Selinger_2010} apart from the last component which is our notation for $+$; the first component denotes an empty diagram. 
Figure \ref{fig:string-equations} displays the additional equations which these diagrams satisfy, in addition to the standard SMC equations. 
The first four equations are those displayed in Definition \ref{def:enriched-prop},  while the final four axiomatize $+$ as an \textit{n-ary} associative, commutative and idempotent operation.  We overload the binary notation $+$ for our $n$-ary notation.
% We will later prove the intuitive fact that these diagrams are sound and complete with respect to their intended categorical semantics, noting that similar diagrammatic languages using boxes to express choice have been used before~\cite{duncan_generalised_2009}. 
Notably, similar diagrammatic languages using boxes to express choice have been used before~\cite{duncan_generalised_2009}. 

\ifdefined \ONECOLUMN
\begin{figure}
	\[  
		\scalebox{0.6}{
		\input{./figures/egraph-strings-equations.tikz}
		}
	\]
	\caption{Equations for a  semilattice-enriched symmetric monoidal category}
	\label{fig:string-equations}
	\end{figure}
\else
\begin{figure*}
\[  
    \scalebox{0.75}{
	\input{./figures/egraph-strings-equations.tikz}
    }
\]
\captionsetup{skip=0pt}
\caption{Equations for a  semilattice-enriched symmetric monoidal category}
\label{fig:string-equations}
\vspace{-8mm}
\end{figure*}
\fi

\section{String Diagram Rewrite Theory}\label{sec:combinatorial-semantics}

In this section,  we recall the fundamental definitions and results of string diagram rewrite theory for symmetric monoidal categories,  recapitulating the correspondence between string diagram rewriting and an appropriate notion of double pushout (DPO) rewriting of certain hypergraphs,  as established in~\cite{bonchi_string_2022-1, bonchi_string_2022-2}.  

First,  we fix some notation.  Let $(-)^*$ be the free monoid monad over $\catname{Set}$.
We extend this to the case of relations $R \subseteq V \times V$,  so that we denote by $R^{*} \subseteq V^* \times V^*$ the element-wise extension of $R$ to a relation on ordered sequences. 
We will sometimes use functional notation for relations,  so that given $R \subseteq V \times W$ and $v \in V$,  $R(v) \subseteq W$. 
We elide associativity and unit isomorphisms associated with coproducts $(+)$,  and denote by $\iota_j: X_{j} \rightarrow X_{1} + \ldots X_{j} + \ldots X_{n}$ the $j^{\text{th}}$ injection,  and $[f,g]: A_1 + A_2 \to B$ the co-pairing of two morphisms $f: A_1 \to B$ and $g:A_2 \to B$. 
We denote by $A +_{f,g} B$ the pushout of the span $A \xleftarrow{f} C \xrightarrow{g} B$.
We will also refer to $A,B$ as \textit{feet} of the cospan, and to $C$ as \textit{carrier}.

\subsection{Hypergraphs}

% In ~\cite{bonchi_string_2022-1},  hypergraphs \textit{with interfaces} and their DPO rewriting are shown to be sound and complete for categorical presentations of SMTs including a Frobenius algebra. 
% Hypergraphs generalise graphs by allowing edges to have multiple sources and targets. 
% When interpreting string diagrams as hypergraphs,  morphisms are interpreted as edges,  and wires as vertices.  
% However,  there is also a need to specify which vertices are to be considered as \textit{input} and \textit{outputs},  corresponding to the dangling wires at the bottom and top of a string diagram. 
% This gives rise to the definition of a hypergraph with interfaces.  
% Subsequent work \cite{bonchi_string_2022-2} builds on this correspondence to give a restriction of hypergraphs with interfaces (called \textit{monogamous} directed acyclic),  leading to a correspondence between symmetric monoidal string diagrams (\textit{i.e.}, without requiring a Frobenius algebra) and the category of (cospans of) such hypergraphs.

In ~\cite{bonchi_string_2022-1}, hypergraphs \textit{with interfaces} and their DPO rewriting are shown to be sound and complete for categorical presentations of SMTs, including Frobenius algebras.
Hypergraphs extend graphs by allowing edges with multiple sources and targets. 
Interpreting string diagrams as hypergraphs maps morphisms to edges and wires to vertices, while the dangling wires at the diagram's top and bottom give rise to \textit{input} and \textit{output} interfaces for a hypergraph.
Subsequent work \cite{bonchi_string_2022-2} refines this to \textit{monogamous} directed acyclic hypergraphs, establishing a correspondence between symmetric monoidal string diagrams (without Frobenius algebra) and the category of (cospans of) such hypergraphs.

\begin{definition}[Category of hypergraphs]\label{def:hypergraph}
A \emph{hypergraph $\mathcal{G}$ over a signature $\Sigma$} is a tuple $(V,E,s,t,l)$,  where $V$ is a finite set of vertices, $E$ is a finite set of edges, $s : E \to V^{*}$ is a source function, $t : E \to V^{*}$ is a target function,  and $l : E \to \Sigma$ is a labelling function that assigns each edge a generator from monoidal signature $\Sigma$.  The labelling function must respect the typing of the generator: for all edges $e$,  $|s(e)| = m$ and $|t(e)| = n$,  where $l(e) : m \to n$.  We call a hypergraph \textit{discrete} if its set of edges is empty.   
A \emph{hypergraph homomorphism} $\phi: \mathcal{F} \to \mathcal{G}$ is given by a pair of functions $\phi_V : V_{\mathcal{F}} \to V_{\mathcal{G}}, \phi_E : E_{\mathcal{F}} \to E_{\mathcal{G}}$ such that the following hold. 
\begin{enumerate}
    \vspace{-1.25mm}
    \item $\phi_V^*(s_{\mathcal{F}}(e)) = s_{\mathcal{G}}(\phi_E(e))$
    \item $\phi_V^*(t_{\mathcal{F}}(e)) = t_{\mathcal{G}}(\phi_E(e))$
    \item $l_{\mathcal{F}}(e) = l_{\mathcal{G}}(\phi_E(e))$
    \vspace{-0.5mm}
\end{enumerate}
% \[
% \begin{array}{ll}
%     \phi_V^*(s_{\mathcal{F}}(e)) = s_{\mathcal{G}}(\phi_E(e)) & \phi_V^*(t_{\mathcal{F}}(e)) = t_{\mathcal{G}}(\phi_E(e))\\
%     l_{\mathcal{F}}(e) = l_{\mathcal{G}}(\phi_E(e)) & \;
% \end{array}
% \]
We denote by $\catname{Hyp(\Sigma)}$ the category of hypergraphs over $\Sigma$ and hypergraph homomorphisms. 
\end{definition}
Note that $\catname{Hyp(\Sigma)}$ has all finite colimits and, in particular,  the coproduct is given by the disjoint union of hypergraphs.
The initial object is given by the empty hypergraph.  

\begin{definition}[Symmetric monoidal category of cospans]
Let $\mathbb{C}$ be a category with all finite colimits.  A cospan from $X$ to $Y$ is a pair of arrows $X \xrightarrow{} A \xleftarrow{} Y$  in $\mathbb{C}$.  This \textit{category of cospans of $\mathbb{C}$},  denoted $\catname{Csp(\mathbb{C})}$,  has as objects those of $\mathbb{C}$ and as arrows isomorphism classes of cospans.
%Two cospans $X \xrightarrow{} A \xleftarrow{} Y$ and $X \xrightarrow{} B \xleftarrow{} Y$ are isomorphic if the following diagram commutes and $\alpha$ is iso.
%\[
%\begin{tikzcd}
%                                               & A \arrow[dd, "\alpha"] &                                                \\
%X \arrow[ru, bend left] \arrow[rd, bend right] &                        & Y \arrow[ld, bend left] \arrow[lu, bend right] \\
%                                               & B                      &                                               
%\end{tikzcd}
%\]
%For $X \in \mathbb{C}$ an $id$-cospan is $X \xrightarrow{id_X} X \xleftarrow{id_X} X$.
Composition of cospans $X \xrightarrow{f} A \xleftarrow{g} Y$ and $Y \xrightarrow{h} B \xleftarrow{i} Z$ is given by pushout: $X \xrightarrow{f} A +_{g,h} B \xleftarrow{i} Z$,  with identity given by the cospan of identities.  
Its monoidal product is given by the coproduct of $\mathbb{C}$ as follows: 
\ifdefined \ONECOLUMN
\[
    (A \xrightarrow{f} \mathcal{G} \xleftarrow{g} B) \otimes (A' \xrightarrow{f'} \mathcal{G'} \xleftarrow{g'} B') \coloneq A + A' \xrightarrow{f+f'} \mathcal{G} + \mathcal{G'} \xleftarrow{g+g'} B + B'  ,
\]
\else
\vspace{-2mm}
\begin{multline*}
    (A \xrightarrow{f} \mathcal{G} \xleftarrow{g} B) \otimes (A' \xrightarrow{f'} \mathcal{G'} \xleftarrow{g'} B') \coloneq\\ A + A' \xrightarrow{f+f'} \mathcal{G} + \mathcal{G'} \xleftarrow{g+g'} B + B'  , 
    \vspace{-1.25cm}
\end{multline*}
\fi
and its unit by the initial object of $\mathbb{C}$.
Symmetry is inherited from $\mathbb{C}$ as the coproduct~\cite{MonoidalCoproduct}.
\end{definition}

Then, morphisms of $\catname{S}(\Sigma)$ can be interpreted as particular cospans of discrete hypergraphs,  which we call hypergraphs with interfaces~\cite{bonchi_string_2022-2}.
\begin{definition}[Hypergraphs with interfaces]
\label{def:cspd}
Define the category of \emph{hypergraphs with interfaces} $\catname{HypI(\Sigma)}$ as the full sub-category of $\catname{Csp(Hyp(\Sigma))}$ with discrete \textit{ordered} hypergraphs as objects.
\end{definition}

In particular, this means that any morphism in $\HypI{\Sigma}$ is of the form $n \xrightarrow{f} \mathcal{G} \xleftarrow{g} m$,
where $n,\;m$ are discrete ordered hypergraphs and $\mathcal{G}$ is an arbitrary hypergraph.  The morphisms $f$ and $g$ into $\mathcal{G}$ will specify which vertices are to be considered as inputs and outputs,  respectively.
The ordering of the interface associated to a hypergraph is used to differentiate between the interpretations of morphisms $f \otimes g$ and $g \otimes f$.

\begin{definition}[Degree,  path,  acyclic] 
The \emph{in-degree} of a vertex $v$ in a hypergraph $\mathcal{G}$ is the number of hyperedges $e \in {E_\mathcal{G}}$ such that $v \in t(e)$.  
Similarly, the out-degree of $v$ is the number of edges $e \in E_\mathcal{G}$ such that $v \in s(e)$.
A \emph{path} from $e_0$ to $e_{n-1}$ of length $n$ in a hypergraph $\mathcal{G}$ is an ordered sequence of edges $[e_0, e_1, \ldots, e_{n-1}] \in E_\mathcal{G}^*$ such that for all $i < n - 1$ some vertex $v \in t(e_i)$ is also in $s(e_{i+1})$.
A hypergraph is \emph{directed acyclic} if it contains no paths of length $n > 0$ that start and end in the same edge.
\end{definition}

In order to make the interpretation complete with respect to $\catname{S}(\Sigma)$,  it is necessary to restrict our attention to the following subset of hypergraphs with interfaces. 
\begin{definition}[Category of MDA Hypergraphs with Interfaces~\cite{bonchi_string_2022-2}]
\label{def:monogamy_hyp}
We call a cospan $n \xrightarrow{f} \mathcal{G} \xleftarrow{g} m$ in $\HypI{\Sigma}$ \emph{monogamous directed acyclic (MDA)} if:
\begin{enumerate}
    \item $\mathcal{G}$ is a directed acyclic hypergraph;
    \item $f$ and $g$ are monomorphisms;
    \item the in-degree and out-degree of every vertex is at most $1$;
    \item vertices with in-degree $0$ are precisely the image of $f$; and
    \item vertices with out-degree $0$ are precisely the image of $g$.
\end{enumerate}
We denote by $\MdaCospans$ the subcategory of $\HypI{\Sigma}$ of MDA cospans.
\end{definition}
Note that $\MdaCospans$ remains a symmetric monoidal category with the same equipment as before.  We now recall from
\cite{bonchi_string_2022-2} the following theorem,  which says the category $\MdaCospans$ is in fact \textit{equivalent} to $\catname{S}(\Sigma)$. 
\begin{theorem}[Corollary 26 of~\cite{bonchi_string_2022-2}]\label{thm:prop-equiv}
    $\catname{S}(\Sigma) \cong \MdaCospans$
\end{theorem}

\subsection{DPOI-Rewriting for Hypergraphs with Interfaces}
% The correspondence of the previous section means that the symmetric monoidal equations of $\catname{S}(\Sigma, \mathcal{E})$ are absorbed in the hypergraph representation.  The equations $\mathcal{E}$ are then to be considered as rewrite rules on string diagrams,  which are then modelled as DPO-rewrites on hypergraphs with interfaces. 
% DPO rewriting is a standard technique for formalising graph rewriting categorically. 
% In this section we recall standard results of double-pushout (DPO) rewriting of hypergraphs,  as well as the specific notion of DPO rewriting which is sound and complete for arbitrary SMTs,  which is called \textit{convex} rewriting. 
The correspondence from the previous section absorbs the symmetric monoidal equations of $\catname{S}(\Sigma, \mathcal{E})$ into the hypergraph representation. 
The equations $\mathcal{E}$ become rewrite rules on string diagrams, modelled as DPO rewrites on hypergraphs with interfaces.
This section reviews standard DPO rewriting for hypergraphs and the \textit{convex} rewriting, which is sound and complete for arbitrary SMTs.

DPO rewriting is a standard categorical approach to graph transformations~\cite{dpo}.
The intuitive notion of (hyper-)graph rewriting is  as expected: given a rewrite rule $\mathcal L \rightsquigarrow \mathcal R$ and a graph $\mathcal{G}$,  we expect to find an occurrence of $\mathcal L$ within $\mathcal{G}$,  remove it,  and then plug $\mathcal R$ into the hole to recover the rewritten graph $\mathcal{H}$.  In DPO rewriting,  such a rewrite rule is formalised as a span $\mathcal L \xleftarrow{} \mathcal K \xrightarrow{} \mathcal R$ in some appropriate category of graphs,  where $\mathcal K$ is some invariant that should be maintained during rewriting.
A typical DPO square is depicted in Figure~\ref{fig:dpo_dpoi}~(left).

% \[
% \begin{tikzcd}
% \mathcal L \arrow[d, "m"'] & \mathcal K \arrow[d] \arrow[l] \arrow[r] & \mathcal R \arrow[d]   \\
% \mathcal G^{\urcorner}     & \mathcal C \arrow[l] \arrow[r] & ^{\ulcorner} \mathcal H
% \end{tikzcd}
% \]
% CHRIS: WHY IS J DISPLAYING AND NOT R? 
% \[
%     \scalebox{0.8}{
%         \tikzfig{combinatorial_semantics/DPO_square}
%     }
% \]

% The morphism $\mathcal L \to \mathcal G$ is called a \textit{match},  defined whenever there is a corresponding \textit{pushout complement}---the pair of morphisms $\mathcal K \to \mathcal \mathcal{L}^{\bot} \to \mathcal G$ completing the top half of the pushout square defined by $\mathcal K \xrightarrow{} \mathcal L \xrightarrow{} \mathcal G$. 
% When the morphisms involved in the rewrite rule are monomorphisms,  as is usually required,  constructing $\mathcal{L}^{\bot}$ can be intuitively understood as finding an occurrence of $\mathcal L$ in $\mathcal G$ and removing from this occurrence everything that has no pre-image in $\mathcal K$.  
% The final graph $\mathcal H$ is then recovered by inserting into the hole everything from $\mathcal R$ that has no pre-image in $\mathcal K$ --- hence the two pushouts: one for deletion and one for insertion.  
% In this case,  we say that $\mathcal G$ is rewritten into $\mathcal H$ via the rewrite rule $\mathcal L \xleftarrow{} \mathcal K \xrightarrow{} \mathcal R$.

The morphism $\mathcal L \to \mathcal G$ is called a \textit{match}, defined when a corresponding \textit{pushout complement} exists, comprising morphisms $\mathcal K \to \mathcal{L}^\bot \to \mathcal G$ completing the top half of the pushout square $\mathcal K \to \mathcal L \to \mathcal G$. 
For monomorphic rewrite rules, constructing $\mathcal{L}^\bot$ involves identifying $\mathcal L$ in $\mathcal G$ and removing elements with no pre-image in $\mathcal K$. 
The final graph $\mathcal H$ is then obtained by inserting into the resulting ``hole'' the elements of $\mathcal R$ without pre-image in $\mathcal K$, resulting in two pushouts: one for deletion and one for insertion. 
Thus, we say that $\mathcal G$ is rewritten into $\mathcal H$ via $\mathcal L \xleftarrow{} \mathcal K \xrightarrow{} \mathcal R$.

To express the rewriting of (hyper-)graphs with interfaces,  \textit{i.e.},  the ones that represent morphisms in $\HypI{\Sigma}$, the DPO formalism has been extended to double-pushout rewriting with \textit{interfaces}  (DPOI) as shown in Figure~\ref{fig:dpo_dpoi}~(right).
In the case of $ \catname{Hyp}(\Sigma)$,  we wish rewrite rules to be pairs of discrete cospans of hypergraphs with matching interfaces: $n \xrightarrow{} \mathcal L \xleftarrow{} m$, $n \xrightarrow{} \mathcal R \xleftarrow{} m$.  
By noting that a cospan $0 \xrightarrow{} \mathcal L \xleftarrow{[f,g]} n + m$ encodes the same data as a cospan $n \xrightarrow{f} \mathcal L \xleftarrow{g} m$ (Remark 3.14~\cite{bonchi_string_2022-1}) we can formulate a rewrite rule as a single span $\mathcal L \xleftarrow{} n+m \xrightarrow{} \mathcal R$,  and we can do similarly for the discrete cospans giving the interface to $\mathcal G$ and $\mathcal H$.  
We now make this identification throughout without further mention.  
Thus, we reach the following variant of DPO rewriting.
\begin{definition}[DPOI rewriting]\label{def:dpoi}
Given a span of morphisms $\mathcal G \xleftarrow{} n+m \xrightarrow{} \mathcal H$ in $\textbf{Hyp}(\Sigma)$,  we say \textit{$\mathcal G$ rewrites to $\mathcal H$ with interface $n+m$ via rewrite rule $\mathcal L \xleftarrow{} i+j \xrightarrow{} \mathcal R$} if there exists an object $\mathcal{L}^{\bot}$ and morphisms which complete the right commutative diagram in Figure~\ref{fig:dpo_dpoi} such that the two marked squares are pushouts.
% \[
% \begin{tikzcd}
% \mathcal L \arrow[d, "m"'] & i+j \arrow[d] \arrow[l] \arrow[r] & \mathcal R \arrow[d]   \\
% \mathcal{G}^{\urcorner}     & \mathcal C \arrow[l] \arrow[r]                      & ^{\ulcorner}\mathcal H \\
%                   & n+m \arrow[lu] \arrow[u] \arrow[ru]          &              
% \end{tikzcd}
% \]
\begin{figure}[t!]
    \begin{subfigure}[T]{0.35\columnwidth}
    \[
                % \tikzfig{../figures/combinatorial_semantics/DPO_square}
                \begin{tikzcd}
                    {\mathcal{L}} & {\mathcal{K}} & {\mathcal{R}} \\
                    {\mathcal{G}} & {\mathcal{L}^{\bot}} & {\mathcal{H}}
                    \arrow["f"', from=1-1, to=2-1]
                    \arrow[from=1-2, to=1-1]
                    \arrow[from=1-2, to=1-3]
                    \arrow[from=1-2, to=2-2]
                    \arrow[from=1-3, to=2-3]
                    \arrow["{\lrcorner}"{rotate=90, xshift=-0.25em, yshift=0.9em,font = \Large}, draw=none, from=2-1, to=1-2]
                    \arrow[from=2-2, to=2-1]
                    \arrow[from=2-2, to=2-3]
                    \arrow["{\lrcorner}"{rotate=180,xshift=-0.5em, yshift=2.35em,font = \Large}, draw=none, from=2-3, to=1-2]
\end{tikzcd}
    \]
    % \subcaption{DPO square}
     % \label{fig:dpo}
    \end{subfigure}
    % \hfill
    \hspace{0.15\columnwidth}
    \begin{subfigure}[T]{0.35\columnwidth}
        \[
    % \scalebox{0.8}{
        % \tikzfig{../figures/combinatorial_semantics/DPOI_square_Hyp}
    % }
    \begin{tikzcd}
        {\mathcal{L}} & {i+j} & {\mathcal{R}} \\
        {\mathcal{G}} & {\mathcal{L}^{\bot}} & {\mathcal{H}} \\
        & {n+m}
        \arrow["f"', from=1-1, to=2-1]
        \arrow[from=1-2, to=1-1]
        \arrow[from=1-2, to=1-3]
        \arrow[from=1-2, to=2-2]
        \arrow[from=1-3, to=2-3]
        \arrow["{\lrcorner}"{rotate=90,xshift=-0.25em, yshift=1em,font = \Large}, draw=none, from=2-1, to=1-2]
        \arrow[from=2-2, to=2-1]
        \arrow[from=2-2, to=2-3]
        \arrow["{\lrcorner}"{rotate=180,xshift=-0.9em, yshift=2.25em,font = \Large}, draw=none, from=2-3, to=1-2]
        \arrow[from=3-2, to=2-1]
        \arrow[from=3-2, to=2-2]
        \arrow[from=3-2, to=2-3]
    \end{tikzcd}
\]
    % \subcaption{DPOI square}
    % \label{fig:dpoi}
    \end{subfigure}
    \caption{DPO and DPOI squares}
    \label{fig:dpo_dpoi}
\end{figure}

\end{definition}
\vspace{-1ex}
We need to restrict the notions of pushout complement and of match in order to preserve monogamy and directed acyclicity.  
\vspace{-1ex}
\begin{definition}[Boundary complement]
\label{def:boundary_original}
For MDA cospans $i \xrightarrow{a_1} L \xleftarrow{a_2} j$ and $n \xrightarrow{b_1} \mathcal G \xleftarrow{b_2} m$
and mono $f : \mathcal L \to \mathcal G$, a pushout complement $i + j \to \mathcal{L}^\bot \to \mathcal G$ as depicted in the square below
% \begin{figure}[h!]
%     \centering
%     \includegraphics[width=0.4\linewidth]{figures/combinatorial_semantics/boundary_complement.png}
% \end{figure}
% https://q.uiver.app/#q=WzAsNSxbMCwwLCJMIl0sWzIsMCwiaStqIl0sWzAsMiwiR157XFx1cmNvcm5lcn0iXSxbMiw0LCJuK20iXSxbMiwyLCJMXntcXGJvdH0iXSxbMyw0LCJbZF8xLGRfMl0iLDJdLFszLDIsIltiXzEsYl8yXSJdLFswLDIsImYiLDJdLFsxLDAsImEgPSBbYV8xLGFfMl0iLDJdLFsxLDQsImMgPSBbY18xLGNfMl0iXSxbNCwyLCJnIl1d
% \[\begin{tikzcd}
% 	\mathcal L && {i+j} \\
% 	\\
% 	{\mathcal G^{\urcorner}} && {\mathcal{L}^{\bot}} \\
% 	\\
% 	&& {n+m}
% 	\arrow["{[d_1,d_2]}"', from=5-3, to=3-3]
% 	\arrow["{[b_1,b_2]}", from=5-3, to=3-1]
% 	\arrow["f"', from=1-1, to=3-1]
% 	\arrow["{[a_1,a_2]}"', from=1-3, to=1-1]
% 	\arrow["{[c_1,c_2]}", from=1-3, to=3-3]
% 	\arrow[from=3-3, to=3-1]
% \end{tikzcd}\]
\[
    \scalebox{0.8}{
        \input{./figures/DPOI_boundary_complement.tikz}
    }
\]
is called a boundary complement if $[c_1, c_2]$ is mono and there exist $d_1 : n \to \mathcal{L}^{\bot}$ and $d_2 : m \to \mathcal{L}^{\bot}$ making the above triangle commute and such that
\[
    n + j \xrightarrow{[d_1,c_2]} \mathcal{L}^{\bot} \xleftarrow{[d_2,c_1]} m + i
\]
is an MDA cospan. 
\end{definition}
%\begin{remark}
%    Note how the requirement for $i \xrightarrow{} L \xleftarrow{} j$ and $n \xrightarrow{} G \xleftarrow{} m$ being monogamous, i.e. morphisms in $\MdaCospans$, is a part of the above definition.
%\end{remark}
Intuitively,  requiring boundary complements ensures that inputs are glued exclusively to outputs and vice-versa in the two pushout squares.
Yet, this restriction is not enough to make DPOI rewriting complete for $\catname{S}(\Sigma,\mathcal{E})$: there might be cospans with carriers $\mathcal{G}$ and $\mathcal{H}$ such that one can be rewritten into another but their corresponding morphisms are not equal modulo $\mathcal{E}$~\cite{bonchi_string_2022-2}.
To make the DPOI rewriting complete, we need to add a restriction on the image of the match.

\begin{definition}[Convex match]
We call a subgraph $\mathcal{H}$ of a hypergraph $\mathcal{G}$ a \emph{convex subgraph} if for all $v_i, v_j \in V_{\mathcal{H}}$ every path from $v_i$ to $v_j$ is also in $\mathcal{H}$.
We call a match $f : \mathcal L \to \mathcal G$ a \emph{convex match} if it is mono and its image is a convex subgraph of $\mathcal G$. 
   
\end{definition}

\begin{definition}[Convex DPOI rewriting]
\label{def:convex_dpo}
Let $\mathfrak{R}$ be a set of DPOI rewrite rules. 
Then, given $\mathcal G \xleftarrow{} n+m$ and $\mathcal H \xleftarrow{} n + m$ in $\catname{Hyp(\Sigma)}$, $\mathcal G$ rewrites convexly into $\mathcal H$ with interface $n + m$ --- notation $(\mathcal G \xleftarrow{} n + m ) \Rrightarrow_{\mathfrak{R}}  (\mathcal H \xleftarrow{} n + m )$ --- if there exist rule $\mathcal L \xleftarrow{} i + j \xrightarrow{} \mathcal R$ in $\mathfrak{R}$ and object $\mathcal{L}^{\bot}$ and cospan arrows $i+j \xrightarrow{} \mathcal{L}^{\bot} \xleftarrow{} n+m$ such that the DPOI diagram of Definition \ref{def:dpoi} commutes and its marked squares are pushouts 
and the following conditions hold
\begin{itemize}
    \item $f : L \to \mathcal G$ is a convex match;
    \item $i + j \to \mathcal{L}^{\bot} \to \mathcal G$ is a boundary complement in the leftmost pushout.
\end{itemize}
\end{definition}
%By requiring convex matching and boundary complements we ensure that the original cospans above are monogamous, e.g., $n \xrightarrow{} G \xleftarrow{} m$.
It is then proved in~\cite{bonchi_string_2022-2} that this notion of DPOI rewriting is sound and complete for $\catname{S}(\Sigma, \mathcal{E})$: for any morphisms $f$ and $g$ in $\catname{S}(\Sigma)$,  we have that $f = g$ modulo $\mathcal{E}$ if and only if there exist a sequence of DPO rewrites (induced by $\mathcal{E}$) between the representations of $f$ and $g$ as (cospans of) hypergraphs.
As SMC equations get absorbed by the representation, it makes this approach computationally appealing as they do not need to be represented explicitly.
% In particular,  Theorem \ref{thm:prop-equiv} means that the structural SMC equations are factored out in the representation of $f$ and $g$,  and the DPO rewrites necessary are only those induced by $\mathcal{E}$. 
%\begin{theorem}[Theorem 35~\cite{bonchi_string_2022-2}]
%For any $f,g \in \mathbf{SMT}(\Sigma, \mathcal{E})$,  we have that 
%\[
%	f = g \iff \llangle \lceil f \rceil \rrangle \Rrightarrow^*_{\mathfrak{R}} \llangle \lceil g \rceil \rrangle~, 
%\]
%where $\mathfrak{R} = \llangle \lceil \mathcal{E} \rceil \rrangle$ and $ \Rrightarrow^*_{\mathfrak{R}}$ denotes the existence of a sequence of convex rewrites. 
%\end{theorem}

\section{E-Hypergraphs}\label{sec:e-hypergraphs}

Hierarchical (hyper-)graphs are a well-studied area of research \cite{plump:hierarchical-graphs, montanari:gs-lambda, palacz:hierarchical-transform, Gaducci:hierarchical-graphs, Ghica:hierarchical}.  
In this section,  we follow analogous steps to those used to achieve a combinatorial representation of $\catname{S}(\Sigma)$ by $\mathbf{MHypI}(\Sigma)$,  but in the semilattice-enriched setting.  
First we define an \textit{e-hypergraph},  which supports a hierarchical notion of ``e-box'' with distinct components. 
Then,  we extend this notion with interfaces via a generalisation of the cospan construction.  
Finally,  we restrict the morphisms under consideration to those satisfying the MDA condition,  as in the previous section.  
% An example diagram representing e-hypergraphs with interfaces can be found in Figure \ref{fig:A+B}.

% Note that here,  and in the rest of the paper,  we will restrict to dealing with signatures $\Sigma$  such that for all $c_1,c_2 \in \Sigma$, if arity (respectively, co-arity) of $c_1$ is 0, then co-arity (respectively, arity) of $c_2$ must be at least 1.
% This is to ensure we do not have terms of type $0 \to 0$.
% This is a restriction we make in order to avoid complications to the technical development of DPOI rewriting for e-hypergraphs given in the subsequent section,  and we explain it further there.
\begin{remark}
Here,  and in the rest of the paper,  we will restrict to dealing with signatures $\Sigma$ for which the corresponding family of $\Sigma$-terms does not contain terms of type $0 \to 0$ (apart from the monoidal unit).
This is a restriction we make in order to avoid complications to the technical development of DPOI rewriting for e-hypergraphs given in the subsequent section.
A discussion of a general case can be found in section~\ref{sec:dpo-fix} of Appendix.
\end{remark}

In the following,  we elide the obvious injections of $V$ and $E$ into $V+E$.

\begin{definition}[E-hypergraphs]
\label{def:e-homo}    
An \emph{e-hypergraph $\mathcal{G}$ over a signature $\Sigma$} is a tuple $(V,E,s,t,l,<,\consistency)$,  where $(V,E,s,t)$ is an unlabelled hypergraph;  $l :E \to \Sigma + 1$ is a labelling function modified to include an extra value $\bot$ in its codomain; $<$ is a strict partial order on $V + E$ called the \textit{child relation}, for which we introduce the following additional notation.
For $x \in V + E$,  the set of \emph{parents} (or, \emph{predecessors}) of $x$ is denoted as $[x) = \{x' ~|~ x' < x\; \}$.
% \[
% 	[x) = \{x' ~|~ x' <_c x\; \} \qquad (x] = \{x' ~|~  x <_{c} x'\}
% \]
We will denote the immediate parent (predecessor) $e$ of $x$ by writing $e <^{\mu} x$, and likewise $x$ is then called an immediate \emph{child} (or, \emph{successor}).
We call edges $e$ such that $l(e) = \bot$ \textit{hierarchical},  and edges $e$ and vertices $v$ such that $[e) = \varnothing$ and $[v) = \varnothing$ \emph{top-level}.
An edge $e$ is called maximal if $\{x' | e < x'\} = \varnothing$.
We require that the child relation satisfies the following conditions:
\begin{enumerate}
	\item each parent set contains exclusively hierarchical edges;
	\item each $x$ has at most one immediate parent;
	\item for all $e$ such that $e$ is maximal, $l(e) \not = \bot$;
	\item if $v \in s(e)$ then $e' <^{\mu} e$ iff $e' <^{\mu} v$ and similarly if $v \in t(e)$.
\end{enumerate}
Finally, $\consistency$ is a \textit{consistency relation} which is given by the union of a family of equivalence relations $\consistency_p$ on each set $\{x \in V + E ~|~ p <^\mu x\}$ of elements which share the same parent where each relation is also closed under connectivity, \textit{i.e.}, if $v \in s(e)$ or $v \in t(e)$ such that $p <^{\mu}(v)$ and $p <^{\mu}(e)$ then $v \consistency_{p} e$.
We require that $\consistency_{p} \not = (E_{p} + V_{p}) \times (E_{p} + V_{p})$ where $V_{p} = \{ v ~ | ~ p <^{\mu} v\}$ and $E_{p} = \{ e ~ | ~ p <^{\mu} e\}$.
Given an e-hypergraph $\mathcal{G}$ we can consider a corresponding \textit{underlying} hypergraph $\mathcal{G}'$ by forgetting $<$ and $\consistency$.
% We collect the union of all these equivalence relations into one relation $\consistency = \bigcup_{p} \consistency_p$,  which is defined on a subset of $V+E$. 
% Note that this relation is symmetric and transitive if considered on $V + E$.
% The consistency relation must additionally satisfy that if $v \in s(e)$ or $v \in t(e)$ and there exists $e'$ such that $e' <^{\mu} v$ and $e' <^{\mu} e$ then $v \consistency e$.
\label{def:consistency_properties}
\end{definition}
Intuitively,  the child relation equips a hypergraph with a hierarchical structure: certain edges,  labelled with $\bot$,  are meant to interpret the dashed e-boxes of string diagrams.  
These edges can be designated as ``parents" to sub-e-hypergraphs,  which interpret the various components of the e-boxes.  
These components are distinguished by the consistency relation, which partitions the maximal sub-e-hypergraph below the hierarchical edge.
% Technically,  condition $(4)$ on the child relations says that $<^{\mu}$ respects connectivity of edges and vertices,  and similarly for the consistency relation. 
An example of an e-hypergraph can be seen below in Figure~\ref{fig:e-cospan-example} (within a grey shaded region).
The e-hypergraph is given by $E = \{e_1, e_2\}, V = \{v_1, \ldots, w_4 \}$, $s$,$t$, $l = \{e_1 \mapsto \bot, e_2 \mapsto f\}$ and
\[
\begin{array}{ccc}
    e_1 <^{\mu} v_3 & v_3 \consistency v_4 & v_3 \not \consistency v_5\\
    \ldots & \ldots & \ldots\\
    e_1 <^{\mu} w_4 & w_3 \consistency w_4 & v_4 \not \consistency w_4
\end{array}
\]
A hierarchical edge is depicted with a dashed box.
The partition induced by $\consistency_{e_1}$ is depicted with a dashed vertical line.
The condition $\consistency_{p} \not = (E_{p} + V_{p}) \times (E_{p} + V_{p})$, in particular, does not allow hierarchical edges with no delimiting lines.
Finally, $e_1$ is the maximal edge and hence labelled.
\begin{definition}[E-hypergraph homomorphism]
An \textit{e-hypergraph homomorphism} $\phi: \mathcal{F} \to \mathcal{G}$ between e-hypergraphs $\mathcal{F},\mathcal{G}$ is a pair of functions $\phi_V : V_{\mathcal{F}} \to V_{\mathcal{G}}, \phi_E : E_{\mathcal{F}} \to E_{\mathcal{G}}$ such that
\begin{enumerate}
    \item $\phi$ is a hypergraph homomorphism
    \item for $v \in V_{\mathcal{F}}$ and $e \in E_{\mathcal{F}}$
	    \[
        \text{if } e <_{\mathcal{F}}^{\mu} v \text{ then } \phi(e) <_{\mathcal{G}}^{\mu} \phi(v)
        \]
        and for $e_1, e_2 \in E_{\mathcal{F}}$
        \[
        \text{if } e_1 <_{\mathcal{F}}^{\mu} e_2 \text{ then } \phi(e_1) <_{\mathcal{G}}^{\mu} \phi(e_2)
        \]
        \item for all $x_1, x_2 \in \mathcal{F}$,  $x_1 ~\consistency_{\mathcal{F}}~ x_2$ implies $\phi(x_1) ~\consistency_{\mathcal{G}}~ \phi(x_2)$. 
\end{enumerate}
% If we further have that if $[x_1) = \varnothing$ then $[\phi(x_1)) = \varnothing$ we call the homomorphism \textit{strict}.
\end{definition}
The conditions on e-hypergraph homomorphisms require to preserve immediate parents and the consistency relation.  
\begin{definition}[Category of e-hypergraphs]
The \emph{category of e-hypergraphs},  denoted  $\catname{EHyp(\Sigma)}$,  has e-hypergraphs as objects and e-hypergraph homomorphisms as morphisms.  
\end{definition}

The category of e-hypergraphs has coproducts,  given by the disjoint union of e-hypergraphs,  and an initial object given by the empty e-hypergraph.  
A concrete description of the pushout of two morphisms in this category is given in Appendix~\ref{sec:appendix:pushout}.
Generally,  the pushout of two e-hypergraph homomorphisms need not exist,  but we prove that it does in all scenarios that we are interested in; most importantly, the pushout for the composition of two cospans with discrete feet always exists.

% This remark is confusing
% \begin{remark}
% When we wish to consider $\phi(\mathcal{F}) \subseteq \mathcal{G}$ as an e-hypergraph in its own right,  we let the immediate parent to be undefined for all $x$ whose immediate parent in $\mathcal{G}$ is outside $\phi(\mathcal{F})$.
% \end{remark}

When modelling enriched string diagrams using e-hypergraphs,  we must keep track of more than simply the ultimate external inputs and outputs of the diagram: each internal e-box has its own inputs and outputs,  which also must be tracked.  
Thus,  we introduce an \textit{extended} cospan construction.  

% To save space
% Since the $\catname{EHyp(\Sigma)}$ does not have pushouts in general,  we skip straight to considering a category of extended cospans of discrete e-hypergraphs (\textit{i.e.},  discrete hypergraphs considered as e-hypergraphs in the obvious way),  for which pushouts do exist.

We use the notation $n \setminus m$ to denote the discrete e-hypergraph with $n - m$ vertices,  in particular with the vertices of $m$ removed from $n$ when vertices of $m$ is a sub-e-hypergraph of $n$.
\begin{definition}[Category of e-hypergraphs with interfaces]
    The category of \textit{e-hypergraphs with extended interfaces} $\Ecospans$ has discrete \textit{ordered} e-hypergraphs as objects,  with hom-sets $\Ecospans(n,m)$ consisting of isomorphism classes of \textit{extended cospans}, defined as follows:  
    \[
    n \xrightarrow{f_{ext}} n' \xrightarrow{f_{int}} \mathcal{G} \xleftarrow{g_{int}} m' \xleftarrow{g_{ext}} m
    \]
    where $\mathcal{G}$ is an e-hypergraph,  and $n, n', m, m'$ are discrete ordered e-hypergraphs,  $f_{ext},g_{ext}$ are monomorphisms in $\catname{EHyp(\Sigma)}$,  and the image of $f_{ext};f_{int}$ and of $g_{ext};g_{int}$ consist exclusively of top-level vertices,  and such that vertices in the strictly internal interface (defined below) are not top-level.  
    We will sometimes write $\mathcal G$ to denote the extended cospan,  where it is clear from context $\mathcal G$ is equipped with extended interfaces. 

\end{definition}
We will call $n$  \textit{external input interfaces} and $n'$ \textit{internal input interfaces}.
We call $n' \setminus f_{ext}(n)$  the \textit{strictly internal input interfaces}.  
We do analogously for the \emph{output interfaces},  with respect to $m$,  $m'$ and $m' \setminus g_{ext}(m)$.  
We occasionally conflate $f_{ext}$ with $f_{ext};f_{int}$ when it is clear from context,  and also conflate $n$ and $m$ with their images in $n'$ and $m'$,  and likewise $m, m, n', m'$ with their images in $\mathcal{G}$. 
Given an edge $e \in E_\mathcal{G}$ such that $l(e) = \bot$,  we call the \textit{inputs of $e$} the intersection of the strictly internal input interface of $\mathcal{G}$ with the immediate successors of $e$,  and analogously for the \emph{outputs of $e$}.  

% Consider a relation $R = \{x R y \text{ if } f_{int}(x) \consistency f_{int}(y) \text{ for } x,y \in n'\}$ and let $S$ be its reflexive closure.
% The latter partitions $n'$ into non-empty subsets $\{p_{j}\}_{j=1}^{k}$. 
% We get an analogous partition for $m'$.

\begin{figure}
    \[
    \scalebox{0.5}{
        \input{./figures/extended_cospan_example.tikz}
    }    
    \]
    \captionsetup{belowskip=-3ex, skip=0pt}
    \caption{Morphism of $\Ecospans$}
    \label{fig:e-cospan-example}
\end{figure}

\begin{definition}
\label{def:iso}

Consider two extended cospans below 
\[
\scalebox{0.75}{
    \input{./figures/isomorphic_e_cospans.tikz}
}
\]
and consider a relation 
\[
R = \{x R y \text{ if } f_{int}(x) \consistency f_{int}(y) \text{ for } x,y \in n'\}
\]
and let $S$ be its reflexive closure.
The latter partitions $n'$ into non-empty subsets $\{p_{j}\}_{j=1}^{k}$. 
We get an analogous partition for $m'$, $n''$, $m''$.
Two extended cospans are isomorphic if there exist isomorphisms $\alpha$, $\beta$ and $\gamma$ making the above diagram commute and such that $\alpha$ and $\gamma$ preserve order within each $p_i$.
\end{definition}
An example of an extended cospan can be seen in Figure~\ref{fig:e-cospan-example}.
The partition of the interfaces is given with colours.
The labels suggest how $f_{int}, f_{ext}$ (respectively, $g_{int}, g_{ext}$) act on vertices, with the ordering of vertices in the interfaces given left to right.
Intuitively, such notion of isomorphism suggests that the position of the image of $n$ within $n'$ relative to the strict internal interface does not matter, for example, vertices
\begin{tikzpicture}
	\begin{pgfonlayer}{nodelayer}
		\node [style=red node, label={above:$v_{3}$}] (0) at (-0.5, 3.5) {};
		\node [style=red node, label={above:$v_4$}] (1) at (0, 3.5) {};
	\end{pgfonlayer}
\end{tikzpicture}
can be swapped with vertices
\begin{tikzpicture}
	\begin{pgfonlayer}{nodelayer}
		\node [style=blue node, label={above:$v_{5}$}] (0) at (-0.5, 3.5) {};
		\node [style=blue node, label={above:$v_6$}] (1) at (0, 3.5) {};
	\end{pgfonlayer}
\end{tikzpicture}
in the input interface yielding an isomorphic cospan.
However, we may not swap
\begin{tikzpicture}
	\begin{pgfonlayer}{nodelayer}
		\node [style=red node, label={above:$v_{3}$}] (0) at (-0.5, 3.5) {};
		\node [style=red node, label={above:$v_4$}] (1) at (0, 3.5) {};
	\end{pgfonlayer}
\end{tikzpicture}
with each other --- this would yield a non-isomorphic cospan.
Further examples of isomorphic and non-isomorphic extended cospans can be found in Appendix~\ref{sec:appendix:iso}.

% \begin{remark}\label{remark:embedding_functor}
% There is also a faithful identity-on-objects functor $E: \HypI{\Sigma} \to \Ecospans$ which maps a cospan $n \xrightarrow{f} \mathcal{G} \xleftarrow{g} m$ to an extended cospan $n \xrightarrow{f_{ext}} n \xrightarrow{f_{int}} \mathcal{G} \xleftarrow{g_{int}} m \xleftarrow{g_{ext}} m$ such that $f_{ext};f_{int} = f$ and $g_{ext};g_{int} = g$.
% \end{remark}

Composition of two morphisms 
\begin{align*}
	n \xrightarrow{f_{ext}} n' \xrightarrow{f_{int}} &\mathcal{F} \xleftarrow{f'_{int}} m' \xleftarrow{f'_{ext}} m\\
	m \xrightarrow{g_{ext}} m'\!' \xrightarrow{g_{int}} &\mathcal{G} \xleftarrow{g'_{int}} k' \xleftarrow{g'_{ext}} k
\end{align*} is computed in two stages.
First, $\mathcal{H}$ is computed as the result of the pushout square shown below: 
% https://q.uiver.app/#q=WzAsMTAsWzAsMCwiZXh0KFgpIl0sWzEsMSwiWCJdLFsyLDIsIkEiXSxbMywwLCJZIl0sWzQsMCwiZXh0KFkpIl0sWzUsMCwiWSciXSxbNiwyLCJCIl0sWzcsMSwiWiJdLFs4LDAsImV4dChaKSJdLFs0LDMsIl57XFx1bGNvcm5lcn1BK197Zl97ZXh0fTtmLGdfe2V4dH07Z31CIl0sWzAsMV0sWzEsMl0sWzMsMiwiZiJdLFs0LDMsImZfe2V4dH0iXSxbNCw1LCJnX3tleHR9IiwyXSxbNSw2LCJnIiwyXSxbNyw2XSxbOCw3XSxbMiw5XSxbNiw5XV0=
\[
\trimbox{0cm 0cm 0cm 0.75cm}{
\adjustbox{scale=0.9,center}
{\input{./figures/pushout_e_cospans.tikz}}
}
\]
then,  the result of composition is defined as follows:
\[
n \xrightarrow{f_{ext};\iota_1} n' + (m'' \setminus m) \xrightarrow{h_{1}} \mathcal{H} \xleftarrow{h_2} k' + (m' \setminus m) \xleftarrow{g'_{ext};\iota_1} k
\]
where $h_i$ are defined as follows,  using $|$ to denote the restriction of a function on a discrete e-hypergraph. 
\[
    h_1 = [ f_{int};p_1, ~(g_{int};p_2)|_{m'' \setminus m} ]
\;
    h_2 = [ g'_{int};p_2, ~(f'_{int};p_1)|_{m' \setminus m} ]
\]
% \[
% \begin{tikzcd}
% Y_{ext} \arrow[d] & 0 \arrow[l] \arrow[d]    \\
% Y'          & Y' \setminus Y_{ext} \arrow[l]
% \end{tikzcd}
% \]
% and then the disjoint union is just a coproduct. 
% where $g'$ is the restriction of $g$ to ${Y' \setminus Y_{ext}}$ and $f'$ the restriction of $f$ to $Y \setminus Y_{ext}$. 
The identity of composition is given by the obvious extended cospan of identities. 
$\Ecospans$ inherits a symmetric monoidal structure from the coproduct (and initial object) structure of $\catname{EHyp({\Sigma}})$,  analogously to Definition~\ref{def:cspd}.

As in the standard cospan construction,  it is necessary to consider isomorphism classes of cospans since composition (defined by pushout) is associative only up-to isomorphism (since pushouts are unique only up-to isomorphism). 
% Unlike in the previous section,  we require (extended) cospans to consist of monomorphisms.  
% This is to ensure the existence of the pushouts needed for composition,  since pushouts do not necessarily exist in $\textbf{EHyp}(\Sigma)$.  
% This also means that we omit the monomorphism condition from the following definition.  
We additionally introduce a notion of \textit{well-typing} of e-hypergraphs,  which ensures each component of an e-box has the same number of inputs (outputs) as every other component. 

% https://q.uiver.app/#q=WzAsNSxbMCwwLCJYICsgWSJdLFsyLDAsIlggKyBZIl0sWzQsMCwiWSArIFgiXSxbNiwwLCJZICsgWCJdLFs4LDAsIlkgKyBYIl0sWzAsMSwiaWRfWCArIGlkX1kiXSxbMSwyLCJcXHNpZ21hX3tYLFl9Il0sWzQsMywiaWRfWSArIGlkX1giLDJdLFszLDIsImlkX1kgKyBpZF9YIiwyXV0=

\begin{definition}[Category of well-typed MDA e-hypergraphs with interfaces]
We call an extended cospan in $\Ecospans$,  as below,   \textit{monogamous directed acyclic (MDA)} if
\[
n \xrightarrow{f_{ext}} n' \xrightarrow{f_{int}} \mathcal{G} \xleftarrow{g_{int}} m' \xleftarrow{g_{ext}} m
\]
\begin{enumerate}
        \item underlying hypergraph of $\mathcal{G}$ is directed acyclic; 
        \item $f_{int}$, $g_{int}$ are monomorphisms; 
        \item in-degree and out-degree of every vertex is at most 1; 
        \item vertices with in-degree $0$ are precisely the image of $f_{int}$; 
        \item vertices with out-degree $0$ are precisely the image of $g_{int}$; 
        % the item below is not assumed by the definition of a cospan of e-hypergraphs
        % \item $[v) = \varnothing$ for all $v$ in the image of $f_{ext};f$, and for all $v$ in the image of $g_{ext};g$. 
\end{enumerate}
\label{def:monogamy_ehyp} 

Further,  consider the sets of input and output vertices,  $I$ and $O$ respectively,  of a hierarchical edge $e$ of an e-hypergraph $\mathcal{G}$. 
The consistency relation of $\mathcal{G}$ partitions $I$ and $O$.  
We call $e$ \emph{well-typed} if the size of each subset in the partition of $I$ is $|s(e)|$ and the size of each partition of $O$ is $|t(e)|$.
We will call an MDA extended-cospan in $\MdaEcospans$ \emph{well-typed} if all hierarchical edges in its carrier are well-typed.
We denote by $\MdaEcospans$ the subcategory of $\Ecospans$ consisting of well-typed MDA extended-cospans. 
\end{definition}

% To save space
% Note,  $\MdaEcospans$ is indeed a category,  and in fact inherits the symmetric monoidal structure of $\Ecospans$:  it is easy to check that identities,  the monoidal unit,  symmetry,  and the composition and tensor product of any two well-typed cospans are all (again) well-typed.  

Unlike in the previous section,  in order to give an equivalence with $\catname{S}^+(\Sigma)$,  we must first develop the theory of DPOI rewriting for e-hypergraphs.
This is because the equations for semilattice enrichment are not subsumeable by the e-hypergraph representation, apart from the commutativity equation which is absorbed by the definition of isomorphic cospans.
% , as semilattice equations involve sharing and copying, operations on string diagrams which are not usually quotiented away.
Instead,  we will treat them via rewriting,  as we do with equations of the signature.  
% This is because the semilattice equations involve sharing and copying, operations on string diagrams which are not usually quotiented away.
% , cf. string diagrams for Cartesian categories.

% To save space
% An ideal combinatorial representation would factor out the two semilattice equations which \textit{do not} involve sharing and copying (namely,  associativity and idempotence) --- however, a limitation of the current work is that we do not achieve this.
% Presently,  we will treat these equations via rewriting as well. 

\section{DPOI-Rewriting for E-Hypergraphs}

Because extended cospans have a more general notion of interface,  including \textit{internal} vertices,  DPOI rewriting as presented in Section \ref{sec:combinatorial-semantics} needs some adjustments.

We do not expect internal interfaces to be preserved during rewriting: for example,  when the semilattice equations are considered as rewrites. 
Thus,  we wish for a rewrite rule to be a pair of \textit{extended} cospans of e-hypergraphs with matching \textit{external} (but not necessarily internal) interfaces,  as follows: 
\[
    n \xrightarrow{} n' \xrightarrow{} \mathcal{L} \xleftarrow{} m' \xleftarrow{} m,
\qquad
    n \xrightarrow{} n'' \xrightarrow{} \mathcal{R} \xleftarrow{} m'' \xleftarrow{} m.
\] 
Analogously to Section \ref{sec:combinatorial-semantics}, observing that the following extended cospans express the same data:
\[
    0 \xrightarrow{} 0 \xrightarrow{} \mathcal{L} \xleftarrow{} n'' + m'' \xleftarrow{} n + m
\]
\[
    0 \xrightarrow{} 0 \xrightarrow{} \mathcal{R} \xleftarrow{} n' + m' \xleftarrow{} n + m
\] allows us to encode rewrite rules as extended cospans of the following form:
\[
\mathcal{L} \xleftarrow{} n' + m' \xleftarrow{} n + m \xrightarrow{} n'' + m'' \xrightarrow{} \mathcal{R},
\] which will fit into the DPO formalism. We make this identification throughout,  without further mention. 

Before introducing our definition of \textit{extended} DPOI rewriting,  we must modify the definition of boundary complement to guarantee that rewriting yields a monogamous directed acyclic e-hypergraph. 
First,  we introduce the  notion of a \textit{down-closed} graph,  which will be necessary for the subsequent development. 
\begin{definition}[Down-closed subgraph]
%A \textit{subgraph} of $\mathcal G $ an e-hypergraph $\mathcal H$ 
%is an e-hypergraph whose edge and vertex sets are subsets of those of $\mathcal H$, 
%and whose remaining equipment is a restriction of that of $\mathcal H$ to those sets. 
    We call  a sub-e-hypergraph $\mathcal G $ of $\mathcal H$ \emph{down-closed} if for all $e \in E_{\mathcal{G}}$,   all children of $e$ are also in $\mathcal{G}$.
\end{definition}    

\begin{definition}[Extended boundary complement]
\label{def:boundary_new}

% https://q.uiver.app/#q=WzAsNyxbMiwwLCJMIl0sWzQsMCwieF97ZXh0fSArIHlfe2V4dH0iXSxbMiwyLCJHXntcXHVyY29ybmVyfSJdLFs0LDIsIkMiXSxbNCw0LCJuX3tleHR9K21fe2V4dH0iXSxbMCwwLCJsX2krbF9vIl0sWzAsMiwiZ19pK2dfbyJdLFsxLDAsIiIsMCx7ImNvbG91ciI6WzAsNjAsNjBdfV0sWzAsMiwibSIsMCx7ImNvbG91ciI6WzAsNjAsNjBdfSxbMCw2MCw2MCwxXV0sWzEsMywiW2lfYyxqX2NdIiwwLHsiY29sb3VyIjpbMCw2MCw2MF19LFswLDYwLDYwLDFdXSxbNCwzLCJbbl9jLG1fY10iLDIseyJjb2xvdXIiOlswLDYwLDYwXX0sWzAsNjAsNjAsMV1dLFszLDIsIiIsMix7ImNvbG91ciI6WzAsNjAsNjBdfV0sWzQsMiwiIiwyLHsiY29sb3VyIjpbMCw2MCw2MF19XSxbNSwwXSxbNiwyXV0=

For MDA cospans 
\[
    i \xrightarrow{} i' \xrightarrow{} \mathcal{L} \xleftarrow{} j' \xleftarrow{} j
\quad\text{ and }\quad
    n \xrightarrow{} n' \xrightarrow{} \mathcal{G} \xleftarrow{} k' \xleftarrow{} k
\] and mono $m : \mathcal{L} \to \mathcal{G}$ in $\catname{EHyp(\Sigma)}$, a pushout complement $i + j \to \mathcal{L}^{\bot} \to \mathcal{G}$
as depicted in the square below
\[
\scalebox{0.75}{

    \input{./figures/DPOI_pushout_complement.tikz}
}
\]
is a \textit{boundary complement} if
\begin{enumerate} 
    \item $m(\mathcal L)$ is a convex down-closed e-hypergraph;
    \item $[c_1,c_2]$ is mono;
    \item for all $v,w$ in $\mathcal{G}$ in the image of $i + j$,  $v$ and $w$ share the same set of parents, and either $v,w$ are top-level or else $v \consistency w$;
    \item for all $v,w$ in $\mathcal{L}^{\bot}$ in the image of $i + j$,  $v$ and $w$ share the same set of parents, and either $v,w$ are top-level or else $v \consistency w$;
    \item there exist $d_1 : n \to \mathcal{L}^\bot$ and $d_2 : k \to \mathcal{L}^\bot$ making the above triangle commute; and
    \item if the image of $\mathcal{L}$'s external interfaces under $m$ consists exclusively of top-level vertices of $\mathcal{G}$ then there exists a \textit{well-typed} extended MDA cospan
    % \[
    % \hspace{-2em} n + j \xrightarrow{f_1 + id_{j}} n' \setminus (i' \setminus i) + j \xrightarrow{[g_1,c_2]} \mathcal{L}^{\bot} \xleftarrow{[g_2,c_1]} k' \setminus (j' \setminus j) + i \xleftarrow{f_2 + id_{i}} k + i
    % \]
\tikzcdset{row sep/mysize/.initial=0.25ex}.
\[
\hspace{-2em}
\begin{tikzcd}[cramped, row sep=mysize, column sep = tiny]
	{n + j} &&&&&& {k + i} \\
	& {n' \setminus (i' \setminus i) + j} & \vspace{16em} & {\mathcal{L}^{\bot}} & \vspace{16em} & {k' \setminus (j' \setminus j) + i}
    \arrow["{f_1 + id_{j}}"{xshift=-0.7ex}, from=1-1, to=2-2]
	\arrow["{f_2 + id_{i}}"'{xshift=1ex}, from=1-7, to=2-6]
	\arrow["{[g_1,c_2]}"{yshift=0.5ex}, from=2-2, to=2-4]
	\arrow["{[g_2,c_1]}"'{yshift=0.5ex}, from=2-6, to=2-4]
\end{tikzcd}
\]
    \item if the image of $\mathcal{L}$'s external interfaces under $m$ consists exclusively of not top-level vertices of $\mathcal{G}$ then there exists a \textit{not necessarily well-typed} extended MDA cospan
    \[
    \hspace{-2em} n \xrightarrow{f_1} n' \setminus (i' \setminus i) + j \xrightarrow{[g_1,c_2]} \mathcal{L}^{\bot} \xleftarrow{[g_2,c_1]} k' \setminus (j' \setminus j) + i \xleftarrow{f_2} k
    \]
\end{enumerate}
where $f_i$ and $g_i$ are defined as follows.  
 %because the image of $m$ is a convex down-closed e-hypergraph, 
The strictly internal interface $i' \setminus i$ of $\mathcal L$ is mapped to the internal interface $n'$ of $\mathcal G$,  since $\mathcal L$ is a down-closed subgraph of $\mathcal G$,  inducing an identification of $i' \setminus i$ in  $n'$. 
Then map $f_1$ is given by $g_{ext}$ with its codomain restricted to $n' \setminus (i' \setminus i)$.\footnote{Noting the image of $g_{ext}$ indeed lies within $n' \setminus (i' \setminus i)$}  The map $g_1$ is derived from the restriction of $g_{int}$ to type $n' \setminus (i' \setminus i) \to \mathcal G$ by further observing that $\mathcal L^\bot$ can be identified within $\mathcal G$ --- and in particular has the internal interface $n'$ of $\mathcal G$ minus $i' \setminus i$.  
The maps $f_2$ and $g_2$ are defined similarly.
%$i' \setminus i$, which is a strict internal interface of $\mathcal{L}$, is necessarily mapped to vertices in the strict internal interface of $\mathcal{G}$ which are also in the image of $n'$ and therefore we can construct a 
%morphism $f' : i' \setminus i \to n'$ by first following $m$ and then composing with the partial inverse of $g_{int}$ which exists because it is mono. 
%Then we define $n' \setminus (i' \setminus i)$ as the pushout complement in the diagram below
%
%\[
%    \tikzfig{combinatorial_semantics/pushout_difference}
%\]
%and $g_1$ is obtained by composition $g';g_{int}$ and then composing with the inverse of $l$.
%Similarly for $g_2$.
% CHRIS: TRY TO DEFINE G1, G2 AS SIMPLY AS POSSIBLE, AND POSSIBLY ADD AN INFORMAL EXPLANATION BELOW THE DEFINITION. 
% Because $\mathcal{L}^{\bot}$ is essentially $\mathcal{G}$ with $\mathcal{L}$ removed, apart from the part that has a pre-image in $i + j$, there is an obvious morphism from $l \to \mathcal{L}^{\bot}$ which is $g_1 = g_{int} |_{D}$ where
% \[
%     D = \{ v \in n' \text{ such that } \not \exists u \in (i' \setminus i) ~ . ~ g_{int}(v) = f'_{int};f(u) \}
% \]
% Similarly for $r \to \mathcal{L}^{\bot}$, which we denote as $g_2$.
% Then there exists a morphism $n \to l$, because there is a morphism $n \to n'$ and the external interfaces of $n'$ and $l$ are the same.
% That is, it is a morphism $f: n \to l$ such that $d_1 = f;g_1$.
% Similarly for $k \to r$.
\end{definition}

First, note that conditions (6) and (7) are disjoint.
Intuitively, the above definition means that when the occurrence of $\mathcal{L}$ within $\mathcal{G}$ is top-level (premise for condition (6)) then $\mathcal{L}$'s external interfaces become the part of $\mathcal{L}^{\bot}$'s external interfaces (as in the case of Definition~\ref{def:dpoi}), and when the occurrence is nested (premise of condition (7)), those interfaces become the part of $\mathcal{L}^{\bot}$'s strict internal interfaces: one can think about the latter as applying Definition~\ref{def:dpoi} in a nested context.
In both cases, when removing the occurrence of $\mathcal{L}$ from $\mathcal{G}$, $\mathcal{L}$'s strict internal interfaces are removed from strict internal interfaces of $\mathcal{G}$ to form internal interfaces of the corresponding MDA cospans above.
The morphisms are then constructed correspondingly to require that inputs (respectively, outputs) of $\mathcal{R}$ are glued to the outputs (respectively, inputs) of $\mathcal{L}^{\bot}$.

%The notion of a boundary complement above is essentially the same as the one from Section~\ref{sec:combinatorial-semantics}.
%CHRIS: THE FOLLOWING SENTENCE IS UNCLEAR. 
%The only difference is that we account for internal interfaces to require monogamous-ness and these interfaces must be computed explicitly to define the monogamous cospan.
%
%\begin{figure}
%    \centering
%    \[
%    \scalebox{0.5}{\tikzfig{combinatorial_semantics/interfaces_change_example}}\]
%    \caption{Hypothetical DPO square for $\catname{MACsp_{D}(EHyp_{\Sigma})}$}
%    \label{fig:interface_change_example}
%\end{figure}

% The second point above means that when cutting the mono-occurrence of $L$ out from $G$ we also need to remove the inner interfaces of $L$ from the whole interface of $G$. \question{According to the pushout squares above, $c_i$ is constructed by removing from $g_i$ everything from $l_i$ which does not have a pre-image in $i$, i.e., everything except the outermost interfaces (same holds for $c_o$)}

Note that,  in the above definition,  the extended cospan 
    \[
    \vspace{-1mm}
    n \xrightarrow{} n' \setminus (i' \setminus i) + j \xrightarrow{[g_1,c_2]} \mathcal{L}^{\bot} \xleftarrow{[g_2,c_1]} k' \setminus (j' \setminus j) + i \xleftarrow{} k
    \vspace{-0.5mm}
    \]
 is not necessarily well-typed.
This is because its input internal interface may contain vertices that previously were a part of the output internal interface and vice versa.
Examples and non-examples of a boundary complement can be found in Appendix~\ref{sec:appendix:iso}.
% \begin{remark}
% The boundary complement conditions, in particular, prohibit finding a match for $\mathcal L$ in $\mathcal G$,  below.  
% \[
% 	\mathcal{L} = \scalebox{0.8}{\tikzfig{../figures/combinatorial_semantics/f_times_g}} \qquad \mathcal{G} = \scalebox{0.6}{\tikzfig{../figures/combinatorial_semantics/f_plus_g_inline}}
% \]
% \end{remark}

\begin{proposition}
\label{prop:boundary_unique}
    The boundary complement in~\ref{def:boundary_new} when it exists is unique.
\end{proposition}

We are now ready to define \textit{convex extended DPOI (EDPOI) rewriting} for $\catname{EHyp({\Sigma})}$.  
It is analogous to Definition~\ref{def:convex_dpo},  except we must construct internal interfaces explicitly.
More precisely, when removing the occurrence of $\mathcal{L}$ from $\mathcal{G}$, \emph{i.e.}, by computing the pushout complement, the internal interfaces of the resulting e-hypergraph should be modified to \textit{exclude} the vertices corresponding to $\mathcal{L}$'s strictly internal interfaces (since the vertices they map to have been removed). 
Then, when gluing $\mathcal{R}$ into the hole,  the internal interfaces of the resulting e-hypergraph should be modified to \textit{include} the strictly internal interfaces of $\mathcal{R}$ (since new internal interfaces for them to map to have been added).

\begin{definition}[Convex EDPOI rewriting]
\label{def:dpoi-e}
Given an extended span of morphisms 
\[
    \mathcal{G} \xleftarrow{} n' + k' \xleftarrow{} n + k \xrightarrow{} n'' + k'' \xrightarrow{} \mathcal{H}
\]
in $\catname{EHyp(\Sigma)}$, we say $\mathcal{G}$ \textit{rewrites (convexly) to} $\mathcal{H}$ (\textit{with external interface} $n + k$ and \textit{taking internal interface} $n'+k'$ \textit{to} $n'' + k''$) --- denoted by $\mathcal{G} \Rrightarrow \mathcal H$  --- \textit{via a rewrite rule} 
\[
    \mathcal{L} \xleftarrow{} i' + j' \xleftarrow{} i + j \xrightarrow{} i'' + j'' \xrightarrow{} \mathcal{R}
\] 
if there exists an object $\mathcal{C}$ and morphisms which complete the following commutative diagram 
\[
 \scalebox{0.75}{
    \input{./figures/DPOI_square.tikz}
 }
\]
such that the two marked squares are pushouts and the following conditions hold:
    \begin{enumerate}
        \label{dpoi-e:assumptions}
       % \item $m : \mathcal{L} \to \mathcal{G}$ is a convex match;
        %\item $m(L)$ is a down-closed e-hypergraph;
        \item $i + j \to \mathcal{C} \to \mathcal{G}$ is a boundary complement;
	\item the internal interfaces of $\mathcal H$ are such that:
	\[
		n'' = n' \setminus (i' \setminus i) + (i'' \setminus i) \qquad k'' = k' \setminus (j' \setminus j) +  (j'' \setminus j)\\
	\]
	\item the map $f_1 = [g_1,h_1]: n'' \to \mathcal H$ in the diagram above consists of $g_1$ as defined in Definition \ref{def:boundary_new} of boundary complement,  and $h_1$ which is the restriction of the composite $i'' \to \mathcal R \to \mathcal H$ to $i'' \setminus i$,  and similarly for the map $f_2: k'' \to \mathcal H$.  
    \end{enumerate}
% The internal interfaces of the resulting e-hypergraph $\mathcal{H}$ from the cospan
%  \[
%     \mathcal{H} \xleftarrow{} n'' + k'' \xleftarrow{} n + k
% \]
\end{definition}
Given  a set $\mathfrak{R}$ of EDPOI rewrite rules and e-hypergraphs with extended interfaces $\mathcal G$ and $\mathcal H$,  we write $\mathcal G \Rrightarrow_\mathfrak{R} \mathcal H$ if there exists a EDPOI rewrite in $\mathfrak R$ such that via it $\mathcal G$ rewrites to $\mathcal H$,  and we write $\mathcal G \Rrightarrow^*_\mathfrak{R} \mathcal H$ for reflexive, symmetric, and transitive closure of $\Rrightarrow{}$.
An example of rewriting of extended cospans can be found in Appendix~\ref{sec:appendix:iso}.
% Together, the conditions ensure the result of a rewrite is a well-typed MDA e-hypergraph. 

With our modified definition of extended DPOI rewriting in hand,  we proceed to prove that our combinatorial representation of $\catname{SLat}$-SMCs is sound and complete in an appropriate sense.

%\begin{proposition}
%    $n \xrightarrow{} n'' \xrightarrow{} \mathcal{H} \xleftarrow{} k'' \xleftarrow{} k$ is a well-typed mda-cospan.
%\end{proposition}
   
%Above we defined rewriting for cospans with empty input interfaces, when we will further build a correspondence between $\textsf{PROP}(\Sigma)^{+}$ and $\WellTypedMdaEcospans$.  %we will operate with cospans with non-empty input interfaces as well, so below we formulate what it means to rewrite cospans with non-empty input interfaces.
%
%\begin{definition}[Convex EDPOI in $\WellTypedMdaEcospans$]
%    
%We say that an \textit{mda-cospan}
%\[
%    n \xrightarrow{} n' \xrightarrow{} \mathcal{G} \xleftarrow{} k' \xleftarrow{} k
%\]
%\textit{rewrites convexly into another mda-cospan}
%
%\[
%    n \xrightarrow{} n'' \xrightarrow{} \mathcal{H} \xleftarrow{} k'' \xleftarrow{} k
%\]
%
%under a rewrite rule 
%% \[
%%     \langle x_{ext} \xrightarrow{a_1} l_{i} \xrightarrow{a_2} L \xleftarrow{b_2} l_{o} \xleftarrow{b_1} y_{ext}, 
%%     x_{ext} \xrightarrow{a_1} r_{i} \xrightarrow{a_2} R \xleftarrow{b_2} r_{o} \xleftarrow{b_1} y_{ext} \rangle
%% \]
%
%\[\mathcal{L} \xleftarrow{} i' + j' \xleftarrow{} i + j \xrightarrow{} i'' + j'' \xrightarrow{} \mathcal{R}\]
%
%if 
%\begin{align*}    
%    \mathcal{G} \xleftarrow{} n + k \Rrightarrow^{+} \mathcal{H} \xleftarrow{} n + k
%\end{align*}
%under the same rewrite rule.
%We will use the same $\Rrightarrow^{+}$ and the concrete definition will be clear from the context.
%
%\end{definition}
%

\section{Soundness and Completeness}\label{sec:soundness-and-completeness}

This section contains the main technical results of the paper,  building on those of \cite{bonchi_string_2022-2}.
We give the interpretation of $\catname{S}(\Sigma)^{+}$ terms in $\WellTypedMdaEcospans$ by extending the interpretation of $\catname{S}(\Sigma)$ in $\MdaCospans$ (see Appendix) by first endowing each hom-set of cospans in $\WellTypedMdaEcospans$ with a semilattice structure turning it into an $\catname{SLat}$-category.
% To achieve the latter we quotient each hom-set by EDPOI structural rewrite rules that correspond to the axioms of an $\catname{SLat}$-SMC:
% the SMC equations are factored out in our representation of terms in $\WellTypedMdaEcospans$,  but the semilattice equations are not.
Detailed proofs of the statements in this section can be found in Appendix~\ref{proof:appendix:soundness}.

% We first construct an interpretation of terms generating $\catname{S}(\Sigma)^{+}$ in $\WellTypedMdaEcospans$.  
% As expected,  the SMC equations are factored out in our representation of terms in $\WellTypedMdaEcospans$,  but the semilattice equations are not.  
% Instead,  we implement the semilattice equations via EDPOI rewrites,  leading to the following soundness and completeness result: for any morphisms $f$ and $g$ in $\catname{S}(\Sigma, \mathcal E)$,  $f = g$  if and only if there exists a sequence of EDPOI-rewrites --- each induced by either a structural semilattice equation or $\mathcal E$ --- between their combinatorial representations. 

\begin{definition}
We define the join $+$ of two extended cospans in $\WellTypedMdaEcospans(n,m)$ as in Figure~\ref{fig:A+B}
by introducing a hierarchical edge with $n$ inputs and $k$ outputs that has both carriers of the two cospans as its successors.
The two carriers become two different consistent components of a newly introduced hierarchical edge such that the interface morphisms are preserved.

% Saving space
% In particular, $f_{1,\text{int}}(b_i) = f_{\text{int}}(b_i)$ and so on.
% The only extended cospan of type $0 \to 0$ has the empty e-hypergraph as its carrier and we will overload $+$ for such cospans to have the empty carrier as well.
% That is, we will let 
% \begin{align*}
% 	0 \to 0 \to \;&\varnothing \xleftarrow{} 0 \xleftarrow{} 0\\
% 	&\;+ \hspace{6em} \Coloneqq \hspace{2em} 0 \to 0 \to \varnothing \xleftarrow{} 0 \xleftarrow{} 0\\
% 	0 \to 0 \to \;&\varnothing \xleftarrow{} 0 \xleftarrow{} 0
% \end{align*}
\end{definition}

To make this newly defined $+$ obey the laws of a semilattice-enriched SMC we will quotient $\WellTypedMdaEcospans$ by certain EDPOI structural rewrite rules.
We define the appropriate notion of quotient by EDPOI rewrites next.  

% Let the function $\llbracket - \rrbracket$ extend in the obvious way to apply to equations (\textit{i.e.}, pairs of $\Sigma^+$-terms).
% Given a set of equations $\mathcal{E}$ of $\Sigma^{+}$-terms, $\llbracket \mathcal{E} \rrbracket$ is defined as $\{\langle \llbracket l \rrbracket, \llbracket r \rrbracket \rangle \text{ for each } l = r \in \mathcal{E} \} \cup \{\langle \llbracket r \rrbracket, \llbracket l \rrbracket \rangle \text{ for each } l = r \in \mathcal{E}\}$.
% \update{below definition depends on whether we require the homomorphism to reflect the conflicts}
\begin{definition}[Quotient by rewrites]  
Given a set of EDPOI rewrite rules $\mathcal{E}$,  we denote by $\WellTypedMdaEcospans/\mathcal{E}$ the category $\WellTypedMdaEcospans$ quotiented by the following relation.
\vspace{-2mm}
\[
	f \sim g \quad \text{if} \quad f \Rrightarrow^{*}_{\mathcal{E}} g ~ . 
\vspace{-1mm}
\]
\end{definition}

\begin{figure}
    \[
    \scalebox{0.55}{
    \input{./figures/f_plus_g_new.tikz}
    }
    \]
    \captionsetup{belowskip=-3.5ex}
    \caption{$+$ of two morphisms in $\WellTypedMdaEcospans$}
    \label{fig:A+B}
\end{figure}
% Recall from Section \ref{sec:e-hypergraphs} that $ \WellTypedMdaEcospans$ has a symmetric monoidal structure inherited from the coproduct structure of $\textbf{EHypI}(\Sigma)$.  
% Recall further that typed $\Sigma^+$-terms are those categorical combinator terms freely constructed from generators $c \in \Sigma$, $\textsf{id}_I$, $\textsf{id}_1$, $\sym_{1,1}$, $(;\!)$ and $\otimes$,  and $f+g$.  
% We will write $\Sigma^+(n,m)$ for the set of $\Sigma^+$-terms of type $n \to m$.  

% pag. 11, Def VI.2: it is not clear to me that this relation is closed under composition
% It is closed under composition since the composition and rewrites are defined in terms of pushouts.

\begin{definition}
	Let $\mathcal{S}$ be the set of EDPOI rewrite rules (or rather schema rules) that correspond to string diagrammatic equations in Figure~\ref{fig:string-equations}.
\end{definition}

For example, the first equations would yield two EDPOI schema rules (one for each reading the span from left to right and from right to left):

\[
\adjustbox{scale=0.4}{
	\input{./figures/semilattice_rule_1.tikz}~.
}
\]
Note that we do not need to introduce rewrite rules for commutativity equations as these are absorbed by our notion of isomorphism in Definition~\ref{def:iso}.
Then we have the following result. 
\begin{proposition}[Soundness]
\label{prop:soundness}
The category $\WellTypedMdaEcospans/{\mathcal{S}}$ is a semilattice-enriched PROP. 
\end{proposition}
Note that the symmetric monoidal equations are absorbed in the hypergraph representation,  but the semilattice equations are instead covered by the fact we have quotiented our representation by certain rewrites. 
Ultimately, we want to give interpretation of morphisms in $\catname{S}^{+}(\Sigma)$ in $\WellTypedMdaEcospans / \mathcal{S}$ for which we will need a few lemmas.

% First, notice that this quotienting gives us certain normal forms.

% \begin{lemma}
% \label{lemma:normal_form}
% For each cospan $f$ in ${\WellTypedMdaEcospans}/{\mathcal{S}}$ there is a \textit{normal form} such that 
% \[
% 	f = f_1 + \ldots + f_n
% \]
% 	such that each $f_i$ contains no hierarchical edges,  and for all $i \neq j$ we have $f_i \neq f_j$.
% \end{lemma}

% The reason for calling the previous result \textit{soundness} is that the interpretation function on $\Sigma^+$-terms induces a semilattice-enriched SMC functor $\llbracket - \rrbracket: \catname{S}(\Sigma)^{+} \to \WellTypedMdaEcospans/{\mathcal{S}}$,  which further induces the following semilattice-enriched SMC functor,  for a set of equations $\mathcal E$:
% \[
% 	\llbracket - \rrbracket_{\mathcal E}: \catname{S}(\Sigma, \mathcal E)^{+} \to \WellTypedMdaEcospans/{\mathcal{S,E}}
% \]
% which by uniqueness must be the free semilattice-enriched functor from $\catname{S}(\Sigma, \mathcal E)^{+}$ induced by the interpretation of $\Sigma$ as itself.  
% Unfolding definitions,  it follows that
% \[
% 	f = g \quad \Rightarrow \quad \llbracket f \rrbracket \Rrightarrow^{*}_{\mathcal{S}, \mathcal{E}} \llbracket g \rrbracket~ . 
% \]
% Note that it is immediate from the previous result that $\WellTypedMdaEcospans/{\mathcal{S,E}}$ is indeed a semilattice-enriched PROP.  

\begin{lemma}
\label{lemma:2functor_equiv}
We have the following equivalence obtained by applying the free enrichment 2-functor $\mathcal{F}$ to $\catname{S}(\Sigma) \cong \MdaCospans$
\[
\catname{S}(\Sigma)^{+} \cong \MdaCospans^{+}~.
\]
\end{lemma}

% \begin{lemma}

% There is an inclusion
% \[
% \MdaCospans \xhookrightarrow{In} \MdaEcospans
% \]

% That is identity on objects and which maps a cospan $n \xrightarrow{} F \xleftarrow{} m$ to $n \xrightarrow{} n \xrightarrow{} F \xleftarrow{} m \xleftarrow{} m$ by recalling that every hypergraph is also an e-hypergraph.
% \end{lemma}

\begin{lemma}
\label{lemma:cospans_plus_equiv}
We have the following equivalence of $\catname{SLat}$-categories
\vspace{-2mm}
\[
\MdaCospans^{+} \cong \MdaEcospans/\mathcal{S}~.
\vspace{-1mm}
\]
\end{lemma}

\begin{theorem}
	\label{thm:completeness_simple}
	We have the following equivalence of $\catname{SLat}$-categories
	\vspace{-2mm}
	\[
		\catname{S}(\Sigma)^{+} \cong \WellTypedMdaEcospans / \mathcal{S}~.
	\vspace{-1mm}
	\]	
	\end{theorem}
% \begin{proof}
% 	The equivalence is given by an $\catname{SLat}$-functor $\llbracket \rrbracket$ that is obtained by composing the functors from Lemma~\ref{lemma:2functor_equiv} and Lemma~\ref{lemma:cospans_plus_equiv}.
% \end{proof}
\vspace{-2mm}
The benefit of working with the equivalence defined above as compared to the one in Lemma~\ref{lemma:2functor_equiv} is that the former allows us to define EDPOI rewriting more naturally as every morphism is an (extended) cospan.
Finally, we can state the following result.

\begin{theorem}[Full completeness]
\label{thm:full-completeness}
Let $\mathcal{E}$ be a set of equations $l = r$ of $\Sigma$-terms, then we have the following equivalence of $\catname{SLat}$-categories
\vspace{-2mm}
\[
	\catname{S}^{+}(\Sigma, \mathcal {E} ) \cong \WellTypedMdaEcospans / \mathcal{S,E}~.
\vspace{-1mm}
\]
where $\mathcal{E}$ on the right is overloaded as a set of EDPOI rewrite rules containing $\langle \llbracket l \rrbracket, \llbracket r \rrbracket  \rangle$ and $\langle \llbracket r \rrbracket, \llbracket l \rrbracket  \rangle$ for every equation $l = r$.
\end{theorem}

The literal meaning of the above statement is that $f = g$ in $\catname{S}^{+}(\Sigma)$ modulo $\mathcal{E}$ if and only if $\llbracket f \rrbracket \Rrightarrow{}^{*}_{\mathcal{E}} \llbracket g \rrbracket$.

Now we are ready to define DPO (EDPOI) rewriting for e-graphs by following the informal translation in Section~\ref{sec:introduction} that renders a given \textit{canonical} \textit{acyclic} \textit{connected} e-graph as a morphism of $\catname{S}^{+}(\Sigma_{C}, \mathcal{E}_{C})$, or as morphisms of $\catname{MEHypI}(\Sigma_{C}) / \mathcal{S}, \mathcal{E}_{C}$  after applying the functor of Theorem~\ref{thm:full-completeness}.
An e-graph is \emph{canonical} if it directly corresponds to its graphical representation (as shown in~\nameref{sec:introduction}), and \emph{connected} if it contains a single connected component.
\begin{proposition}
The data above define a function $\llbracket - \rrbracket : \catname{E}\text{-}\catname{graph} \to \catname{MEHypI}(\Sigma_{C}) / \mathcal{S}, \mathcal{E}_{C}$ where the latter has a naturally defined EDPOI rewriting.
Moreover, we have that the following diagram commutes
\[
\begin{tikzcd}
		e_1 \arrow[r, "\leadsto"] \arrow[d, shift right=2]                & e_2 \arrow[d]                                        \\
		\llbracket e_1 \rrbracket \arrow[r, "\Rrightarrow^{*}"] \arrow[u] & \llbracket e_{2} \rrbracket \arrow[u, shift right=2]
\end{tikzcd}
\]
\end{proposition}

Note that the domain of this function is the set of e-graphs over the same signature and equations, which is not a category.
The arrow $\leadsto$ expresses the fact that e-graph $e_2$ is obtainable from e-graph $e_1$ via the application of equations from $\mathcal{E}$; it includes adding nodes and merging e-classes.
Intuitively, the diagram above implies that we can mimic e-graph transformations by DPO rewrites over e-hypergraphs in a sound and complete manner.
We defer a more elaborate analysis to Appendix~\ref{sec:appendix:e-graph-translation}.

\section{Conclusion}
\label{sec:conclusion}

We have seen how e-graphs over algebraic theories can be expressed in terms of Cartesian categories enriched over the category of semilattices. 
In fact, e-graphs are an instance of a more general construction that can account for monoidal theories in terms of SMCs,  similarly enriched.  
Having described a translation of e-graphs into this categorical framework, we also provide a combinatorial representation of morphisms in the free category.
As expected, this representation factors out the structural SMC equations, while remaining sensitive to the complexity-relevant equations induced by the enrichment.  
We then showed our theory to be sound and complete with respect to our categorical semantics.

In the future,  we intend to investigate a number of natural extensions to mathematical treatment of (acyclic) e-graphs developed here.
Some of the extensions that can be considered are: support for `functorial boxes'~\cite{mellies_functorial_2006} (in addition to e-boxes), which have a variety of applications, including the Cartesian-closure required for the implementation of functional programming languages~\cite{ghica-zanassi2023string};
% , as sketched in Fig. \ref{fig:app1}(a)
the `spiders' used by the ZX-calculus to model quantum circuits~\cite{coecke_interacting_2011,ZX};
% , as sketched in Fig. \ref{fig:app1}(b)
or trace~\cite{joyal_geometry_1991, Hasegawa-traced} which can be used to express feedback in categorical models of digital circuits~\cite{ghica_jung_2017,ghica_compositional_2023} and, indeed in conventional e-graphs to encode infinite equivalence classes.
% , as sketched in Fig.~\ref{fig:app2}. 
All these application domains stand to benefit from the use of optimisation techniques, which in turn can take advantage of the appropriate e-graph technique.
In each case above,  we aim to combine our combinatorial representation with existing combinatorial representations of the relevant structure as the hypergraph representation and DPO-rewriting theory of each case is already well-studied independently of any e-graph structure \cite{ghica_rewriting_2023,alvarez-picallo-functorial_2021}.

\section*{Acknowledgement}
The authors are grateful to Paul Blain Levy for insightful discussions on enriched category theory.
This work was supported by EPSRC grants EP/V001612/1 and EP/Y010035/1.
% Aleksei Tiurin was supported by EPSRC Studentship 2741390.

\bibliographystyle{acm}
\bibliography{bibliography.bib}

\appendix
\ifdefined\ONECOLUMN
% \section{}
\section{$\catname{SLat}$}
\else
\subsection{$\catname{SLat}$}
\fi
\label{sec:appendix:slat}

In this section we define the category of semilattices that we use as a base for enrichment throughout the paper.

\begin{definition}[Semilattice]
    A \textit{semilattice} is a set equipped with an operation that we denote as $+$ which is associative, commutative and idempotent.
  \end{definition}
  
  Note that we do not require the existence of a unit for $+$. 
  Semilattices that satisfy this extra requirement are sometimes called \textit{bounded}, i.e., they are idempotent commutative monoids.
  
  \begin{definition}[Semilattice homomorphism]
  
  A homomorphism between two semilattices $S_{1}$ and $S_{2}$ is a map $h$ that respects $+$.
  That is, for all $s,s' \in S_{1}$, $h(s +_{S_{1}} s') = h(s) +_{S_{2}} h(s')$.
  \end{definition}
  
  \begin{definition}[Category of semilattices]
    
  Semilattices with their respective homomorphisms form a category that we denote $\catname{SLat}$.
  \end{definition}
  
  \begin{proposition}
    $\catname{SLat}$ is a closed symmetric monoidal category.
  \end{proposition}
  \begin{proof}
    The tensor product of two semilattices $S_{1}$ and $S_{2}$ is defined as follows.
    $S_{1} \otimes S_{2}$ consists of pairs $(s_1,s_2)$ $s_{1} \in S_{1}$, $s_{2} \in S_{2}$ quotiented by commutativity, idempotence and associativity and additionally by the following relations
    \begin{itemize}
      \item $(s_{1} +_{S_{1}} s_{1}',s_{2}) \equiv (s_{1},s_{2}) +_{S_{1} \otimes S_{2}} (s_{1}',s_{2})$
      \item $(s_{1}, s_{2} +_{S_{2}} s_{2}') \equiv (s_{1},s_{2}) +_{S_{1} \otimes S_{2}} (s_{1}',s_{2})$
    \end{itemize}
  
    Note that this construction does not identify $(s_1,s_2) +_{S_1 \otimes S_2} (s_1',s_2')$ and $(s_1 +_{S_1} s_1', s_2 +_{S_2} s_2')$.
    The unit for this tensor product is $I = \{*\}$ --- a one-element semilattice.
    Clearly $S \otimes I \cong S$ by mapping $(s,*) \mapsto s$ for all $s \in S$ and vice versa.
    The symmetry, associators and unitors are then obvious morphisms.
    Finally, the category is closed since the set of homomorpisms between two semilattices is a semilattice by defining $(f + g)(x)$ as $f(x) + g(x)$.
    $f + g$ is a homomorphism since $(f + g)(x+y) = f(x+y) + g(x+y) = f(x) + f(y) + g(x) + g(y) = f(x) + g(x) + f(y) + g(y) = f(x+y) + g(x+y)$.
  \end{proof}
  
  This makes $\catname{SLat}$ a suitable base for enrichment.
%    and allows us to take $\mathcal{V} = \catname{SLat}$.

  \begin{definition}
    Forgetful functor $U : \catname{SLat} \to \catname{Set}$ is given by $\catname{SLat}(I_{\catname{SLat}, -}) : \catname{SLat} \to \catname{Set}$.
    \end{definition}
    
    Intuitively, the above functor returns the underlying set of a given semilattice $S$ as each morphism from $\{*\} \to S$ picks out an element of $S$.
    
    \begin{proposition}[Special case of Proposition 6.4.6~\cite{Borceux_1994}]
      The forgetful functor $U : \catname{SLat} \to \catname{Set}$ has a left adjoint free functor $F : \catname{Set} \to \catname{SLat}$.
    \end{proposition}
    \begin{proof}
      The functor $F$ is defined by letting $F(A) = \coprod_{A} I_{\catname{SLat}}$ where the latter is a coproduct which is a `free' product of semilattices defined as follows.
      The elements of $S_{1} \coprod S_{2}$ are sequences $s_{1} + s_{2} + \ldots + s_{n}$ where each $s_{i}$ is either from $S_{1}$ or $S_{2}$ quotiented by all relations in $S_{1}$ and $S_{2}$ and by the necessary equations to turn it into a semilattice.
      The adjunction is then given by the following natural isomorphisms
      \begin{align*}
      \catname{SLat}(\coprod_{A}(I_{\catname{SLat}}), B) &\cong \prod_{A}(\catname{SLat}(I_{\catname{SLat}}, B))\\
                                                         &\cong \catname{Set}(A,\catname{SLat}(I_{\catname{SLat}}, B))\\
                                                         &\cong \catname{Set}(A, U(B))
      \end{align*}
      The first isomorphism is given by the fact that hom-functor makes limits into colimits in its first argument, the second being given by the property that $|\catname{Set}(A,B)| = |B|^{|A|} = |\underbrace{B \times \ldots \times B}_{|A|}|$, and the last isomorphism is given by the definition of $U$.
      Furthermore, we have, 
      \[
      F(I) \cong I_{\catname{SLat}}
      \]
      and 
      \begin{align*}
      F(A) \otimes F(B) &\cong (\coprod_{A} I_{\catname{SLat}}) \otimes (\coprod_{B} I_{\catname{SLat}})\\
            &\cong \coprod_{A} (I_{\catname{SLat}} \otimes \coprod_{B} I_{\catname{SLat}})\\
            &\cong \coprod_{A} (\coprod_{B} (I_{\catname{SLat}} \otimes I_{\catname{SLat}}))\\
            &\cong \coprod_{A} (\coprod_{B} I_{\catname{SLat}})\\
            &\cong \coprod_{A \times B} I_{\catname{SLat}}\\
            &\cong F(A \times B)
      \end{align*}
    In the above we used the fact that functors $- \otimes X$ and $X \otimes -$ preserve colimits as they are both left-adjoint by symmetry and monoidal closedness of $\catname{SLat}$.
    \end{proof}

    \begin{definition}
        Given two $\catname{SLat}$-functors $F,G : \mathbb{C} \to \mathbb{D}$ a $\catname{SLat}$-natural transformation is family of morphisms indexed by objects in $\mathbb{C}$, $\alpha_{A} : I_{\catname{SLat}} \to \mathbb{D}(FA,GA)$ in $\catname{SLat}$ such that the following diagram commutes
        \[
        \adjustbox{width=\linewidth}{
        \begin{tikzcd}
          & {\mathbb{C}(A,A')\otimes I_{\mathcal{V}}} && {\mathbb{D}(FA,FA') \otimes \mathbb{D}(FA',GA')} \\
          {\mathbb{C}(A,A')} &&&& {\mathbb{D}(FA,GA')} \\
          & {I_{\mathcal{V}} \otimes \mathbb{C}(A,A')} && {\mathbb{D}(FA,GA) \otimes \mathbb{D}(GA,GA')}
          \arrow["{G_{A,A'}\otimes \alpha_{A'}}"', from=1-2, to=1-4]
          \arrow["{c_{FA,FA',GA'}}"', from=1-4, to=2-5]
          \arrow["{r^{-1}}", from=2-1, to=1-2]
          \arrow["{l^{-1}}", from=2-1, to=3-2]
          \arrow["{\alpha_{A} \otimes F_{A,A'}}", from=3-2, to=3-4]
          \arrow["{c_{FA,GA,GA'}}", from=3-4, to=2-5]
      \end{tikzcd}}
      \]
      \end{definition}

      Intuitively, in the case of free enrichment $\alpha_{A}$ picks a usual natural transformation from $\mathbb{D}(FA,GA)$.
      
      \begin{example}
        Consider a copy natural transformation $\triangle_{A} : A \to A \times A$ of $\catname{S}(\Sigma_{C},\mathcal{E}_{C})$.
        The later Cartesian PROP is lifted to a Cartesian $\catname{SLat}$-PROP where the naturality condition can be traced component-wise where ${*}$ is a one element semilattice:
      
      \[\adjustbox{width=\linewidth}{
      \begin{tikzcd}  
          & {f + g \otimes *} && {\triangle_{A} \otimes ((f+g) \times (f+g))}\\
          f + g &&&& {\triangle_{A'};((f + g) \times (f+g)) = (f + g);\triangle_{A'}}\\
          & {* \otimes f + g} && {f + g \otimes \triangle_{A'}}
          \arrow[from=1-2, to=1-4]
          \arrow[from=1-4, to=2-5]
          \arrow[from=2-1, to=1-2]
          \arrow[from=2-1, to=3-2]
          \arrow[from=3-2, to=3-4]
          \arrow[from=3-4, to=2-5]
      \end{tikzcd}}
      \]
      
      which is the usual naturality condition.
      
      \end{example}
      
\ifdefined\ONECOLUMN
\section{Proofs for Section \ref{sec:e-hypergraphs}: E-hypergraphs}
\else
\subsection{Proofs for Section \ref{sec:e-hypergraphs}: E-hypergraphs}
\fi
\label{sec:appendix:pushout}

In this section we define the conditions under which the pushout in $\catname{EHyp}(\Sigma)$ exists, and we explicitly construct such a pushout.
We also show the uniqueness of a boundary pushout complement.
These condition will play a crucial role in showing that for each enrichment equation in $\textbf{PROP}^{+}(\Sigma)$ there is a sequence of rewrites in $\WellTypedMdaEcospans$.

\begin{proposition}
    Two cospans in $\catname{Hyp}(\Sigma)$ (\textit{i.e.}, morphisms in $\HypI{\Sigma}$) $n \xrightarrow{f} \mathcal{G} \xleftarrow{g} m$ and $n \xrightarrow{f'} \mathcal{G}' \xleftarrow{g'} m$ are equal if and only if
    the following two cospans  $\catname{EHyp}(\Sigma)$ (\textit{i.e.} morphisms in $\Ecospans$)
    \[
    n \xrightarrow{f_{ext}} n \xrightarrow{f_{int}} \mathcal{G} \xleftarrow{g_{int}} m \xleftarrow{g_{ext}} m
    \]
    \[
        n \xrightarrow{f_{ext}'} n \xrightarrow{f_{int}'} \mathcal{G} \xleftarrow{g_{int}'} m \xleftarrow{g_{ext}'} m    
    \]
    such that $f = f_{ext};f_{int}, f'=f_{ext}';f_{int}'$ (respectively, $g = g_{ext};g_{int}, g' = g_{ext}';g_{int}'$) are equal.
\end{proposition}
This essentially means that $\catname{\HypI{\Sigma}}$ faithfully embeds into $\Ecospans$. 
\begin{proof}
    Recall that cospans are equal when they are isomorphic.
    The proposition means that the following cospans are isomorphic (where $\alpha$ is an isomorphism)
    \[\begin{tikzcd}
            & \mathcal{G} \arrow[dd, "\alpha"] &                                                            \\
    n \arrow[ru, "f", bend left] \arrow[rd, "f'"', bend right] &                                  & m \arrow[lu, "g"', bend right] \arrow[ld, "g'", bend left] \\
            & \mathcal{G}'                     &                                                           
    \end{tikzcd}
    \]
    if and only if the following cospans are isomorphic
    \[
        \begin{tikzcd}
            & n \arrow[r, "f_{int}"] \arrow[dd, "\beta"] & \mathcal{G} \arrow[dd, "\alpha"] & m \arrow[l, "g_{int}"'] \arrow[dd, "\gamma"] &                                                                        \\
n \arrow[ru, "f_{ext}", bend left] \arrow[rd, "f_{ext}'"', bend right] &                                            &                                  &                                              & m \arrow[lu, "g_{ext}"', bend right] \arrow[ld, "g_{ext}'", bend left] \\
            & n \arrow[r, "f_{int}'"']                    & \mathcal{G}'                     & m \arrow[l, "g_{int}'"]                      &                                                                       
\end{tikzcd}    
    \]
    such that $f = f_{ext};f_{int}, f'=f_{ext}';f_{int}'$ (respectively, $g = g_{ext};g_{int}, g' = g_{ext}';g_{int}'$).
    The direction from right to left is straightforward: since the bottom diagram commutes $f;\alpha = f_{int};f_{ext};\alpha = f_{int}';f_{ext}' = f'$ and similarly for $g$.
    
    Let's show the direction from left to right.
    Note that $f_{ext}$ and $f_{ext}'$ have inverses since they are injective and surjective (the injection is per the definition and surjection is because domain and codomain are equal).
    Then we need $f_{ext};\beta = f_{ext}'$ and we define $\beta = f_{ext}^{-1};f_{ext}'$.
    Similarly, $\beta^{-1} = f_{ext}'^{-1};f_{ext}$. 
    We can check that $\beta^{-1};\beta = f_{ext}'^{-1};f_{ext};f_{ext}^{-1};f_{ext}' = f_{ext}'^{-1};id;f_{ext}' = f_{ext}'^{-1};f_{ext}' = id$.
    Then we have that $f_{ext};\beta;f_{int}' = f_{ext};f_{ext}^{-1};f_{ext}';f_{int}' = f_{ext}';f_{int}' = f' = f;\alpha = f_{ext};f_{int};\alpha$.
    Since $f_{ext}$ is surjective, we have $\beta;f_{int}' = f_{int};\alpha$.
    Analogously for $\beta^{-1}$ and the right-hand sides of the diagrams.
\end{proof}

\ifdefined\ONECOLUMN
\section{Existence of pushout}
\else
\subsection{Existence of pushout}
\fi

In constructing the pushout the functional versions of e-hypergraph relations will be useful.
\begin{remark}
    Both $<$ and $\consistency$ can be considered as (partial) functions defined on $V_{\mathcal{F}} + E_{\mathcal{F}}$, \textit{i.e.}, on the coproduct of vertices and edges.
    To make things well-typed, we will use corresponding coproduct injections $\iota_{V} : {V_{\mathcal{F}}} \to V_{\mathcal{F}} + E_{\mathcal{F}}$ and $\iota_{E} : {E_{\mathcal{F}}} \to V_{\mathcal{F}} + E_{\mathcal{F}}$ when passing either a vertex or an edge into these functions.
    For example, an immediate successor of a vertex $x$ can be written functionally as $<_{\mathcal{F}}^{\mu}(\iota_{V_{\mathcal{F}}}(x))$.
\end{remark}

\begin{remark}
    Using functional notation we can reformulate the notion of a homomorphism between two e-hypergraphs in the following way.
\end{remark}
\begin{definition}
        \label{def:e-homo-2}    
        A \emph{homomorphism} $\phi: \mathcal{F} \to \mathcal{G}$ of e-hypergraphs $\mathcal{F},\mathcal{G}$ is a pair of functions $\phi_V : V_{\mathcal{F}} \to V_{\mathcal{G}}, \phi_E : E_{\mathcal{F}} \to E_{\mathcal{G}}$ such that
        
        \begin{enumerate}
            \item $\phi$ is hypergraph homomorphism.
            
            \item When $x$ is not a top-level vertex, 
                \[
                \phi_{E}(<_{\mathcal{F}}^{\mu}(\iota_{V_{\mathcal{F}}}(x))) = <_{\mathcal{G}}^{\mu}(\phi_{V};\iota_{V_{\mathcal{G}}}(x))
                \]
                and
                \[
                \phi_{E}(<_{\mathcal{F}}^{\mu}(\iota_{E_{\mathcal{F}}}(x))) = <_{\mathcal{G}}^{\mu}(\phi_{E};\iota_{E_{\mathcal{G}}}(x))  
                \] when $x$ is a not top-level edge.
                \item
            When $x \in E_{\mathcal{F}}$
            \[
                [\phi_{V};\iota_{V_{\mathcal{G}}}, \phi_{E};\iota_{E_{\mathcal{G}}} ]^{*}(\consistency_{\mathcal{F}}(\iota_{E_{\mathcal{F}}}(x)))
                \subseteq
                \consistency_{\mathcal{G}}(\phi_{E};\iota_{E_{\mathcal{G}}}(x))
            \]
            where $\phi_{V};\iota_{V_{\mathcal{G}}} : V_{\mathcal{F}} \to V_{\mathcal{G}} + E_{\mathcal{G}}$, and similarly for $\phi_{E};\iota_{E_\mathcal{G}}$ so that $[\phi_{V};\iota_{V_{\mathcal{G}}}, \phi_{E};\iota_{E_{\mathcal{G}}}] : V_{\mathcal{F}} + E_{\mathcal{F}} \to  V_{\mathcal{G}} + E_{\mathcal{G}}$.
            \item When $x \in V_{\mathcal{F}}$
            \[
                [\phi_{V};\iota_{V_{\mathcal{G}}}, \phi_{E};\iota_{E_{\mathcal{G}}}]^{*}(\consistency_{\mathcal{F}}(\iota_{V_{\mathcal{F}}}(x)))
                \subseteq
                \consistency_{\mathcal{G}}(\phi_{V};\iota_{V_{\mathcal{G}}}(x)).
            \]
            \end{enumerate}
\end{definition}

\begin{theorem}[Existence of pushouts in $\catname{EHyp}(\Sigma)$]
\label{th:existence_of_pushouts}
Consider the following span in $\catname{EHyp}(\Sigma)$
% https://q.uiver.app/#q=WzAsMyxbMCwwLCJaIl0sWzIsMCwiWCJdLFswLDIsIlkiXSxbMCwxLCJmIl0sWzAsMiwiZyIsMl1d
\[\begin{tikzcd}
	Z && X \\
	\\
	Y
	\arrow["f", from=1-1, to=1-3]
	\arrow["g"', from=1-1, to=3-1]
\end{tikzcd}\]
such that
\begin{enumerate}
\label{pushout:assumptions}
    \item $Z$ is a \textit{discrete} e-hypergraph.
    \item \label{assumption:equal_predecessors} $[f_{V}(v_i)) = [f_{V}(v_j))$ and $[g_{V}(v_i)) = [g_{V}(v_j))$ for all $v_{i},v_{j}$ in $V_{Z}$.
    \item \label{assumption:non_ambiguous_predecessors} If $[f_{V}(v)) \not = \varnothing$ then $[g_{V}(v)) = \varnothing$ and if $[g_{V}(v)) \not = \varnothing$ then $[f_{V}(v)) = \varnothing$.
    \item $\consistency(f_{V}(v_i)) = \consistency(f_{V}(v_j))$ and $\consistency(g_{V}(v_i)) = \consistency(g_{V}(v_j))$ for all $v_i,v_j$ in $V_{Z}$.
\end{enumerate}    
then the pushout $X +_{f,g} Y$ exists.
\end{theorem}
\begin{proof}
    Consider the diagram below.
    \[
        \adjustbox{scale=1.25}{
            \begin{tikzcd}
            Z \arrow[r, "f"] \arrow[d, "g"']                                   & X \arrow[d, "\sfrac{\iota_1}{\sim R}"] \arrow[rdd, "j_1", bend left] &   \\
            Y \arrow[r, "\sfrac{\iota_2}{\sim R}"] \arrow[rrd, "j_2"', bend right] & \sfrac{X+Y}{\sim R} \arrow[rd, "u"]                              &   \\
                                                                            &                                                                  & Q
            \end{tikzcd}}
    \]
    Then, the pushout of e-hypergraphs $X$ and $Y$ is computed in two steps.
    First, a coproduct of $X+Y$ is computed which is 
    \[
        X + Y = \{V_{X} + V_{Y}, E_{X} + E_{Y}, s_{X+Y}, t_{X+Y}, \consistency_{X+Y}, <_{X+Y} \}
    \]
    where $s_{X+Y} : E_{X} + E_{Y} \to (V_{X} + V_{Y})^{*}$ which can be defined as a copairing $[s'_{X}, s'_{Y}] : E_{X} + E_{Y} \to (V_{X} + V_{Y})^{*}$ and $s'_{X} : E_{X} \to (V_{X} + V_{Y})^{*}$, $s'_{Y} : E_{Y} \to (V_{X} + V_{Y})^{*}$ defined as $s'_{X} = s_{X};\iota_{1,V}^{*}$ and $s'_{Y} = s_{Y};\iota_{2,V}^{*}$ where $\iota_1,\iota_2$ are corresponding coproduct injections, similarly for $t_{X+Y}, <_{X+Y}, \consistency_{X+Y}$.
    We will omit labels as they are irrelevant to pushout construction.
    Consider relations
    \ifdefined\ONECOLUMN
    \begin{align*}
            S_{V} &= \;\{
            (x_i,y_j) \in (V_{X} + V_{Y}) \times (V_{X} + V_{Y})\; |
            \exists z \in V_{Z} \; . \; x_i = f_{V};\iota_{1,V}(z) \text{ and }\\
            &\qquad y_j = g_{V};\iota_{2,V}(z) \text{ where $x_i \in V_{X}$ and $y_j \in V_{Y}$ }
            \}\\
            &\cup \;\{
                (y_j,x_j) \in (V_{X} + V_{Y}) \times (V_{X} + V_{Y})\; |
          \exists z \in V_{Z} \; . \; x_i = f_{V};\iota_{1,V}(z) \text{ and }\\
          &\qquad y_j = g_{V};\iota_{2,V}(z) \text{ where $x_i \in V_{X}$ and $y_j \in V_{Y}$ }
        \}\\
        &\cup \;\{(x,x)\;\text{where}\; x \in V_{X} + V_{Y}\}\\
        S_{E} &= \;\{(x,x) \text{ where } x \in E_{X} + E_{Y}\}
    \end{align*}
    \else 
    \begin{align*}    
    S_{V} = \{&
          \;(x_i,y_j) \in (V_{X} + V_{Y}) \times (V_{X} + V_{Y})\; | \\
          &\;\exists z \in V_{Z} \; . \; x_i = f_{V};\iota_{1,V}(z) \text{ and } y_j = g_{V};\iota_{2,V}(z)\\
          &\;\text{ where $x_i \in V_{X}$ and $y_j \in V_{Y}$ }\}\\
          \cup&\\
          \{&
          \;(y_j,x_j) \in (V_{X} + V_{Y}) \times (V_{X} + V_{Y})\; | \\
          &\;\exists z \in V_{Z} \; . \; x_i = f_{V};\iota_{1,V}(z) \text{ and } y_j = g_{V};\iota_{2,V}(z)\\
          &\;\text{ where $x_i \in V_{X}$ and $y_j \in V_{Y}$ }\}\\
          \cup&\\
          \{
          &\;(x,x)\;\text{where}\; x \in V_{X} + V_{Y}\}\\
        \\
    S_{E} &= \{(x,x) \text{ where } x \in E_{X} + E_{Y}\}
    \end{align*}
    \fi
    % \begin{align*}
    % S_{E} &= \{\\
    % &(x_i,y_j) \in (E_{X} + E_{Y}) \times (E_{X} + E_{Y})\; |\\
    % &\exists z \in E_{Z} \; . \; x_i = f_{E};\iota_{1,E}(z) \text{ and } y_j = g_{E};\iota_{2,E}(z)\\
    % &\text{ where $x_i \in E_{X}$ and $y_j \in E_{Y}$}\\
    % &\}
    % \end{align*}
and let relations $R_{V},R_{E}$ be their transitive closures respectively.

We then quotient the set vertices and edges in $X + Y$ by $R_{V}$ and $R_{E}$ respectively.
\ifdefined \ONECOLUMN
\begin{align*}
    \sfrac{X + Y}{\sim (R_{V},R_{E})} = \{
        &\sfrac{V_{X} + V_{Y}}{\sim (R_{V},R_{E})}, \sfrac{E_{X} + E_{Y}}{\sim (R_{V},R_{E})}, \sfrac{s_{X+Y}}{\sim (R_{V},R_{E})},\\
        &\sfrac{t_{X,Y}}{\sim (R_{V},R_{E})}, \consistency_{\sfrac{X+Y}{\sim (R_{V},R_{E})}}, <_{\sfrac{X + Y}{\sim (R_{V},R_{E})}}\}    
\end{align*}
\else
    \begin{align*}
        \sfrac{X + Y}{\sim (R_{V},R_{E})} = (&
            \sfrac{V_{X} + V_{Y}}{\sim (R_{V},R_{E})},\\
            &\sfrac{E_{X} + E_{Y}}{\sim (R_{V},R_{E})},\\
            &\sfrac{s_{X+Y}}{\sim (R_{V},R_{E})},\\
            &\sfrac{t_{X,Y}}{\sim (R_{V},R_{E})},\\
            &\consistency_{\sfrac{X+Y}{\sim (R_{V},R_{E})}},\\
            &<_{\sfrac{X + Y}{\sim (R_{V},R_{E})}})    
    \end{align*}
\fi
We will refer to $\sim (R_{V},R_{E})$ just as $\sim$, \textit{e.g.}, by writing $\sfrac{V_{X} + V_{Y}}{\sim}$ and concrete relation will be clear from the context.
Where necessary we will refer to $S_{V}$ and $S_{R}$ as $\sim_{S}$.
We have 
\[
    \sfrac{s_{X+Y}}{\sim} : \sfrac{E_{X} + E_{Y}}{\sim} \to (\sfrac{V_{X} + V_{Y}}{\sim})^{*}
\]
and there is an obvious surjective function $[]_{V} : (V_{X} + V_{Y}) \to (\sfrac{V_{X} + V_{Y}}{\sim})$ that maps elements to their equivalence classes and $[]_{V}^{*}$ is its extension to sequences.
There is also $[] : E_{X} + E_{Y} \to \sfrac{E_{X} + E_{Y}}{\sim}$ and we will omit subscripts as the correct type will be clear from the argument.
We then define 
\[
    \sfrac{s_{X+Y}}{\sim}([e]) = s_{X+Y};[]^{*}(e) = [s_{X+Y}(e)]^{*}
\]
We will also use subscripts when it is important to tell if an element of $E_{X} + E_{Y}$ has a pre-image in either $E_{X}$ or $E_{Y}$ by writing $e_{x}$ or $e_{y}$.
Similarly, we will write $v_{x}$ to refer to a vertex with a pre-image in $V_{X}$.
Likewise, 
\[
    t_{\sfrac{X+Y}{\sim}}([e]) = [t_{X+Y}(e)]^{*}
\]

These definitions make $([]_{V},[]_{E})$ automatically a homomorphism with respect to source and target maps.
These maps are also automatically well-defined functions since $[e_1] = [e_2]$ implies $e_1 = e_2$ since the mappings for edges are mono.

% So far, the construction matches the construction of the pushout in $\catname{Hyp_{\Sigma}}$ which can be found in~\cite{bonchi_string_2022-1} and hence $s_{\sfrac{X+Y}{\sim}}$ and $t_{\sfrac{X+Y}{\sim}}$ are well-defined and $[]$ is a homomorphism with respect to sources and targets and therefore $\sfrac{\iota_1}{\sim}$ and $\sfrac{\iota_{2}}{\sim}$ are.

Recall that $<$ is essentially a transitive closure of $<^{\mu}$ and hence we can first define 
    \[<_{\sfrac{X+Y}{\sim}}^{\mu}(\iota_{\sfrac{E_{X} + E_{Y}}{\sim}}[e])
    \] and 
    \[
        <_{\sfrac{X+Y}{\sim}}^{\mu}(\iota_{\sfrac{V_{X} + V_{Y}}{\sim}}[u])
    \] where $\iota_*$ are injections into $\sfrac{V_{X} + V_{Y}}{\sim} + \sfrac{E_{X} + E_{Y}}{\sim}$.
    We will consider several cases.
    First, assume that $u$ is an element of $V_{X+Y}$ and
    \begin{enumerate}
        \item If there exists $v$ such that $<_{X+Y}^{\mu}(\iota_{V_{X} + V_{Y}}(v))$ is defined and $u \sim v$
              we let
              \[
                <_{\sfrac{X+Y}{\sim}}^{\mu}(\iota_{\sfrac{V_{X} + V_{Y}}{\sim}}([u])) = [<_{X+Y}^{\mu}(\iota_{V_{X} + V_{Y}}(v))]
              \]
        \item \label{def:child_respects_connectivity} $u$ has no pre-image in $V_{Z}$ and $<_{X+Y}^{\mu}(\iota_{V_{X} + V_{Y}}(u))$ is undefined (\textit{i.e.}, $[u) = \varnothing$).
              If there exists $v$ such that $<_{X+Y}^{\mu}(\iota_{V_{X} + V_{Y}}(v)) = e'$ and such that there is an \textit{undirected} path from $[u]$ to $[v]$, then we define
              \[ 
                <_{\sfrac{X+Y}{\sim}}^{\mu}(\iota_{\sfrac{V_{X} + V_{Y}}{\sim}}[u]) = <_{\sfrac{X+Y}{\sim}}^{\mu}(\iota_{\sfrac{V_{X} + V_{Y}}{\sim}}[v])
              \]
    \end{enumerate}
    Note that the existence of a path in the definition (2) implies that there exists $v' \sim v$ such that there is a path from $u$ to $v'$ and furthermore $v \not = v'$.
    If $v' = v$ then it would mean that $[u) \not = \varnothing$ since $[v) \not = \varnothing$ and there is a path from $u$ to $v$.
    There is no case when $u$ has a pre-image in $V_{Z}$ and yet $<_{X+Y}^{\mu}(\iota_{V_{X} + V_{Y}}(u))$ is undefined and there exists $v$ such that $<_{X+Y}^{\mu}(\iota_{V_{X} + V_{Y}}(v)) = e$ and there is a path from $[u]$ to $[v]$.
    The existence of a path would mean that there is $u' \sim u$ and $v' \sim v$ such that there is a path from $u'$ to $v'$ and since $u$ has a pre-image in $V_{Z}$, $u = f_{V};\iota_{1,V}(z_1)$.
    $v' \sim v$ means that both $v'$ and $v$ have a pre-image in $V_{Z}$ (unless $v' = v$ in which case $[u') \not = \varnothing$ as there is a path from $u'$ to $v'$) and since $u = f_{V};\iota_{1,V}(z_1)$ and $[u) = \varnothing$ $f_{V};\iota_{1,V}(z) = \varnothing$ for all $z$.
    Since $[v) \not = \varnothing$ and $v' \sim v$, $v$ is necessarily in the image of $g_{V};\iota_{2,V}$ and for all $z$ $[g_{V};\iota_{2,V}(z)) \not = \varnothing$ as well as for some $z'$ such that $u'' = g_{V};\iota_{2,V}(z')$ and $u'' \sim_{S} u$ which would mean that the case (1) is applicable to $u$.
    Similarly, when $u = g_{V};\iota_{2,V}(z)$.
    
    Otherwise, we leave $<_{\sfrac{X+Y}{\sim}}^{\mu}(\iota_{\sfrac{V_{X} + V_{Y}}{\sim}}([u]))$ undefined.
    % \question{In the case when images of $f$ and $g$ are top-level the function is undefined. Shall we introduce bottom?}
    Now, assume that $e$ is an element of $E_{X+Y}$.
    Consider two cases.
    \begin{enumerate}
    \item  $<_{X+Y}^{\mu}(\iota_{E_{X} + E_{Y}}(e)) = e'$.
            Then we define
        \[
            <_{\sfrac{X+Y}{\sim}}^{\mu}(\iota_{\sfrac{E_{X} + E_{Y}}{\sim}}[e]) = [<_{X+Y}^{\mu}(\iota_{E_{X} + E_{Y}}(e))]
        \]
    \item $[e) = \varnothing$ and there exists $v$ such that $<_{X+Y}^{\mu}(\iota_{V_{X} + V_{Y}}(v)) = e'$ and such that there is an undirected path from $[e]$ to $[v]$, we define 
        \[
            <_{\sfrac{X+Y}{\sim}}^{\mu}(\iota_{\sfrac{E_{X} + E_{Y}}{\sim}}[e]) = <_{\sfrac{X+Y}{\sim}}^{\mu}(\iota_{\sfrac{V_{X} + V_{Y}}{\sim}}[v])
        \] 
    \end{enumerate}
    Otherwise we leave $<_{\sfrac{X+Y}{\sim}}^{\mu}(\iota_{E_{X} + E_{Y}}{\sim}([e]))$ undefined.
    % We will also build this relation in steps, first we do steps for vertices (1) - (3), then step (1) for edges and finally, we will interleave step (2) for edges and step (4) for vertices until fixed point.
    % Essentially, this interleaving means that vertices inherit the relation from incident edges and vice versa.
    Clearly, all the cases above are disjoint.

    Let's show that the definition does not depend on a particular choice of $v$ in (1).
    Suppose there exist $v_1$ and $v_2$ such that $v_1 \not = v_2$ and $[v_1) \not = \varnothing$ and $[v_2) \not = \varnothing$ and such that $u \sim v_1$ and $u \sim v_2$.
    Then, it must be the case that $<_{X+Y}^{\mu}(\iota_{V_{X} + V_{Y}}(v_1)) = <_{X+Y}^{\mu}(\iota_{V_{X} + V_{Y}}(v_2))$.
    Consider cases.
    \begin{itemize}
        \item $u = v_1$ and $u \sim_{S} v_2$. Then there exists $z \in V_{Z}$ such that $u = f_{V};\iota_{1,V}(z) = v_1$ and $v_2 = g_{V};\iota_{2,V}(z)$.
              Since $[v_1) \not = \varnothing$, it must be the case that $[v_2) = \varnothing$ according to assumption \ref{assumption:non_ambiguous_predecessors} and we get a contraction to $[v_2) \not = \varnothing$.
        \item $u = v_1$ and $u \sim v_2$. The latter implies that $u$ has a pre-image in $V_{Z}$ and so does $v_2$.
              For both $[v_1) \not = \varnothing$ and $[v_2) \not = \varnothing$ they should be both in the image of $f_{V};\iota_{1,V}$ or $g_{V};\iota_{2,V}$ and according to assumption \ref{assumption:equal_predecessors} their sets of predecessors must be equal.
        \item $u \sim_{S} v_1$ and $u \sim_{S} v_2$. This implies $u = f_{V};\iota_{1,V}(z_1) = f_{V};\iota_{1,V}(z_2)$ and $v_{1} = g_{V};\iota_{2,V}(z_1)$ and $v_{2} = g_{V};\iota_{2,V}(z_2)$.
              By assumption \ref{assumption:equal_predecessors} it must be the case that $[g_{V}(z_1)) = [g_{V}(z_2))$ which implies that $[v_1) = [v_2)$.
              The case when $u = g_{V};\iota_{2,V}(z_1)$ is symmetric.
        \item $u \sim v_1$ and $u \sim v_2$. Same as above, by assumption \ref{assumption:non_ambiguous_predecessors} both $v_1$ and $v_2$ should be in the image of $f_{V};\iota_{1,V}$ or $g_{V};\iota_{2,V}$ and by assumption \ref{assumption:equal_predecessors} $[v_1) = [v_2)$.
    \end{itemize}

    Let's show that the definition does not depend on a particular choice of $v$ in (2).
    Suppose, $[u) = \varnothing$ and there exist $v_1$ and $v_2$ such that $<_{X+Y}^{\mu}(v_1) = e_1$ and $<_{X+Y}^{\mu}(v_2) = e_2$ and $e_1 \not = e_2$ and such that there is a path from $[u]$ to $[v_1]$ and from $[u]$ to $[v_2]$.
    The existence of a path from $[u]$ to $[v_1]$ means that there is $v_1' \sim v_1$ such that there is a path from $u'$ to $v_1'$ and similarly for $v_2$ and $v_2'$.
    We will proceed by case analysis.
    \begin{itemize}
            \item If $v_1 = v_1'$ or $v_2 = v_2'$ it means that there is a path from $u$ to $v_1$ or from $u$ to $v_2$ and since $[v_1) \not = \varnothing$ as well as $[v_2) \not = \varnothing$ it must be the case that $[u) \not = \varnothing$ which contradicts the assumption that $[u) = \varnothing$.
            \item Suppose $v_1 \sim_{S} v_1'$ which means there exists $z_1 \in V_{Z}$ such that $v_1 = f_{V};\iota_{1,V}(z_1)$ and $v_1' = g_{V};\iota_{2,V}(z_1)$ and suppose $v_2 \sim_{S} v_2'$ which further means there exists $z_2 \in V_{Z}$ such that $v_2 = f_{V};\iota_{1,V}(z_2)$ and $v_2' = g_{V};\iota_{2,V}(z_2)$.
                  Then we have
                  \begin{align*}
                    <_{X+Y}^{\mu}(\iota_{V_{X} + V_{Y}}(v_1)) &= <_{X+Y}^{\mu}(f_{V};\iota_{1,V};\iota_{V_{X} + V_{Y}}(z_1))\\
                                                              &= \iota_{1,E}(<_{X}^{\mu}(f_{V};\iota_{V_{X}}(z_1)))
                \end{align*}
                  and
                  \begin{align*}
                    <_{X+Y}^{\mu}(\iota_{V_{X} + V_{Y}}(v_2)) &= <_{X+Y}^{\mu}(f_{V};\iota_{1,V};\iota_{V_{X} + V_{Y}}(z_2))\\
                                                              &= \iota_{1,E}(<_{X}^{\mu}(f_{V};\iota_{V_{X}}(z_2)))
                \end{align*}
                and $e_1 \not = e_2$ implies $<_{X}^{\mu}(f_{V};\iota_{V_{X}}(z_1)) \not = <_{X}^{\mu}(f_{V};\iota_{V_{X}}(z_2))$ which contradicts the assumption \ref{assumption:equal_predecessors}.
                The case if $v_2 = g_{V};\iota_{2,V}(z_2)$ and $v_1 = f_{V};\iota_{1,V}(z_1)$ ultimately contradicts the assumption \ref{assumption:non_ambiguous_predecessors} as it entails that both
                \[
                    <_{X}^{\mu}(f_{V};\iota_{V_{X}}(z_1)) \not = \varnothing
                \]
                and
                \[
                    <_{Y}^{\mu}(g_{V};\iota_{V_{Y}}(z_2)) \not = \varnothing
                \]
                The cases when $v_1 = g_{V};\iota_{2,V}(z_1)$ and $v_2 = f_{V};\iota_{1,V}(z_2)$, and $v_1 = g_{V};\iota_{2,V}(z_1)$ and $v_2 = g_{V};\iota_{2,V}(z_2)$ are analogous.
            \item Suppose $v_1 \sim_{S} v_1'$ and $v_2 \sim v_2'$ via $w = (x_1, \ldots, x_n)$ such that there is a path from $u$ to $v_1'$ and from $u$ to $x_n = v_2'$.
                  The mere fact that $v_2 \sim v_2'$ implies that $v_2$ should have a pre-image in $V_{Z}$ and the same holds for $v_1$.
                  By the same reasoning as above it must be the case that $<_{X+Y}^{\mu}(\iota_{V_{X} + V_{Y}}(v_1)) = <_{X+Y}^{\mu}(\iota_{V_{X} + V_{Y}}(v_2))$.
            \item The case $v_1 \sim v_1'$ via $w = (x_1, \ldots, x_n)$ and $v_2 \sim_{S} v_1$ is symmetric to the above.
            \item The case when $v_1 \sim v_1'$ and $v_2 \sim v_2'$ is analogous.
    \end{itemize}

    Let's check that this is a well-defined function.
    Suppose $v_1, v_2$ are vertices and $v_1 \sim v_2$, then $<_{\sfrac{X+Y}{\sim}}^{\mu}(\iota_{\sfrac{V_{X}+V_{Y}}{\sim}}[v_1]) = <_{\sfrac{X+Y}{\sim}}^{\mu}(\iota_{\sfrac{V_{X} + V_{Y}}{\sim}}[v_2])$.
    \begin{itemize}
        \item If $[v_2) \not = \varnothing$, then, by definition we have
                \begin{align*}
                    <_{\sfrac{X+Y}{\sim}}^{\mu}(\iota_{\sfrac{V_{X}+V_{Y}}{\sim}}[v_1]) &= [<_{X+Y}^{\mu}(\iota_{V_{X} + V_{Y}}(v_2))]\\
                     &= \text{ as $v_2 \sim v_2$ and $[v_2) \not = \varnothing$ }\\
                     &= <_{\sfrac{X+Y}{\sim}}^{\mu}(\iota_{\sfrac{V_{X}+V_{Y}}{\sim}}[v_2])
                \end{align*}
        \item The case when $[v_1) \not = \varnothing$ is symmetric.
        \item If both $[v_2) = \varnothing$ and $[v_1) = \varnothing$ but there exists $[v_3) \not = \varnothing$ such that $v_1 \sim v_3$, then
        \begin{align*}
            <_{\sfrac{X+Y}{\sim}}^{\mu}(\iota_{\sfrac{V_{X}+V_{Y}}{\sim}}[v_1]) &= [<_{X+Y}^{\mu}(\iota_{V_{X} + V_{Y}}(v_3))]\\
             &= \text{ as $v_3 \sim v_3$ and $[v_3) \not = \varnothing$ }\\
             &= <_{\sfrac{X+Y}{\sim}}^{\mu}(\iota_{\sfrac{V_{X}+V_{Y}}{\sim}}[v_3])\\
             &= \text{ as $v_2 \sim v_1 \sim v_3$}\\
             &= <_{\sfrac{X+Y}{\sim}}^{\mu}(\iota_{\sfrac{V_{X}+V_{Y}}{\sim}}[v_2])
        \end{align*}
    \end{itemize}
    We do not need to check well-defined-ness for edges as $e_1 \sim e_2$ just implies $e_1 = e_2$ for edges.

    We also need to make sure that $([]_{V},[]_{E})$ is homomorphic with respect to $<^{\mu}$.
    That is, we need to check if $[v) \not = \varnothing$, then $[<_{X+Y}^{\mu}(\iota_{V_{X} + V_{Y}}(v))] = <_{\sfrac{X+Y}{\sim}}^{\mu}(\iota_{\sfrac{V_{X} + V_{Y}}{\sim}}([v]))$.
    Since $v \sim v$ this is homomorphic by definition.
    Similarly for edges.
    
    Then $<_{\sfrac{X+Y}{\sim}} : V_{\sfrac{X+Y}{\sim}} + E_{\sfrac{X+Y}{\sim}} \to (E_{\sfrac{X+Y}{\sim}})^{*}$ is defined as a transitive closure of $<_{\sfrac{X+Y}{\sim}}^{\mu}$.
        
    Let's now define $\consistency_{\sfrac{X+Y}{\sim}}$.
    We can consider the consistency relation from the coproduct as a function 
    \[
        \consistency_{X+Y} : (V_{X} + V_{Y}) + (E_{X} + E_{Y}) \to 2^{(V_{X} + V_{Y}) + (E_{X} + E_{Y})}
    \]
    Quotienting the values of the function gives us 
    \[
        \consistency_{X+Y}' : V_{X} + V_{Y} + E_{X} + E_{Y} \to 2^{(\sfrac{V_{X} + V_{Y}}{\sim} + \sfrac{(E_{X} + E_{Y})}{\sim})}
     \]
    which is essentially $[ []_{V};\iota_{\sfrac{V_{X} + V_{Y}}{\sim}}, []_{E};\iota_{\sfrac{E_{X} + E_{Y}}{\sim}}]^{*}$ (a copairing extended to sequences that we will further denote as $[ []_{V}^{\consistency} []_{E}^{\consistency}]$) applied to the return value of $\consistency_{X+Y}$.
    We then first define an auxiliary relation $\consistency^{\hashtag}$ similarly to how $<^{\mu}_{\sfrac{X+Y}{\sim}}$ was defined. We begin with defining $\consistency^{\hashtag}_{\sfrac{X+Y}{\sim}}$ for vertices.

    \begin{enumerate}
        \item If there exists $v$ such that $\consistency_{X+Y}(\iota_{V_{X} + V_{Y}}(v)) \not = \varnothing$ and $u \sim v$, we let
              \ifdefined \ONECOLUMN
              \[
                \consistency_{\sfrac{X+Y}{\sim}}^{\hashtag}(\iota_{\sfrac{V_{X} + V_{Y}}{\sim}}([u]))
                =
                [[]_{V}^{\consistency},[]_{E}^{\consistency}]^{*}(\consistency_{X+Y}(\iota_{V_{X} + V_{Y}}(v)))
              \]
              \else
              \begin{align*}
                &\consistency_{\sfrac{X+Y}{\sim}}^{\hashtag}(\iota_{\sfrac{V_{X} + V_{Y}}{\sim}}([u]))\\
                &=\\
                &[[]_{V}^{\consistency},[]_{E}^{\consistency}]^{*}(\consistency_{X+Y}(\iota_{V_{X} + V_{Y}}(v)))
            \end{align*}
            \fi
        \item $u$ has no pre-image in $V_{Z}$ and $\consistency_{X+Y}(\iota_{V_{X} + V_{Y}}(u)) = \varnothing$ and there exists $v$ such that $\consistency_{X+Y}(\iota_{V_{X} + V_{Y}}(v)) \not = \varnothing$ and such that there is an undirected path from $[u]$ to $[v]$.
        Then we let
        \[
            \consistency_{\sfrac{X+Y}{\sim}}^{\hashtag}(\iota_{\sfrac{V_{X} + V_{Y}}{\sim}}([u])) = \consistency_{\sfrac{X+Y}{\sim}}^{\hashtag}(\iota_{\sfrac{V_{X} + V_{Y}}{\sim}}([v]))
        \]
    \end{enumerate}

    Next we define $\consistency_{\sfrac{X+Y}{\sim}}^{\hashtag}$ for edges.

    \begin{enumerate}
        \item If $\consistency_{X+Y}(\iota_{E_{X} + E_{Y}}(e)) \not = \varnothing$, then
                \begin{align*}
                    \consistency_{\sfrac{X+Y}{\sim}}^{\hashtag}(\iota_{\sfrac{E_{X} + E_{Y}}{\sim}}([e]_{E})) =
                    [[]_{V}^{\consistency}, []_{E}^{\consistency}]^{*}(\consistency_{X+Y}(\iota_{E_{X} + E_{Y}}(e)))
                \end{align*}
        \item $\consistency_{X+Y}(\iota_{E_{X} + E_{Y}}(e)) = \varnothing$ and there exists $v$ such that $\consistency_{X+Y}(\iota_{V_{X} + V_{Y}}(v)) \not = \varnothing$ and such that there is an undirected path from $[e]$ to $[v]$
        then
        \[
            \consistency_{\sfrac{X+Y}{\sim}}^{\hashtag}(\iota_{\sfrac{E_{X} + E_{Y}}{\sim}}([e]_{E})) = \consistency_{\sfrac{X+Y}{\sim}}^{\hashtag}(\iota_{V_{X} + V_{Y}}([v]))
        \]
    \end{enumerate}
    
    The well-definedness of this construction follows by the same argument as the well-definedness of $<_{\sfrac{X+Y}{\sim}}^{\mu}$.
    Then we define
    \ifdefined \ONECOLUMN
    \[
        \consistency_{\sfrac{X+Y}{\sim}}(\iota_{\sfrac{V_{X} + V_{Y}}{\sim}}([v])) \qquad \text{and} \qquad \consistency_{\sfrac{X+Y}{\sim}}(\iota_{\sfrac{E_{X} + E_{Y}}{\sim}}([e]))
    \]
    \else
    \[
        \consistency_{\sfrac{X+Y}{\sim}}(\iota_{\sfrac{V_{X} + V_{Y}}{\sim}}([v]))
    \]
    and
    \[
        \consistency_{\sfrac{X+Y}{\sim}}(\iota_{\sfrac{E_{X} + E_{Y}}{\sim}}([e]))
    \]
    \fi

    as closures of $\consistency^{\hashtag}_{\sfrac{X+Y}{\sim}}$ as below.
    \[
      (\consistency^{\hashtag}_{\sfrac{X+Y}{\sim}}(\iota_{\sfrac{*}{\sim}}([x])))^{c}
    \]
    where $c$ denotes a closure and $\iota_{\sfrac{*}{\sim}}$ is $\iota_{\sfrac{V_{X} + V_{Y}}{\sim}}$ or $\iota_{\sfrac{E_{X} + E_{Y}}{\sim}}$ depending on whether $[x_i]$ comes from $V_{X} + V_{Y}$ or $E_{X} + E_{Y}$, is the smallest set such that
    \begin{itemize}
        \item $[x] \in (\consistency^{\hashtag}_{\sfrac{X+Y}{\sim}}(\iota_{\sfrac{*}{\sim}}([x])))^{c}$ if $\consistency^{\hashtag}_{\sfrac{X+Y}{\sim}}(\iota_{\sfrac{*}{\sim}}([x])) \not = \varnothing$
        \item if $[y] \in \consistency^{\hashtag}_{\sfrac{X+Y}{\sim}}(\iota_{\sfrac{*}{\sim}}([x]))$ then 
        \[
            [x] \in (\consistency^{\hashtag}_{\sfrac{X+Y}{\sim}}(\iota_{\sfrac{*}{\sim}}([y])))^{c}
        \]
        % \item if $[e] \in \consistency^{\hashtag}_{\sfrac{X+Y}{\sim}}(\iota_{\sfrac{V_{X} + V_{Y}}{\sim}}([v]))$ then 
        % \[
        %     [v] \in (\consistency^{\hashtag}_{\sfrac{X+Y}{\sim}}(\iota_{\sfrac{E_{X} + E_{Y}}{\sim}}([e])))^{c}
        % \]
        \item for any sequence $([x_1], \ldots, [x_n])$ such that
        \ifdefined \ONECOLUMN
        \[[x_i] \in \consistency^{\hashtag}_{\sfrac{X+Y}{\sim}}(\iota_{\sfrac{*}{\sim}}([x_{i+1}])) \qquad \text{or} \qquad [x_{i+1}] \in \consistency^{\hashtag}_{\sfrac{X+Y}{\sim}}(\iota_{\sfrac{*}{\sim}}([x_{i}]))\]
        \else
        \[
            [x_i] \in \consistency^{\hashtag}_{\sfrac{X+Y}{\sim}}(\iota_{\sfrac{*}{\sim}}([x_{i+1}]))
        \] or 
        \[
            [x_{i+1}] \in \consistency^{\hashtag}_{\sfrac{X+Y}{\sim}}(\iota_{\sfrac{*}{\sim}}([x_{i}]))
        \]
        \fi
         for $i < n$ both
        \ifdefined \ONECOLUMN
        \[
            [x_1] \in (\consistency^{\hashtag}_{\sfrac{X+Y}{\sim}}(\iota_{\sfrac{*}{\sim}}([x_{n}])))^{c} \qquad \text{and} \qquad [x_n] \in (\consistency^{\hashtag}_{\sfrac{X+Y}{\sim}}(\iota_{\sfrac{*}{\sim}}([x_{1}])))^{c}
        \]
        \else
         \[
            [x_1] \in (\consistency^{\hashtag}_{\sfrac{X+Y}{\sim}}(\iota_{\sfrac{*}{\sim}}([x_{n}])))^{c}
        \]
        and
        \[
            [x_n] \in (\consistency^{\hashtag}_{\sfrac{X+Y}{\sim}}(\iota_{\sfrac{*}{\sim}}([x_{1}])))^{c}
        \]
        \fi
    \end{itemize}
    Let's check that $([]_{V},[]_{E})$ is a homomorphism with respect to $\consistency$.
    We need to check, if
    \ifdefined \ONECOLUMN
    \[
        [[]_{V};\iota_{\sfrac{V_{X} + V_{Y}}{\sim}},[]_{E};\iota_{\sfrac{E_{X} + E_{Y}}{\sim}}]^{*}(\consistency_{X+Y}(\iota_{V_{X} + V_{Y}}(v)))
        \subseteq
        \consistency_{\sfrac{X+Y}{\sim}}(\iota_{\sfrac{V_{X} + V_{Y}}{\sim}}([v]))
    \]
    \else
    \begin{align*}
        &[[]_{V};\iota_{\sfrac{V_{X} + V_{Y}}{\sim}},[]_{E};\iota_{\sfrac{E_{X} + E_{Y}}{\sim}}]^{*}(\consistency_{X+Y}(\iota_{V_{X} + V_{Y}}(v)))\\
        &\subseteq\\
        &\consistency_{\sfrac{X+Y}{\sim}}(\iota_{\sfrac{V_{X} + V_{Y}}{\sim}}([v]))
    \end{align*}
    \fi

    \begin{enumerate}
        \item The first case is when $\consistency_{X+Y}(\iota_{V_{X} + V_{Y}}(v)) = \varnothing$.
              Then the property trivially holds as the empty set is a subset of any set.
        \item Suppose $\consistency_{X+Y}(\iota_{V_{X} + V_{Y}}(v)) \not = \varnothing$.
              Then, 
              \[
                \consistency_{\sfrac{X+Y}{\sim}}^{\hashtag}(\iota_{\sfrac{V_{X} + V_{Y}}{\sim}}([v])) = [[]_{V}^{\consistency},[]_{E}^{\consistency}]^{*}(\consistency_{X+Y}(\iota_{V_{X} + V_{Y}})(v))
              \]
              since $v \sim v$,
              and hence
              \ifdefined \ONECOLUMN
              \begin{align*}
                [[]_{V}^{\consistency},[]_{E}^{\consistency}]^{*}(\consistency_{X+Y}(\iota_{V_{X} + V_{Y}})(v))
                &\subseteq
                \consistency_{\sfrac{X+Y}{\sim}}^{\hashtag}(\iota_{\sfrac{V_{X} + V_{Y}}{\sim}}([v]))\\
                &\subseteq
                (\consistency_{\sfrac{X+Y}{\sim}}^{\hashtag}(\iota_{\sfrac{V_{X} + V_{Y}}{\sim}}([v])))^{c} = \consistency_{\sfrac{X+Y}{\sim}}(\iota_{\sfrac{V_{X} + V_{Y}}{\sim}}([v]))\\
              \end{align*}
              \else
              \begin{align*}
                &[[]_{V}^{\consistency},[]_{E}^{\consistency}]^{*}(\consistency_{X+Y}(\iota_{V_{X} + V_{Y}})(v))\\
                &\subseteq\\
                &\consistency_{\sfrac{X+Y}{\sim}}^{\hashtag}(\iota_{\sfrac{V_{X} + V_{Y}}{\sim}}([v]))\\
                &\subseteq\\
                &(\consistency_{\sfrac{X+Y}{\sim}}^{\hashtag}(\iota_{\sfrac{V_{X} + V_{Y}}{\sim}}([v])))^{c} = \consistency_{\sfrac{X+Y}{\sim}}(\iota_{\sfrac{V_{X} + V_{Y}}{\sim}}([v]))\\
              \end{align*}
              \fi
              by recalling that 
              \[
              [[]_{V}^{\consistency}, []_{E}^{\consistency}] = [[]_{V};\iota_{\sfrac{V_{X} + V_{Y}}{\sim}}, []_{E};\iota_{\sfrac{E_{X} + E_{Y}}{\sim}}]
              \].
    \end{enumerate}

    The cases for edges are analogous.

    So far we have shown that relations $<_{\sfrac{X+Y}{\sim}}$ and $\consistency_{\sfrac{X+Y}{\sim}}$ are well-defined and $([]_{V},[]_{E})$ is a homomorphism.
    By the Proposition~\ref{prop:pushout_is_e_hypergraph} the constructed $\sfrac{X+Y}{\sim}$ is an e-hypergraph.
    To show that it is a pushout we need to check the universal property.
    Let $u_{V}([\iota_{1,V}(x)]) = j_{1,V}(x)$ and $u_{V}([\iota_{2,V}(y)]) = j_2(y)$.
    We will need to check that this a well-defined function as well as that it is indeed a homomorphism.

    Let's say that $v_1 \sim v_2$ then we need to show that $u_{V}([v_1]) = u_{V}([v_2])$.
    $v_1 \sim v_2$ means there is a sequence $w = (x_1, \ldots, x_n)$ such that $v_1 = x_1$, $v_2 = x_n$ and for all $i < n$ $x_i \sim_{S} x_{i+1}$. We will proceed by induction on the length of $w$.
    \begin{itemize}
        \item $|w| = 2$, then $v_1 \sim_{S} v_2$ and there exists $z \in V_{Z}$ such that $v_1 = f_{V};\iota_{1,V})(z)$ and $v_2 = g_{V};\iota_{2,V}(z)$ and we have
        \[
            u_{V}([v_1]) = u_{V}([f_{V};\iota_{1,V}(z)]) = j_1({f_{V}(z)})
        \] 
        and 
        \[
            u_{V}([v_2]) = u_{V}([g_{V};\iota_{2,V}(z)]) = j_{2}(g_{V}(z))
        \] 
        which are equal by commutativity.
        \item Assume if $v_1 \sim v_2$ via $w : |w| \leq n$ then $u_{V}([v_1]) = u_{V}([v_2])$.
        \item Suppose $v_1 \sim v_2$ via $w = (x_1, \ldots, x_n, x_{n+1})$.
              By hypothesis, we know that $u_{V}([v_1]) = u_{V}([x_1]) = u_{V}([x_n])$ and that $u_{V}([x_n]) = u_{V}([x_{n+1}]) = u_{V}([v_2])$ and hence $u_{V}([v_1]) = u_{V}([v_2])$.
    \end{itemize}
    Because $v_1 \sim t_2$, it means that there exists $z \in V_{Z}$ such that $t_1 = f_{V};\iota_{1,V}(z)$ and $t_2 = g_{V};\iota_{2,V}(z)$.
    Then, we have
    \[
    u_{V}([t_1]) = u_{V}([f_{V};\iota_{1,V}(z)]) = j_1({f_{V}(z)})
    \] and 
    \[
    u_{V}([t_2]) = u_{V}([g_{V};\iota_{2,V}(z)]) = j_{2}(g_{V}(z))
    \] which are equal by commutativity.

    Let's show that such $(u_{V},u_{E})$ is unique.
    Assume that there is $u'_{V}$ such that $j_{1,V}(v) = u'_{V}([\iota_{1,V}(v)])$ and $j_{2,V}(v) = u'_{V}([\iota_{2,V}(v)])$.
    Then we have $u_{V}[\iota_{1,V}(v)] = j_{1,V}(v) = u'_{V}([\iota_{1,V}(v)])$ for all $v \in V_{X}$ (similarly for $v \in V_{Y}$) and hence $u_{V} = u'_{V}$.
    Similarly for $u_{E}$.

    Let's now check if $u = (u_{V},u_{E})$ is a homomorphism.
    \begin{enumerate}
        \item The first condition to check is if $u_{V}^{*}(s_{\sfrac{X+Y}{\sim}}([e])) = s_{Q}(u_{E}([e]))$. 
              By construction of $s_{\sfrac{X+Y}{\sim}}$ we have
              \[
                u_{V}^{*}(s_{\sfrac{X+Y}{\sim}}([e])) = u_{V}^{*}([s_{X+Y}(e)]^{*})
              \]
              Suppose $e = \iota_{E_{X}}(e')$ and then
              \begin{align*}
                u_{V}^{*}([s_{X+Y}(e)]^{*}) &= u_{V}^{*}([s_{X+Y}(\iota_{1,E}(e'))]^{*})\\
                                            &= u_{V}^{*}([(\iota_{1,V}^{*}(s_{X}(e')))]^{*})
              \end{align*}
              Since $j_1$ and $j_2$ are homomorphisms, it must be the case that $s_{Q}(j_{1,E}(e')) = j_{1,V}^{*}(s_{X}(e'))$.
              By definition, 
              \begin{align*}                
                u_{V}^{*}([(\iota_{1,V}^{*}(s_{X}(e')))]^{*}) &= j_{1,V}^{*}((s_{X}(e')))\\ 
                &= s_{Q}(j_{1,E}(e'))\\
                &= s_{Q}(u_{E}([\iota_{1,E}(e')]))\\
                &= s_{Q}(u_{E}([e]))
            \end{align*}
            Similarly for $e = \iota_{E_{Y}}(e')$.
        \item The case for targets is analogous.
        \item The next condition to check is if $[[v]) \not = \varnothing$ then 
    \[
        <^{\mu}_{Q}(\iota_{V_{Q}}(u_{V}([v]))) = u_{E}(<^{\mu}_{\sfrac{X+Y}{\sim}}(\iota_{\sfrac{V_{X} + V_{Y}}{\sim}}[v]))
    \]
    We need to consider a few cases.
    \begin{itemize}
        \item Suppose $<_{\sfrac{X+Y}{\sim}}^{\mu}(\iota_{\sfrac{V_{X} + V_{Y}}{\sim}}([v])) = [<_{X+Y}^{\mu}(\iota_{V_{X} + V_{Y}}(u))]$ for $u \sim v$.
              Then assume $u = \iota_{1,V}(u_{x})$ and then 
              \begin{align*}
                [<_{X+Y}^{\mu}(\iota_{V_{X} + V_{Y}}(u))] &= [<_{X+Y}^{\mu}(\iota_{V_{X} + V_{Y}}(\iota_{1,V}(u_{x})))]\\
                &= [\iota_{1,E}(<_{X}^{\mu}(\iota_{V_{X}}(u_{x})))]
            \end{align*}
            and
            \begin{align*}
                u_{E}([\iota_{1,E}(<_{X}^{\mu}(\iota_{V_{X}}(u_{x})))]) &= j_{1,E}(<_{X}^{\mu}(\iota_{V_{X}}(u_{x})))\\
                &= <_{Q}^{\mu}(\iota_{V_{Q}}(j_{1,V}(u_{x})))\\
                &= <_{Q}^{\mu}(\iota_{V_{Q}}(u([\iota_{1,V}(u_{x})])))\\
                &= <_{Q}^{\mu}(\iota_{V_{Q}}(u([v])))
            \end{align*}
            and we have
            \[
                <^{\mu}_{Q}(\iota_{V_{Q}}(u_{V}([v]))) = u_{E}(<^{\mu}_{\sfrac{X+Y}{\sim}}(\iota_{\sfrac{V_{X} + V_{Y}}{\sim}}[v]))
            \]
            The case when $u = \iota_{2,V}(u_{y})$ is symmetric.
        \item Suppose $<_{\sfrac{X+Y}{\sim}}^{\mu}(\iota_{\sfrac{V_{X} + V_{Y}}{\sim}}([v])) = <_{\sfrac{X+Y}{\sim}}^{\mu}(\iota_{\sfrac{V_{X} + V_{Y}}{\sim}}([v']))$ such that there is a path from $[v]$ to $[v']$ and $v$ has no pre-image in $V_{Z}$.
              By the argument above, we know that
              \[
                <^{\mu}_{Q}(\iota_{V_{Q}}(u_{V}([v']))) = u_{E}(<^{\mu}_{\sfrac{X+Y}{\sim}}(\iota_{\sfrac{V_{X} + V_{Y}}{\sim}}[v']))
              \]
              and by definition
              \[
                u_{E}(<^{\mu}_{\sfrac{X+Y}{\sim}}(\iota_{\sfrac{V_{X} + V_{Y}}{\sim}}[v'])) = u_{E}(<^{\mu}_{\sfrac{X+Y}{\sim}}(\iota_{\sfrac{V_{X} + V_{Y}}{\sim}}[v]))
              \]
              since $(u_{V},u_{E})$ preserves sources and targets and there is a path from $[v]$ to $[v']$, there is a path from $u_{V}([v])$ to $u_{V}([v'])$ and because $Q$ is an e-hypergraph
              \[
                <^{\mu}_{Q}(\iota_{V_{Q}}(u_{V}([v']))) = <^{\mu}_{Q}(\iota_{V_{Q}}(u_{V}([v])))
              \]
              and finally
              \[
                <^{\mu}_{Q}(\iota_{V_{Q}}(u_{V}([v]))) = u_{E}(<^{\mu}_{\sfrac{X+Y}{\sim}}(\iota_{\sfrac{V_{X} + V_{Y}}{\sim}}[v]))
              \]
        \end{itemize}
    The cases for edges are analogous.
    \item The last condition to check is if
    \begin{align*}
    &[u_{V};\iota_{V_{Q}}, u_{E};\iota_{E_{Q}}]^{*}(\consistency_{\sfrac{X+Y}{\sim}}(\iota_{\sfrac{V_{X} + V_{Y}}{\sim}}([v])))\\
    &\;\subseteq\\
    &\consistency_{Q}(u_{V};\iota_{V_{Q}}([v]))
    \end{align*}
    Recall that
    \[
      \consistency_{\sfrac{X+Y}{\sim}}(\iota_{\sfrac{V_{X} + V_{Y}}{\sim}}([v])) = (\consistency^{\hashtag}_{\sfrac{X+Y}{\sim}}(\iota_{\sfrac{V_{X} + V_{Y}}{\sim}}([v])))^{c}
    \]
    and first we will check if
    \begin{align*}
      &[u_{V};\iota_{V_{Q}}, u_{E};\iota_{E_{Q}}]^{*}(\consistency_{\sfrac{X+Y}{\sim}}^{\hashtag}(\iota_{\sfrac{V_{X} + V_{Y}}{\sim}}([v])))\\
      &\;\subseteq\\
      &\consistency_{Q}(u_{V};\iota_{V_{Q}}([v]))
    \end{align*}
    Consider cases
    \begin{itemize}
        \item If $\consistency_{\sfrac{X+Y}{\sim}}^{\hashtag}(\iota_{\sfrac{V_{X} + V_{Y}}{\sim}}([v])) = \varnothing$ then the inclusion is trivial and hence we assume that
        \[
            \consistency_{\sfrac{X+Y}{\sim}}^{\hashtag}(\iota_{\sfrac{V_{X} + V_{Y}}{\sim}}([v])) \not = \varnothing
        \]
        \item Suppose $\consistency_{X+Y}(\iota_{V_{X} + V_{Y}}(u)) \not = \varnothing$ and $v \sim u$ and then
        \[
            \consistency_{\sfrac{X+Y}{\sim}}^{\hashtag}(\iota_{\sfrac{V_{X} + V_{Y}}{\sim}}([v])) = [[]_{V}^{\consistency},[]_{E}^{\consistency}]^{*}(\consistency_{X+Y}(\iota_{V_{X}+V_{Y}}(u)))
        \]
        Assume $u = \iota_{1,V}(u_{x})$, then
        \ifdefined \ONECOLUMN
        \begin{align*}
            [[]_{V}^{\consistency},[]_{E}^{\consistency}]^{*}(\consistency_{X+Y}(\iota_{V_{X}+V_{Y}}(u)))
            &=
            [[]_{V}^{\consistency},[]_{E}^{\consistency}]^{*}(\consistency_{X+Y}(\iota_{V_{X}+V_{Y}}(\iota_{1,V}(u_{x}))))\\
            &=
            [\iota_{1,V};[]_{V}^{\consistency},\iota_{1,E};[]_{E}^{\consistency}]^{*}(\consistency_{X}(\iota_{V_{X}}(u_{x})))
        \end{align*}
        \else
        \begin{align*}
            &[[]_{V}^{\consistency},[]_{E}^{\consistency}]^{*}(\consistency_{X+Y}(\iota_{V_{X}+V_{Y}}(u)))\\
            &=\\
            &[[]_{V}^{\consistency},[]_{E}^{\consistency}]^{*}(\consistency_{X+Y}(\iota_{V_{X}+V_{Y}}(\iota_{1,V}(u_{x}))))\\
            &=\\
            &[\iota_{1,V};[]_{V}^{\consistency},\iota_{1,E};[]_{E}^{\consistency}]^{*}(\consistency_{X}(\iota_{V_{X}}(u_{x})))
        \end{align*}
        \fi

        and
        \ifdefined \ONECOLUMN
        \begin{align*}
            [u_{V};\iota_{V_{Q}}, u_{E};\iota_{E_{Q}}]^{*}(\consistency_{\sfrac{X+Y}{\sim}}^{\hashtag}(\iota_{\sfrac{V_{X} + V_{Y}}{\sim}}([v])))
            &=
            [\iota_{1,V};[]_{V};u_{V};\iota_{V_{Q}},\iota_{1,E};[]_{E};u_{E};\iota_{Q_{E}}]^{*}(\consistency_{X}(\iota_{V_{X}}(u_{x})))\\
            &=
            [j_{1,V};\iota_{V_{Q}},j_{1,E};\iota_{Q_{E}}]^{*}(\consistency_{X}(\iota_{V_{X}}(u_{x})))\\
            &\subseteq
            \consistency_{Q}(\iota_{V_{Q}}(j_{1,V}(u_{x})))\\
            &=
            \consistency_{Q}(\iota_{V_{Q}}(u_{V}[\iota_{1,V}(u_{x})]))\\
            &=
            \consistency_{Q}(\iota_{V_{Q}}(u_{V}[u]))\\
            &=
            \consistency_{Q}(\iota_{V_{Q}}(u_{V}[v]))
        \end{align*}
        \else
        \begin{align*}
            &[u_{V};\iota_{V_{Q}}, u_{E};\iota_{E_{Q}}]^{*}(\consistency_{\sfrac{X+Y}{\sim}}^{\hashtag}(\iota_{\sfrac{V_{X} + V_{Y}}{\sim}}([v])))\\
            &=\\
            &[\iota_{1,V};[]_{V};u_{V};\iota_{V_{Q}},\iota_{1,E};[]_{E};u_{E};\iota_{Q_{E}}]^{*}(\consistency_{X}(\iota_{V_{X}}(u_{x})))\\
            &=\\
            &[j_{1,V};\iota_{V_{Q}},j_{1,E};\iota_{Q_{E}}]^{*}(\consistency_{X}(\iota_{V_{X}}(u_{x})))\\
            &\subseteq\\
            &\consistency_{Q}(\iota_{V_{Q}}(j_{1,V}(u_{x})))\\
            &=\\
            &\consistency_{Q}(\iota_{V_{Q}}(u_{V}[\iota_{1,V}(u_{x})]))\\
            &=\\
            &\consistency_{Q}(\iota_{V_{Q}}(u_{V}[u]))\\
            &=\\
            &\consistency_{Q}(\iota_{V_{Q}}(u_{V}[v]))
        \end{align*}
        \fi
        The case when $u = \iota_{2,V}(u_{y})$ is symmetric.
        \item Suppose that $v$ has no pre-image in $V_{Z}$ and $\consistency_{X+Y}(\iota_{V_{X} + V_{Y}}(v)) = \varnothing$ and there exists $u$ such that $\consistency_{X+Y}(\iota_{V_{X} + V_{Y}}(u)) \not = \varnothing$ and such that there is an undirected path from $[v]$ to $[u]$ and then
        \[
            \consistency_{\sfrac{X+Y}{\sim}}^{\hashtag}(\iota_{\sfrac{V_{X} + V_{Y}}{\sim}}([v])) = \consistency_{\sfrac{X+Y}{\sim}}^{\hashtag}(\iota_{\sfrac{V_{X} + V_{Y}}{\sim}}([u]))    
        \]
        By the argument above it follows that
        \begin{align*}
            &[u_{V};\iota_{Q_{V}},u_{E};\iota_{Q_{E}}]^{*}(\consistency_{\sfrac{X+Y}{\sim}}^{\hashtag}(\iota_{\sfrac{V_{X} + V_{Y}}{\sim}}([u])))\\
            &\;\subseteq\\
            &\consistency_{Q}(\iota_{V_{Q}}(u_{V}([u])))
        \end{align*}
        \end{itemize}
       Since $(u_{V},u_{E})$ preserves paths, there is a path from $u_{V}([u])$ to $u_{V}([v])$ and since $Q$ is an e-hypergraph
       \ifdefined \ONECOLUMN
       \begin{align*}
        \consistency_{\sfrac{X+Y}{\sim}}^{\hashtag}(\iota_{\sfrac{V_{X} + V_{Y}}{\sim}}([v]))
        &=
        \consistency_{\sfrac{X+Y}{\sim}}^{\hashtag}(\iota_{\sfrac{V_{X} + V_{Y}}{\sim}}([u]))\\
        &\subseteq
        \consistency_{Q}(\iota_{V_{Q}}(u_{V}([u])))\\
        &=
        &\consistency_{Q}(\iota_{V_{Q}}(u_{V}([v])))
       \end{align*}
       \else
       \begin{align*}
        &\consistency_{\sfrac{X+Y}{\sim}}^{\hashtag}(\iota_{\sfrac{V_{X} + V_{Y}}{\sim}}([v]))\\
        &=\\
        &\consistency_{\sfrac{X+Y}{\sim}}^{\hashtag}(\iota_{\sfrac{V_{X} + V_{Y}}{\sim}}([u]))\\
        &\subseteq\\
        &\consistency_{Q}(\iota_{V_{Q}}(u_{V}([u])))\\
        &=\\
        &\consistency_{Q}(\iota_{V_{Q}}(u_{V}([v])))
    \end{align*}
    \fi
       So far we have shown that
       \begin{align*}
        &[u_{V};\iota_{V_{Q}}, u_{E};\iota_{E_{Q}}]^{*}(\consistency_{\sfrac{X+Y}{\sim}}^{\hashtag}(\iota_{\sfrac{V_{X} + V_{Y}}{\sim}}([v])))\\
        &\;\subseteq\\
        & \consistency_{Q}(u_{V};\iota_{V_{Q}}([v]))
       \end{align*}
       and
       \begin{align*}
        &[u_{V};\iota_{V_{Q}}, u_{E};\iota_{E_{Q}}]^{*}(\consistency_{\sfrac{X+Y}{\sim}}^{\hashtag}(\iota_{\sfrac{E_{X} + E_{Y}}{\sim}}([e])))\\
        &\;\subseteq\\
        &\consistency_{Q}(u_{E};\iota_{E_{Q}}([e]))
       \end{align*} follows by the same argument.
       Then, to check if
       \begin{align*}
        &[u_{V};\iota_{V_{Q}}, u_{E};\iota_{E_{Q}}]^{*}(\consistency_{\sfrac{X+Y}{\sim}}(\iota_{\sfrac{V_{X} + V_{Y}}{\sim}}([v])))\\
        &\; \subseteq\\
        &\consistency_{Q}(u_{V};\iota_{V_{Q}}([v]))
       \end{align*}
       we need to verify if
       \begin{itemize}
        \item $[u_{V};\iota_{V_{Q}},u_{E};\iota_{E_{Q}}]^{*}(\{[v]\}) \subseteq \consistency_{Q}(u_{V};\iota_{V_{Q}}([v]))$.
              This holds because $\consistency_{Q}$ is reflexive.
        \item $\iota_{\sfrac{V_{X} + V_{Y}}{\sim}}([u]) \in \consistency_{\sfrac{X+Y}{\sim}}^{\hashtag}(\iota_{\sfrac{V_{X} + V_{Y}}{\sim}}([v]))$ then 
        \[
            [u_{V};\iota_{V_{Q}}, u_{E};\iota_{E_{Q}}]^{*}(\{[v]\}) \in \consistency_{Q}(u_{V};\iota_{V_{Q}}([u]))
        \]
              The first bit implies that $u_{V};\iota_{V_{Q}}([u]) \in \consistency_{Q}(u_{V};\iota_{V_{Q}}([v]))$ and since $\consistency_{Q}$ is symmetric it must be the case that $u_{V_{Q}};\iota_{V_{Q}}[v] \in \consistency_{Q}(u_{V};\iota_{V_{Q}}([u]))$
        \item There exists a sequence $w = (x_1, \ldots, x_n)$ such that 
        \[
            \iota_{\sfrac{V_{X} + V_{Y}}{\sim}}([x_i]) \in \consistency_{\sfrac{X+Y}{\sim}}^{\hashtag}(\iota_{\sfrac{V_{X} + V_{Y}}{\sim}}([x_{i+1}]))
        \] 
        
        or 
        \[
            \iota_{\sfrac{V_{X} + V_{Y}}{\sim}}([x_{i+1}]) \in \consistency_{\sfrac{X+Y}{\sim}}^{\hashtag}(\iota_{\sfrac{V_{X} + V_{Y}}{\sim}}([x_{i}]))
        \]
        for $i < n$ (or similarly for edges), then $u_{V};\iota_{V_{Q}}([x_1]) \in \consistency_{Q}(u_{V};\iota_{V_{Q}}([x_n]))$.
        \begin{itemize}
            \item Base case is when $w = (x_1,x_2)$.
                  Then the statement is trivial as 
                  \begin{align*}
                  &[u_{V};\iota_{V_{Q}},u_{E};\iota_{E_{Q}}]^{*}(\consistency_{\sfrac{X+Y}{\sim}}^{\hashtag}(\iota_{\sfrac{V_{X} + V_{Y}}{\sim}}([x_2])))\\
                  &\subseteq\\
                  &\consistency_{Q}(u_{V};\iota_{V_{Q}}([x_2]))
                  \end{align*}
            \item Suppose the statement is true for all $w$ such that $|w| \leq n$.
            \item Consider a sequence $w = (x_1, \ldots, x_{n}, x_{n+1})$.
                  By assumption, we know that $u_{V};\iota_{V_{Q}}([x_1]) \in \consistency_{Q}(u_{V};\iota_{V_{Q}}([x_n]))$.
                  The fact that $\iota_{\sfrac{V_{X} + V_{Y}}{\sim}}([x_{n}]) \in \consistency_{\sfrac{X+Y}{\sim}}^{\hashtag}(\iota_{\sfrac{V_{X} + V_{Y}}{\sim}}([x_{n+1}]))$ (or its symmetric counterpart) implies that $u_{V};\iota_{V_{Q}}([x_{n}]) \in \consistency_{Q}(u_{V};\iota_{V_{Q}}([x_{n+1}]))$ and by the transitivity and symmetry of $\consistency_{Q}$ it follows that
                \[
                    u_{V};\iota_{V_{Q}}([x_1]) \in \consistency_{Q}(u_{V};\iota_{V_{Q}}([x_{n+1}]))
                \]
                and
                \[
                    u_{V};\iota_{V_{Q}}([x_{n+1}]) \in \consistency_{Q}(u_{V};\iota_{V_{Q}}([x_{1}]))
                \]
        \end{itemize}
       \end{itemize}
       Therefore, we have
       \begin{align*}
        &[u_{V};\iota_{V_{Q}}, u_{E};\iota_{E_{Q}}]^{*}(\consistency_{\sfrac{X+Y}{\sim}}(\iota_{\sfrac{V_{X} + V_{Y}}{\sim}}([v])))\\
        &\;\subseteq\\
        &\consistency_{Q}(u_{V};\iota_{V_{Q}}([v]))
       \end{align*}
       and similarly for edges. Which concludes that $(u_{V},u_{E})$ is a homomorphism.
    \end{enumerate}
\end{proof}

\begin{remark}
    The first assumption in~\ref{pushout:assumptions} can be weakened by allowing $E_{Z}$ to contain unlabelled edges with no inputs and outputs such that:
    \begin{itemize}
        \item $[f_{E}(e_{i})) = [f_{E}(e_{j}))$ and $[g_{E}(e_{i})) = [g_{E}(e_{j}))$ for all $e_{i},e_{j} \in E_{Z}$ 
        \item If $[f_{E}(e)) \not = \varnothing$ then $[g_{E}(e)) = \varnothing$ and if $[g_{E}(e)) \not = \varnothing$ then $[f_{E}(e)) = \varnothing$.
        \item $\consistency(f_{E}(e_{i})) = \consistency(f_{E}(e_{j}))$ and $\consistency(g_{E}(e_i)) = \consistency(g_{E}(e_j))$ for all $e_i,e_j$ in $E_{Z}$.
    \end{itemize}
\end{remark}

\begin{remark}
    If in a span $Y \xleftarrow{g} Z \xrightarrow{f} X$ of morphisms in $\catname{EHyp}(\Sigma)$ the morphisms $f$ and $g$ are monos, then the arrows $f'$ and $g'$ of the pushout $Y \xrightarrow{g'} X +_{f,g} Y \xleftarrow{f'} X$ are also monos.
\end{remark}
This is because $[]_{V}$ identifies two vertices $v_1$ and $v_2$ only if there exists a vertex $z$ in $V_{Z}$ such that $f_{V};\iota_{1,V}(z) = v_1$ and $g_{V};\iota_{2,V}(z) = v_2$ and therefore $f' = \iota_{1};[]$ and $g' = \iota_{2};[]$ are monos.
Similarly for edges.

\begin{figure}
    \[
        \scalebox{0.5}{
        \input{./figures/pushout_example.tikz}
        }
\]
\captionsetup{belowskip=-3ex}
\caption{Pushout in $\catname{EHyp}(\Sigma)$}
\label{fig:pushout_example}
\end{figure}

\begin{example}
    Consider a diagram in Figure~\ref{fig:pushout_example} where the subscripts define the corresponding morphisms.
    E-hypergraph $Q$ is not a pushout as there is no homomorphism $u$ that would complete the diagram.
    In particular, there is no $(u_{V},u_{E})$ such that 
    \[
        [u_{V};\iota_{P_{V}},u_{E};\iota_{P_{E}}]^{*}(\consistency_{Q}(\iota_{V_{Q}}(q_1))) \subseteq \consistency_{P}(u_{V};\iota_{P_{V}}(q_1))
    \]
    If we let $u_{V}(i_{1,V}(v)) = j_{1,V}(v)$ for all $v \in V_{X}$ and $u_{V}(i_{2,V}(v)) = j_{2,V}(v)$ and similarly for edges.
    Then 
    \begin{align*}
        &[u_{V};\iota_{P_{V}}, u_{E};\iota_{P_{E}}]^{*}(\consistency_{Q}(\iota_{V_{Q}}(q_1)))\\
        &=\\
        &[u_{V};\iota_{P_{V}}, u_{E};\iota_{P_{E}}]^{*}(\{q_1, q_2, q_3, q_4, f, h, g\})\\
        &=\\
        &[\iota_{P_{V}}(j_{1,V}(x_1)), \iota_{P_{V}}(j_{1,V}(x_2)), \iota_{P_{V}}(j_{1,V}(x_3)),\\
         &\; \iota_{P_{V}}(j_{1,V}(x_4)), \iota_{P_{E}}(j_{1,E}(f)), \iota_{P_{E}}(j_{1,E}(g)), \iota_{P_{V}}(j_{2,V}(y_1)),\\
         &\; \iota_{P_{V}}(j_{2,V}(y_2)), \iota_{P_{E}}(j_{2,E}(h))]\\
        &=\\
        &[p_1, p_2, p_3, p_4, f, h, g]\\
        &\not \subseteq\\
        &\consistency_{P}(u_{V};\iota_{P_{V}}(q_1))\\
        &=\\
        &\consistency_{P}(j_{1,V};\iota_{P_{V}}(x_1))\\
        &=\\
        &\consistency_{P}(\iota_{P_{V}}(p_1))\\
        &=\\
        &[p_1, f, h, p_2]
    \end{align*}
\end{example}

% So far we have shown that the constructed $\sfrac{X+Y}{\sim}$ is a pushout but we are yet to show that it is an e-hypergraph.
% The proof of the latter will follow below and the following corollaries will come in handy.
% We will then proceed by showing that both relations satisfy the properties as defined in Definition~\ref{def:e-homo}.

% By construction we have the following corollaries for $<_{\sfrac{X+Y}{\sim}}$ and $\consistency_{\sfrac{X+Y}{\sim}}$.

\begin{lemma}
\label{lemma:child_irreflexive}
    $<_{\sfrac{X+Y}{\sim}}^{\mu}$ is irreflexive.
\end{lemma}
\begin{proof}
    Suppose there exist $[e]$ such that $[e] = <_{\sfrac{X+Y}{\sim}}^{\mu}(\iota_{\sfrac{E_{X} + E_{Y}}{\sim}}([e]))$.
    Then we have two cases.
    \begin{itemize}
      \item $<_{\sfrac{X+Y}{\sim}}^{\mu}(\iota_{\sfrac{E_{X} + E_{Y}}{\sim}}([e])) = [<_{X+Y}^{\mu}(\iota_{E_{X} + E_{Y}}(e))] = [e]$
      This would mean that $<_{X+Y}^{\mu}(\iota_{E_{X} + E_{Y}}(e)) = e$ as $[]_{E}$ is injective, which would mean that the relation $<_{X+Y}^{\mu}$ is not irreflexive because $e$ is the predecessor of $e$.
      \item Suppose $<_{\sfrac{X+Y}{\sim}}^{\mu}(\iota_{\sfrac{E_{X} + E_{Y}}{\sim}}([e])) = <_{\sfrac{X+Y}{\sim}}^{\mu}(\iota_{\sfrac{V_{X} + V_{Y}}{\sim}}([v])) = [<_{X+Y}^{\mu}(\iota_{V_{X} + V_{Y}}(v))]$ such that $<_{X+Y}^{\mu}(\iota_{V_{X} + V_{Y}}(v))$ is defined and there is a path from $[e]$ to $[v]$.
            This implies
            \[
                [e] = [<_{X+Y}^{\mu}(\iota_{V_{X} + V_{Y}}(v))]
            \]
            and
            \[
               e = <_{X+Y}^{\mu}(\iota_{V_{X} + V_{Y}}(v))
            \]
            The existence of a path from $[e]$ to $[v]$ means there is a path from $e$ to $v'$ such that $v' \sim v$.
            \begin{itemize}
                \item If $v' = v$, then there is a path from $e$ to $v$ which contradicts $e = <_{X+Y}^{\mu}(\iota_{V_{X} + V_{Y}}(v))$
                \item Consider the case when $v' \sim v$.
                      $e = <_{X+Y}^{\mu}(\iota_{V_{X} + V_{Y}}(v))$ implies that both $e$ and $v$ are in the image of $(\iota_{1,V},\iota_{1,E})$ or $(\iota_{2,V},\iota_{2,E})$.
                      Suppose $e = \iota_{1,E}(e_{x})$ and $v = \iota_{1,V}(v_{x})$.
                      Since $v' \sim v$, $v_{x} = f_{V}(z_1)$. 
                      And since there is a path from $v'$ to $e$, $v' = \iota_{1,V}(v'_{x})$ and $v'_{x} = f_{V}(z_2)$.
                      This further implies
                      \begin{align*}
                        <_{X+Y}^{\mu}(\iota_{V_{X} + V_{Y}}(v)) &= \iota_{1,E}(<_{X}^{\mu}(\iota_{V_{X}}(v_{x})))\\
                                                                &= \iota_{1,E}(<_{X}^{\mu}(\iota_{V_{X}}(f_{V}(z_1))))\\
                                                                &= \iota_{1,E}(<_{X}^{\mu}(\iota_{V_{X}}(f_{V}(z_2))))\\
                                                                &= \iota_{1,E}(e_{x})
                    \end{align*}
                      and
                      \[
                        e_{x} = <_{X}(\iota_{V_{X}}(f_{V}(z_2))) = <_{X}(\iota_{V_{X}}(v_{x}'))
                      \]
                      which contradicts that there is a path from $e_{x}$ to $v_{x}'$.
                      The case when $e = \iota_{2,E}(e_{y})$ and $v = \iota_{2,V}(v_{y})$ is symmetric.
            \end{itemize}
    \end{itemize}                  
      We do not need to check this requirement for vertices because a vertex can not be a predecessor.
\end{proof}

\begin{lemma}
\label{lemma:child_assymetric}
    $<_{\sfrac{X+Y}{\sim}}^{\mu}$ is asymmetric.
\end{lemma}
\begin{proof}
    Let $e_1,e_2$ be edges. 
    Then, if $[e_1] = <_{\sfrac{X+Y}{\sim}}^{\mu}(\iota_{\sfrac{E_{X} + E_{Y}}{\sim}}([e_2]))$ it must not be the case that $[e_2] = <_{\sfrac{X+Y}{\sim}}^{\mu}(\iota_{\sfrac{E_{X} + E_{Y}}{\sim}}([e_1]))$.
    For $[e_1] = <_{\sfrac{X+Y}{\sim}}^{\mu}(\iota_{\sfrac{E_{X} + E_{Y}}{\sim}}([e_2]))$ we have two cases.
    \begin{itemize}
      \item 
      \[
        [e_1] = [<_{X+Y}^{\mu}(\iota_{E_{X} + E_{Y}}(e_2))]
      \]
      which implies $e_1 = <_{X+Y}^{\mu}(\iota_{E_{X} + E_{Y}}(e_2))$.
      Having at the same time $[e_2] = <_{\sfrac{X+Y}{\sim}}^{\mu}(\iota_{\sfrac{E_{X} + E_{Y}}{\sim}}([e_1]))$ gives us also two cases.
      \begin{itemize}
        \item $[e_2] = [<_{X+Y}^{\mu}(\iota_{E_{X} + E_{Y}}(e_1))]$ and $e_2 = <_{X+Y}^{\mu}(\iota_{E_{X} + E_{Y}}(e_1))$ which implies that $e_1$ is the predecessor of $e_2$ and vice versa which contradicts e-hypergraphness of $X+Y$.
        \item
            \[
                [e_2] = [<_{X+Y}^{\mu}(\iota_{V_{X} + V_{Y}}(v))]
            \]
            such that $[e_1) = \varnothing$ and there is a path from $[e_1]$ to $[v]$, i.e. a path from $e_1$ to $v'$ such that $v' \sim v$.
            \begin{itemize}
                \item $v' \sim v$ and there is a path from $e_1$ to $v'$. 
                Suppose $e_2 = \iota_{1,E}(e_{2,x})$ and $v = \iota_{1,V}(v_{x})$.
                Since $e_1$ is the predecessor of $e_2$, the former should also be in the image of $\iota_{1,E}$, i.e. $e_1 = \iota_{1,E}(e_{1,x})$, and since there is a path from $e_1$ to $v'$, $v' = \iota_{1,V}(v'_{x})$.
                $v' \sim v$ implies $v_{x} = f_{V}(z_1)$ and $v'_{x} = f_{V}(z_2)$.
                Then,
                \begin{align*}
                    <_{X+Y}^{\mu}(\iota_{V_{X} + V_{Y}}(v)) &= \iota_{1,E}(<_{X}^{\mu}(\iota_{V_{X}}(v_{x})))\\
                                                            &= \iota_{1,E}(<_{X}^{\mu}(\iota_{V_{X}}(f_{V}(z_1))))\\
                                                            &= \iota_{1,E}(<_{X}^{\mu}(\iota_{V_{X}}(f_{V}(z_2))))\\
                                                            &= \iota_{1,E}(<_{X}^{\mu}(\iota_{V_{X}}(v_{x}')))\\
                                                            &= <_{X+Y}^{\mu}(\iota_{V_{X} + V_{Y}}(v'))\\
                                                            &= <_{X+Y}^{\mu}(\iota_{E_{X} + E_{Y}}(e_1))
                \end{align*}
                which contradicts $e_1 = <_{X+Y}^{\mu}(\iota_{E_{X} + E_{Y}}(e_2))$.
                The case when $e_2 = \iota_{2,E}(e_{2,y})$ and $v = \iota_{2,V}(v_{y})$ is symmetric.
            \end{itemize}
        \end{itemize}
            \item The other case is when          
            \[
                [e_1] = [<_{X+Y}^{\mu}(\iota_{V_{X} + V_{Y}}(v_1))]
            \]
            such that there is a path from $e_2$ to $v'_1$ and $v_1 \sim v'_1$ and $[e_2) = \varnothing$.
            \begin{itemize}
                \item The case when $[e_2] = [<_{X+Y}^{\mu}(\iota_{E_{X} + E_{Y}}(e_1))]$ is analogous to the last case above.
                \item Suppose 
                \[
                    [e_2] = [<_{X+Y}^{\mu}(\iota_{V_{X} + V_{Y}}(v_2))]
                \] such that there is a path from $e_1$ to $v'_2$ and $v_2 \sim v'_2$ and $[e_1) = \varnothing$.
                    Suppose $e_1 = \iota_{1,E}(e_{1,x})$ and $v_1 = \iota_{1,V}(v_{1,x})$ which implies $v'_2 = \iota_{1,V}(v'_{2,x})$.
                    Since $v_1 \sim v_1'$, $v_{1,x} = f_{V}(z_1)$ and similarly $v'_{2,x} = f_{V}(z_2)$.
                    \begin{align*}
                        e_1 &= <_{X+Y}^{\mu}(\iota_{V_{X} + V_{Y}}(v_1))\\
                            &= \iota_{1,E}(<_{X}^{\mu}(\iota_{V_{X}}(v_{1,x})))\\
                            &= \iota_{1,E}(<_{X}^{\mu}(\iota_{V_{X}}(f_{V}(z_1))))\\
                            &= \iota_{1,E}(<_{X}^{\mu}(\iota_{V_{X}}(f_{V}(z_2))))\\
                            &= \iota_{1,E}(<_{X}^{\mu}(\iota_{V_{X}}(v'_{2,x})))\\
                            &= <_{X+Y}^{\mu}(\iota_{V_{X} + V_{Y}}(v'_2))\\
                            &= <_{X+Y}^{\mu}(\iota_{E_{X} + E_{Y}}(e_1))
                    \end{align*}
                    And we get a contradiction to irreflexivity of $<_{X+Y}^{\mu}$.
            \end{itemize}
      \end{itemize} 
    Once again, we do not need to check for vertices as a vertex can not be on the left-hand side of the relation.
\end{proof}

\begin{lemma}
\label{lemma:child_respect}
    $<_{\sfrac{X+Y}{\sim}}^{\mu}$ respects connectivity.
\end{lemma}
\begin{proof}
        We need to show that if $[v] \in s_{\sfrac{X+Y}{\sim}}([e_1])$ and $[e_2] = <_{\sfrac{X+Y}{\sim}}^{\mu}(\iota_{\sfrac{E_{X} + E_{Y}}{\sim}}([e_1]))$ then $[e_2] = <_{\sfrac{X+Y}{\sim}}^{\mu}(\iota_{\sfrac{V_{X} + V_{Y}}{\sim}}[v])$.
        We have two cases.
        \begin{itemize}
            \item \[
            [e_2] = [<_{X+Y}^{\mu}(\iota_{V_{X} + V_{Y}}(e_1))]
            \]
            By definition, $s_{\sfrac{X+Y}{\sim}}([e_1]) = [s_{X+Y}(e_1)]^{*}$ and $[v] \in [s_{X+Y}(e_1)]^{*}$ means there exists $v' \sim v$ such that $v' \in s_{X+Y}(e_1)$.
            \begin{itemize}
                \item $v' = v$ and then $v \in s_{X+Y}(e_1)$ and therefore $e_2 = <_{X+Y}^{\mu}(\iota_{V_{X} + V_{Y}}(v))$ and $[e_2] = [<_{X+Y}^{\mu}(\iota_{V_{X} + V_{Y}}(v))]$.
                \item $v' \sim v$ and then 
                \begin{align*}
                    [e_2] &= [<_{X+Y}^{\mu}(\iota_{V_{X} + V_{Y}}(v'))]\\
                          &= <_{\sfrac{X+Y}{\sim}}^{\mu}(\iota_{\sfrac{V_{X} + V_{Y}}{\sim}}([v']))\\
                          &= <_{\sfrac{X+Y}{\sim}}^{\mu}(\iota_{\sfrac{V_{X} + V_{Y}}{\sim}}([v]))
                \end{align*}
            \end{itemize}
            \item \[
                [e_2] = [<_{X+Y}^{\mu}(\iota_{V_{X} + V_{Y}}(v_1))]    
            \]
            such that there is $v_1' \sim v_1$ and there is a path from $e_1$ to $v_1'$ and $[e_1) = \varnothing$.
            $[v] \in [s_{X+Y}(e_1)]^{*}$ means there exists $v' \sim v$ such that $v' \in s_{X+Y}(e_1)$.
            \begin{itemize}
                \item $v' = v$ and then $v \in s_{X+Y}(e_1)$. Since there is a path from $e_1$ to $v_1'$, there is also a path from $v'$ to $v_1'$ and by definition
                \[
                    <_{\sfrac{X+Y}{\sim}}^{\mu}(\iota_{\sfrac{V_{X} + V_{Y}}{\sim}}([v])) = <_{\sfrac{X+Y}{\sim}}^{\mu}(\iota_{\sfrac{V_{X} + V_{Y}}{\sim}}([v_1])) = [e_2]
                \]
                \item $v' \sim v$ and then there is a path from $v'_1$ to $v'$. 
                Suppose $v_1 = f_{V};\iota_{1,V}(z_1)$ and since $v_1' \in s_{X+Y}(e_1)$, $[e_1) = \varnothing$, both $v_1'$ and $v'$ should be in the image of $g_{V};\iota_{2,V}$.
                Since $[v_1) \not = \varnothing$ it must be the case that for all $z$ $[f_{V}(z)) \not = \varnothing$ as well as for some $v'' \sim v'$ such that $v'' = f_{V};\iota_{1,V}(z'')$.
                Clearly such $v''$ exists.
                It is either $v$, if $v = f_{V};\iota_{1,V}(z'')$, or some $x_i$ in $(x_1, \ldots, x_n)$ such that $x_1 = v'$ and $x_n = v$.
                \begin{align*}
                    [e_2] &= [<_{X+Y}^{\mu}(\iota_{V_{X} + V_{Y}}(v_1))]\\
                          &= [\iota_{1,E}(<_{X}(\iota_{V_{X}}(f_{V}(z_1))))]\\
                          &= [\iota_{1,E}(<_{X}(\iota_{V_{X}}(f_{V}(z''))))]\\
                          &= [<_{X+Y}(\iota_{V_{X} + V_{Y}}(f_{V};\iota_{1,V}(z'')))]\\
                          &= [<_{X+Y}(\iota_{V_{X} + V_{Y}}(v''))]\\
                          &= <_{\sfrac{X+Y}{\sim}}^{\mu}(\iota_{\sfrac{V_{X} + V_{Y}}{\sim}}([v'']))\\
                          &= <_{\sfrac{X+Y}{\sim}}^{\mu}(\iota_{\sfrac{V_{X} + V_{Y}}{\sim}}([v]))
                \end{align*}
                The case when $v_1 = g_{V};\iota_{2,V}(z_1)$ is symmetric.
            \end{itemize}
            \end{itemize}
            \item Finally, we need to check that if $[v] \in s([e_1])$ and $[e_2] = <_{\sfrac{X+Y}{\sim}}^{\mu}(\iota_{\sfrac{V_{X} + V_{Y}}{\sim}}([v]))$ then $[e_2] = <_{\sfrac{X+Y}{\sim}}^{\mu}(\iota_{\sfrac{E_{X} + E_{Y}}{\sim}}([e_1]))$.
            Once again, $[v] \in s_{\sfrac{X+Y}{\sim}}([e_1])$ means that there exists $v' \in s_{X+Y}(e_1)$ such that $v' \sim v$.
            \begin{itemize}
                \item \[
                [e_2] = [<_{X+Y}^{\mu}(\iota_{V_{X} + V_{Y}}(v''))]
                \] such that $v'' \sim v$ and $[v'') \not = \varnothing$.
                \begin{itemize}
                    \item If $v' = v$ and $v'' = v$ then $v \in s_{X+Y}(e_1)$ and since $[v'') = [v) \not = \varnothing$ and $v \in s_{X+Y}(e_1)$, it must be the case that
                    $<_{X+Y}^{\mu}(\iota_{E_{X} + E_{Y}}(e_1)) = <_{X+Y}^{\mu}(\iota_{E_{X} + E_{Y}}(v))$ and further
                    \begin{align*}
                    <_{\sfrac{X+Y}{\sim}}^{\mu}(\iota_{\sfrac{E_{X} + E_{Y}}{\sim}}([e_1])) &= [<_{X+Y}^{\mu}(\iota_{E_{X} + E_{Y}}(e_1))]\\
                                                                                            &= [<_{X+Y}^{\mu}(\iota_{E_{X} + E_{Y}}(v))]\\
                                                                                            &= <_{\sfrac{X+Y}{\sim}}^{\mu}(\iota_{\sfrac{V_{X} + V_{Y}}{\sim}}([v]))\\
                                                                                            &= [e_2]
                    \end{align*}
                    \item If $v' = v $ and $v'' \sim v$ then $v \in s_{X+Y}(e_1)$.
                          Suppose $v'' = f_{V};\iota_{1,V}(z'')$.
                          Consider cases.
                          \begin{itemize}
                            \item $v = f_{V};\iota_{1,V}(z)$, then $f_{V}(z) = f_{V}(z'') \not = \varnothing$ and $[e_1) = [v) = [v'')$
                            \item $v = g_{V};\iota_{2,V}(z)$, then $g_{V}(z) = \varnothing$ and $[e_1) = \varnothing$.
                                  Since there is a path from $v$ to $e_1$ and $[v'') \not = \varnothing$ and $v'' \sim v$, by definition,
                                  \begin{align*}
                                    <_{\sfrac{X+Y}{\sim}}^{\mu}(\iota_{\sfrac{E_{X} + E_{Y}}{\sim}}([e_1])) &= <_{\sfrac{X+Y}{\sim}}^{\mu}(\iota_{\sfrac{V_{X} + V_{Y}}{\sim}}([v'']))\\
                                    &= [e_2]
                                  \end{align*}
                          \end{itemize}
                          Case when $v'' = g_{V};\iota_{2,V}(z'')$ is symmetric.
                    \item If $v' \sim v$ and $v'' = v$ then $v' \in s_{X+Y}(e_1)$. This implies that $[v) \not = \varnothing$
                          \begin{itemize}
                            \item If $v = f_{V};\iota_{1,V}(z)$ and $v' = f_{V};\iota_{1,V}(z')$ then $[e_1) = [v') = [v) \not = \varnothing $ since $[f_{V}(z')) = [f_{V}(z))$.
                                  That is,
                                  \begin{align*}
                                    <_{\sfrac{X+Y}{\sim}}^{\mu}(\iota_{\sfrac{E_{X} + E_{Y}}{\sim}}([e_1])) &= [<_{X+Y}^{\mu}(\iota_{E_{X} + E_{Y}})(e_1)]\\
                                                                                                    &= [<_{X+Y}^{\mu}(\iota_{V_{X} + V_{Y}})(v)]\\
                                                                                                    &= [<_{X+Y}^{\mu}(\iota_{V_{X} + V_{Y}})(v'')]\\
                                                                                                    &= [e_2]
                                  \end{align*}
                            \item $v = f_{V};\iota_{1,V}(z)$ and $v' = g_{V};\iota_{2,V}(z')$ and then $[g_{V}(z')) = \varnothing$ and $[e_1) = \varnothing$.
                                  By definition,
                                  \begin{align*}
                                  <_{\sfrac{X+Y}{\sim}}^{\mu}(\iota_{\sfrac{E_{X} + E_{Y}}{\sim}}([e_1])) &= <_{\sfrac{X+Y}{\sim}}^{\mu}(\iota_{\sfrac{V_{X} + V_{Y}}{\sim}}([v]))\\
                                                                                                          &= <_{\sfrac{X+Y}{\sim}}^{\mu}(\iota_{\sfrac{V_{X} + V_{Y}}{\sim}}([v'']))\\
                                                                                                          &= [e_2]
                                  \end{align*}
                          \end{itemize}
                          The case when $v = g_{V};\iota_{1,V}(z)$ is symmetric.
                    \item If $v' \sim v$ and $v'' \sim v$ then $v' \in s_{X+Y}{e_1}$, $[v'') \not = \varnothing$ and $v'' \sim v'$.
                    \begin{itemize}
                        \item If $v'' = f_{V};\iota_{1,V}(z'')$ and $v' = f_{V};\iota_{2,V}(z')$ then $[e_1) = [v') = [v'')$ and
                            \begin{align*}
                                    <_{\sfrac{X+Y}{\sim}}^{\mu}(\iota_{\sfrac{E_{X} + E_{Y}}{\sim}}([e_1])) &= [<_{X+Y}^{\mu}(\iota_{E_{X} + E_{Y}})(e_1)]\\
                                                                                                    &= [<_{X+Y}^{\mu}(\iota_{V_{X} + V_{Y}})(v')]\\
                                                                                                    &= [<_{X+Y}^{\mu}(\iota_{V_{X} + V_{Y}})(v'')]\\
                                                                                                    &= [e_2]
                                  \end{align*}
                        \item If $v'' = f_{V};\iota_{1,V}(z'')$ and $v' = g_{V};\iota_{2,V}(z')$ then $[g_{V}(z')) = \varnothing$ and $[e_1) = \varnothing$. 
                        Then
                        \begin{align*}
                            <_{\sfrac{X+Y}{\sim}}^{\mu}(\iota_{\sfrac{E_{X} + E_{Y}}{\sim}}([e_1])) &= <_{\sfrac{X+Y}{\sim}}^{\mu}(\iota_{\sfrac{V_{X} + V_{Y}}{\sim}}([v'']))\\
                                                                                                    &= [e_2]
                        \end{align*}
                        The case when $v'' = g_{V};\iota_{2,V}(z'')$ is symmetric.
                    \end{itemize}
                \end{itemize}
                \item Now suppose that 
                \[
                [e_2] = <_{\sfrac{X+Y}{\sim}}^{\mu}(\iota_{\sfrac{V_{X} + V_{Y}}{\sim}}([v'']))
                \]
                such that $[v'') \not = \varnothing$, $[v) = \varnothing$ and there is a path from $[v]$ to $[v'']$, i.e. there is a path from $v$ to $v'$ and $v' \sim v''$.
                In its turn $[v] \in s_{\sfrac{X + Y}{\sim}}([e_1])$ implies $v_1 \in s_{X+Y}(e_1)$ and $v \sim v_1$.
                As per the comment below~\ref{def:child_respects_connectivity}, necessarily $v = v_1$ as $v$ should not have a pre-image in $V_{Z}$.
                This implies that $v \in s_{X+Y}(e_1)$ and that $[e_1) = \varnothing$.
                As there is a path from $v$ to $v'$, there is a path from $e_1$ to $v'$ and by definition, since $v'' \sim v'$ and $[v'') \not = \varnothing$,
                \begin{align*}
                <_{\sfrac{X+Y}{\sim}}^{\mu}(\iota_{\sfrac{E_{X} + E_{Y}}{\sim}}([e_1])) &= <_{\sfrac{X+Y}{\sim}}^{\mu}(\iota_{\sfrac{V_{X} + V_{Y}}{\sim}}([v'']))\\
                 &\;= [e_2]
                \end{align*}
            \end{itemize}    
\end{proof}

\begin{lemma}
\label{lemma:closures_of_equal_sets_are_equal}
If
\[
\consistency_{\sfrac{X+Y}{\sim}}^{\hashtag}(\iota_{\sfrac{*}{\sim}}([x_1])) = \consistency_{\sfrac{X+Y}{\sim}}^{\hashtag}(\iota_{\sfrac{*}{\sim}}([x_2]))
\]
where $\iota_{\sfrac{*}{\sim}}$ is either $\iota_{\sfrac{V_{X} + V_{Y}}{\sim}}$ or $\iota_{\sfrac{E_{X} + E_{Y}}{\sim}}$ depending on whether $[x_i]$ is in $V_{X} + V_{Y}$ or $E_{X} + E_{Y}$.
Then
\[
\consistency_{\sfrac{X+Y}{\sim}}(\iota_{\sfrac{*}{\sim}}([x_1])) = \consistency_{\sfrac{X+Y}{\sim}}(\iota_{\sfrac{*}{\sim}}([x_2]))
\]
That is, closures of both sides are equal.
\end{lemma}

\begin{proof}
    To check that the sets after closures are equal, we need to check if every element of the set on the left is also in the set on the right and vice versa.
    Apart from $\consistency_{\sfrac{X+Y}{\sim}}^{\hashtag}(\iota_{\sfrac{*}{\sim}}([x_1]))$ the set on the left also contains
    \begin{itemize}
        \item $[x_1]$ as per the reflexivity.
        So we need to check if $[x_1] \in \consistency_{\sfrac{X+Y}{\sim}}(\iota_{\sfrac{*}{\sim}}([x_2]))$.
              Because 
              \[
              \consistency_{\sfrac{X+Y}{\sim}}^{\hashtag}(\iota_{\sfrac{*}{\sim}}([x_1])) = \consistency_{\sfrac{X+Y}{\sim}}^{\hashtag}(\iota_{\sfrac{*}{\sim}}([x_2]))
              \]
              there exists some $[y]$ such that 
              \[
              [y] \in \consistency_{\sfrac{X+Y}{\sim}}^{\hashtag}(\iota_{\sfrac{*}{\sim}}([x_1]))
              \]
               and
            \[
                [y] \in \consistency_{\sfrac{X+Y}{\sim}}^{\hashtag}(\iota_{\sfrac{*}{\sim}}([x_2]))
            \]
            that is, we have a sequence $([x_1],[y],[x_2])$ and by transitivity
               \[
               [x_1] \in \consistency_{\sfrac{X+Y}{\sim}}(\iota_{\sfrac{*}{\sim}}([x_2]))
               \]
                and
            \[
            [x_2] \in \consistency_{\sfrac{X+Y}{\sim}}(\iota_{\sfrac{*}{\sim}}([x_1]))
            \]
        \item $[y]$ for all $[y]$ such that $[x_1] \in \consistency_{\sfrac{X+Y}{\sim}}^{\hashtag}(\iota_{\sfrac{*}{\sim}}([y]))$.
              Let's check if $[y] \in \consistency_{\sfrac{X+Y}{\sim}}(\iota_{\sfrac{*}{\sim}}([x_2]))$.
              As above, there exists $[z]$ such that
              \[
              [z] \in \consistency_{\sfrac{X+Y}{\sim}}^{\hashtag}(\iota_{\sfrac{*}{\sim}}([x_1]))
              \]
               and
              \[
                [z] \in \consistency_{\sfrac{X+Y}{\sim}}^{\hashtag}(\iota_{\sfrac{*}{\sim}}([x_2]))
              \]
              By a sequence $([y],[x_1],[z],[x_2])$ we get that
              \[
                [y] \in \consistency_{\sfrac{X+Y}{\sim}}^{\hashtag}(\iota_{\sfrac{*}{\sim}}([x_2]))
              \]
              Similarly for $[y]$ such that $[x_2] \in \consistency_{\sfrac{X+Y}{\sim}}^{\hashtag}(\iota_{\sfrac{*}{\sim}}([y]))$.
        \item For any sequence $([x_1],\ldots,[x_n])$ such that $[x_i] \in \consistency_{\sfrac{X+Y}{\sim}}^{\hashtag}(\iota_{\sfrac{*}{\sim}}([x_{i+1}]))$ or $[x_{i+1}] \in \consistency_{\sfrac{X+Y}{\sim}}(\iota_{\sfrac{*}{\sim}}([x_{i}]))$ for $i < n$ the set on the left contains $[x_n]$.
              Let's check if $[x_n] \in \consistency_{\sfrac{X+Y}{\sim}}(\iota_{\sfrac{*}{\sim}}([x_2]))$.
              Since there exists $[y]$ such that $[y] \in \consistency_{\sfrac{X+Y}{\sim}}^{\hashtag}(\iota_{\sfrac{*}{\sim}}([x_1]))$ and $[y] \in \consistency_{\sfrac{X+Y}{\sim}}^{\hashtag}(\iota_{\sfrac{*}{\sim}}([x_2]))$ we can extend the sequence to $([x_2],[y],[x_1],\ldots,[x_n])$ and hence $[x_n] \in \consistency_{\sfrac{X+Y}{\sim}}(\iota_{\sfrac{*}{\sim}}([x_2]))$.
              Similarly, given a sequence $([x_2],\ldots, [x_n])$ and extending it to $([x_1],[y],[x_2],\ldots,[x_n])$ we get $[x_n] \in \consistency_{\sfrac{X+Y}{\sim}}(\iota_{\sfrac{*}{\sim}}([x_1]))$.
    \end{itemize}
\end{proof}

\begin{lemma}
\label{lemma:consistency_respect}
$\consistency_{\sfrac{X+Y}{\sim}}$ respects connectivity.
\end{lemma}
\begin{proof}
    Suppose $[v_1] \in s_{\sfrac{X+Y}{\sim}}([e_1])$ and $[[v_1]) = [[e_1]) \not = \varnothing$, then it must be the case that 
    \[
    \consistency_{\sfrac{X+Y}{\sim}}(\iota_{\sfrac{V_{X} + V_{Y}}{\sim}}([v_1])) = \consistency_{\sfrac{X+Y}{\sim}}(\iota_{\sfrac{E_{X} + E_{Y}}{\sim}}([e_1]))
    \]
    By the lemma above, to prove this, it is sufficient to show that

    \[
    \consistency_{\sfrac{X+Y}{\sim}}^{\hashtag}(\iota_{\sfrac{V_{X} + V_{Y}}{\sim}}([v_1])) = \consistency_{\sfrac{X+Y}{\sim}}^{\hashtag}(\iota_{\sfrac{E_{X} + E_{Y}}{\sim}}([e_1]))
    \]

    $[v_1] \in s_{\sfrac{X+Y}{\sim}}([e_1])$ means there exists $v_1' \in s_{X+Y}(e_1)$ such that $v_1 \sim v_1'$.
    We then have two cases.
    \begin{itemize}
      \item $v_1 = v_1'$ and then $v_1 \in s_{X+Y}(e_1)$.
      \begin{itemize}
         \item $<_{\sfrac{X+Y}{\sim}}^{\mu}(\iota_{\sfrac{V_{X} + V_{Y}}{\sim}}([v_1])) = [<_{X+Y}^{\mu}(\iota_{V_{X} + V_{Y}}(v_1''))]$ such that $v_1'' \sim v_1$ and $[v_1'') \not = \varnothing$.
            \begin{itemize}
                \item $v_1'' = v_1$ and then $[v_1) \not = \varnothing$ and $[e_1) \not = \varnothing$.
                      Furthermore, necessarily $\consistency_{X+Y}(\iota_{E_{X} + E_{Y}}(e_1)) = \consistency_{X+Y}(\iota_{V_{X} + V_{Y}}(v_1))$.
                      By applying $([]_{V},[]_{E})$ we get
                      \ifdefined \ONECOLUMN
                      \begin{align*}
                        \text{by functionality}\\
                          [[]_{V};\iota_{\sfrac{V_{X} + V_{Y}}{\sim}},[]_{E};\iota_{\sfrac{E_{X} + E_{Y}}{\sim}}]^{*}(\consistency_{X+Y}(\iota_{E_{X} + E_{Y}}(e_1))) &=
                          [[]_{V};\iota_{\sfrac{V_{X} + V_{Y}}{\sim}},[]_{E};\iota_{\sfrac{E_{X} + E_{Y}}{\sim}}]^{*}(\consistency_{X+Y}(\iota_{V_{X} + V_{Y}}(v_1)))\\
                          \text{by unfolding the definition}\\
                          \consistency_{\sfrac{X+Y}{\sim}}^{\hashtag}(\iota_{\sfrac{E_{X} + E_{Y}}{\sim}}([e_1])) &= \consistency_{\sfrac{X+Y}{\sim}}^{\hashtag}(\iota_{\sfrac{V_{X} + V_{Y}}{\sim}}([v_1]))\\
                          \end{align*} 
                      \else
                      \begin{align*}
                      \text{by functionality}\\
                        [[]_{V};\iota_{\sfrac{V_{X} + V_{Y}}{\sim}},[]_{E};\iota_{\sfrac{E_{X} + E_{Y}}{\sim}}]^{*}(\consistency_{X+Y}(\iota_{E_{X} + E_{Y}}(e_1))) &=\\
                        [[]_{V};\iota_{\sfrac{V_{X} + V_{Y}}{\sim}},[]_{E};\iota_{\sfrac{E_{X} + E_{Y}}{\sim}}]^{*}(\consistency_{X+Y}(\iota_{V_{X} + V_{Y}}(v_1))) &\\
                        \text{by unfolding the definition}\\
                        \consistency_{\sfrac{X+Y}{\sim}}^{\hashtag}(\iota_{\sfrac{E_{X} + E_{Y}}{\sim}}([e_1])) &=\\
                        \consistency_{\sfrac{X+Y}{\sim}}^{\hashtag}(\iota_{\sfrac{V_{X} + V_{Y}}{\sim}}([v_1]))&
                        \end{align*} 
                      \fi
                \item $v_1'' \sim v_1$. If $v_1'' = f_{V};\iota_{1,V}(z_1'')$ and $v_1 = f_{V};\iota_{1,V}(z_1)$ we get
                    \begin{align*}
                        \consistency_{\sfrac{X+Y}{\sim}}^{\hashtag}&(\iota_{\sfrac{V_{X} + V_{Y}}{\sim}}([v_1])) =\\
                                                                                                                &= [[]_{V}^{\consistency},[]_{E}^{\consistency}]^{*}(\consistency_{X+Y}(\iota_{V_{X} + V_{Y}}(v'')))\\
                                                                                                                &= [[]_{V}^{\consistency},[]_{E}^{\consistency}]^{*}(\consistency_{X+Y}(\iota_{V_{X} + V_{Y}}(f_{V};\iota_{1,V}(z_1''))))\\
                                                                                                                &= [[]_{V}^{\consistency},[]_{E}^{\consistency}]^{*}(\consistency_{X+Y}(\iota_{V_{X} + V_{Y}}(f_{V};\iota_{1,V}(z_1))))\\
                                                                                                                &= [[]_{V}^{\consistency},[]_{E}^{\consistency}]^{*}(\consistency_{X+Y}(\iota_{E_{X} + E_{Y}}(e_1)))\\
                                                                                                                &= \consistency_{\sfrac{X+Y}{\sim}}^{\hashtag}(\iota_{\sfrac{E_{X} + E_{Y}}{\sim}}([e_1]))
                    \end{align*}
                    If $v_1'' = f_{V};\iota_{1,V}(z_1'')$ and $v_1 = g_{V};\iota_{1,V}(z_1)$ then $[v_1) = [e_1) = \varnothing$ and 
                    \[
                        \consistency_{X+Y}(\iota_{V_{X} + V_{Y}}(v_1)) = \consistency_{X+Y}(\iota_{E_{X} + E_{Y}}(e_1)) = \varnothing
                    \]
                    \begin{align*}
                    \consistency_{\sfrac{X+Y}{\sim}}^{\hashtag}(\iota_{\sfrac{E_{X} + E_{Y}}{\sim}}([e_1])) &= \consistency_{\sfrac{X+Y}{\sim}}^{\hashtag}(\iota_{\sfrac{V_{X} + V_{Y}}{\sim}}([v_1'']))\\
                                                                                                            &= \consistency_{\sfrac{X+Y}{\sim}}^{\hashtag}(\iota_{\sfrac{V_{X} + V_{Y}}{\sim}}([v_1]))
                    \end{align*}
                \end{itemize}
                \item $<_{\sfrac{X+Y}{\sim}}^{\mu}(\iota_{\sfrac{V_{X} + V_{Y}}{\sim}}([v_1])) = <_{\sfrac{X+Y}{\sim}}^{\mu}(\iota_{\sfrac{V_{X} + V_{Y}}{\sim}}([v_2]))$ such that $[v_2) \not = \varnothing$ and there is a path from $[v_1]$ to $[v_2]$. 
                By definition, this implies $[v_1) = \varnothing$ and $[e_1) = \varnothing$.
                Again, by definition, we have
                \begin{align*}
                    \consistency_{\sfrac{X+Y}{\sim}}^{\hashtag}(\iota_{\sfrac{V_{X} + V_{Y}}{\sim}}([v_1])) &= \consistency_{\sfrac{X+Y}{\sim}}^{\hashtag}(\iota_{\sfrac{V_{X} + V_{Y}}{\sim}}([v_2]))\\
                                                                                                    &= \consistency_{\sfrac{X+Y}{\sim}}^{\hashtag}(\iota_{\sfrac{E_{X} + E_{Y}}{\sim}}([e_1]))
                \end{align*}
                \end{itemize}
                \item $v_1 \sim v_1'$ and then $v_1' \in s_{X+Y}(e_1)$.
                \begin{itemize}
                    \item $<_{\sfrac{X+Y}{\sim}}^{\mu}(\iota_{\sfrac{V_{X} + V_{Y}}{\sim}}([v_1])) = [<_{X+Y}^{\mu}(\iota_{V_{X} + V_{Y}}(v_1''))]$ such that $v_1 \sim v_1''$ and $[v_1'') \not = \varnothing$.
                    \begin{itemize}
                        \item $v_1 = v_1''$ and then $[v_1) \not = \varnothing$.
                              Suppose $v_1 = f_{V};\iota_{1,V}(z_1)$ and $v_1' = f_{V};\iota_{1,V}(z_1')$. Then $[v_1') = [e_1) \not = \varnothing$ and 
                            \begin{align*}
                              \consistency_{\sfrac{X+Y}{\sim}}^{\hashtag}&(\iota_{\sfrac{V_{X} + V_{Y}}{\sim}}[v_1']) =\\
                              &\;= [[]_{V}^{\hashtag},[]_{E}^{\hashtag}]^{*}(\consistency_{X+Y}(\iota_{V_{X} + V_{Y}}(v_1')))\\
                                                                                                                     &\;= [[]_{V}^{\hashtag},[]_{E}^{\hashtag}]^{*}(\consistency_{X+Y}(\iota_{E_{X} + E_{Y}}(e_1)))\\
                                                                                                                     &\;= \consistency_{\sfrac{X+Y}{\sim}}^{\hashtag}(\iota_{\sfrac{E_{X} + E_{Y}}{\sim}}([e_1]))\\
                                                                                                                     &\;= \consistency_{\sfrac{X+Y}{\sim}}^{\hashtag}(\iota_{\sfrac{V_{X} + V_{Y}}{\sim}}[v_1])
                            \end{align*}
                            Suppose $v_1' = g_{V};\iota_{2,V}(z_1')$ and then $[v_1') = [e_1) = \varnothing$.
                            By definition,
                            \begin{align*}
                                &\consistency_{\sfrac{X+Y}{\sim}}^{\hashtag}(\iota_{\sfrac{V_{X} + V_{Y}}{\sim}}[v_1']) =\\
                                &\;[[]_{V}^{\hashtag},[]_{E}^{\hashtag}]^{*}(\consistency_{X+Y}(\iota_{V_{X} + V_{Y}}(v_1)))
                            \end{align*}
                            and since there is a path from $e_1$ to $v_1'$ and $v_1' \sim v_1$
                            \begin{align*}
                                &\consistency_{\sfrac{X+Y}{\sim}}^{\hashtag}(\iota_{\sfrac{E_{X} + E_{Y}}{\sim}}[e_1]) =\\
                                &\;[[]_{V}^{\hashtag},[]_{E}^{\hashtag}]^{*}(\consistency_{X+Y}(\iota_{V_{X} + V_{Y}}(v_1)))
                            \end{align*}
                            The cases when $v_1 = g_{V};\iota_{2,V}(z_1)$ is symmetric.
                        \item $v_1 \sim v_1''$ and then $v_1' \sim v_1''$. Suppose $v_1'' = f_{V};\iota_{1,V}(z_1'')$.
                              If $v_1'$ is in the image of $f_{V};\iota_{1,V}$, then
                              \begin{align*}
                                \consistency_{X+Y}(\iota_{V_{X} + V_{Y}}(v_1')) &= \consistency_{X+Y}(\iota_{V_{X} + V_{Y}}(v_1''))\\ 
                                                                               &= \consistency_{X+Y}(\iota_{V_{X} + V_{Y}}(e_1'))
                              \end{align*}
                              and
                              \begin{align*}
                                \consistency_{\sfrac{X+Y}{\sim}}^{\hashtag}(\iota_{\sfrac{V_{X} + V_{Y}}{\sim}}[v_1]) &=
                                \consistency_{\sfrac{X+Y}{\sim}}^{\hashtag}(\iota_{\sfrac{V_{X} + V_{Y}}{\sim}}[v_1'])\\ &= \consistency_{\sfrac{X+Y}{\sim}}^{\hashtag}(\iota_{\sfrac{E_{X} + E_{Y}}{\sim}}[e_1])
                              \end{align*}
                              If $v_1'$ is in the image of $g_{V};\iota_{2,V}$, then $\consistency_{X+Y}(\iota_{V_{X} + V_{Y}}(v_1')) = \consistency_{X+Y}(\iota_{E_{X} + E_{Y}}(e_1)) = \varnothing$ and by definition as there is a path from $[e_1]$ to $[v_1'']$
                              \begin{align*}
                                \consistency_{\sfrac{X+Y}{\sim}}^{\hashtag}&(\iota_{\sfrac{E_{X} + E_{Y}}{\sim}}([e_1])) =\\
                                &\consistency_{\sfrac{X+Y}{\sim}}^{\hashtag}(\iota_{\sfrac{V_{X} + V_{Y}}{\sim}}([v_1'']))\\ 
                                                                                                                        &= \consistency_{\sfrac{X+Y}{\sim}}^{\hashtag}(\iota_{\sfrac{V_{X} + V_{Y}}{\sim}}([v_1]))
                              \end{align*}
                              The cases when $v_1'' = g_{V};\iota_{2,V}(z_1'')$ are symmetric.
                    \end{itemize}
                \end{itemize}
     \end{itemize}
The case when $[v] \in t([e])$ is symmetric.
\end{proof}

\begin{lemma}
\label{lemma:consistency_child}
    $\consistency_{\sfrac{X+Y}{\sim}}$ is defined only for vertices and edges that are not top-level.
\end{lemma}
\begin{proof}
  $\consistency_{\sfrac{X+Y}{\sim}}(\iota_{\sfrac{*}{\sim}}([x])) \not = \varnothing$ if and only if $<_{\sfrac{X+Y}{\sim}}^{\mu}(\iota_{\sfrac{*}{\sim}}([x]))$ is defined.
  Let's first show from left to right. $\consistency_{\sfrac{X+Y}{\sim}}(\iota_{\sfrac{*}{\sim}}([x])) \not = \varnothing$ implies $\consistency_{\sfrac{X+Y}{\sim}}^{\hashtag}(\iota_{\sfrac{*}{\sim}}([x])) \not = \varnothing$.
  \begin{itemize}
    \item $\consistency_{\sfrac{X+Y}{\sim}}^{\hashtag}(\iota_{\sfrac{*}{\sim}}([x])) = [[]_{V}^{\consistency},[]_{E}^{\consistency}]^{*}(\iota_{*}{\sim}(x'))$ such that $x \sim x'$.
          By e-hypergraph-ness of $X+Y$ this implies that $[x') \not = \varnothing$ and hence $<_{\sfrac{X+Y}{\sim}}^{\hashtag}(\iota_{\sfrac{*}{\sim}}([x'])) = <_{\sfrac{X+Y}{\sim}}^{\hashtag}(\iota_{\sfrac{*}{\sim}}([x])) = [<_{X+Y}^{\mu}(\iota_{*}(x))]$
    \item $\consistency_{\sfrac{X+Y}{\sim}}^{\hashtag}(\iota_{\sfrac{*}{\sim}}([x])) = \consistency_{\sfrac{X+Y}{\sim}}^{\hashtag}(\iota_{\sfrac{V_{X} + V_{Y}}{\sim}}([v]))$ such that there is a path from $[x]$ to $[v]$ and
          $\consistency_{X+Y}(\iota_{V_{X} + V_{Y}}(v)) \not = \varnothing$. This implies that $[v) \not = \varnothing$ and by definition
          \[
            <_{\sfrac{X+Y}{\sim}}^{\mu}(\iota_{\sfrac{*}{\sim}}([x])) = <_{\sfrac{X+Y}{\sim}}^{\mu}(\iota_{\sfrac{*}{\sim}}([v]))  
          \]
  \end{itemize}
  The right-to-left direction is analogous.
\end{proof}

\begin{lemma}
$\consistency_{\sfrac{X+Y}{\sim}}^{p}$ is not total.
\end{lemma}
\begin{proof}
The property holds as we require that the arrows of the span for the pushout do not map vertices to different components and therefore they can not be merged into one component.
\end{proof}

\begin{proposition}[$\sfrac{X+Y}{\sim}$ is an e-hypergraph]
\label{prop:pushout_is_e_hypergraph}
For this we need to check that the defined relations --- $<_{\sfrac{X+Y}{\sim}}$ and $\consistency_{\sfrac{X+Y}{\sim}}$ --- satisfy all the needed properties as per Definition~\ref{def:e-homo}.
\end{proposition}
\begin{proof}
    This proposition then follows by applying the lemmas above.
    % Lemmas~\ref{lemma:child_assymetric}~\ref{lemma:child_irreflexive}~\ref{lemma:child_respect}~\ref{lemma:consistency_child}~\ref{lemma:consistency_respect}.
\end{proof}

We will then discuss the uniqueness of the corresponding boundary pushout complement.

% \begin{definition}[Definition 3.16~\cite{bonchi_string_2022-1}]
% Morphisms $\mathcal{K} \xrightarrow{f} \mathcal{L} \xrightarrow{m} \mathcal{G}$ satisfy the \textit{no-dangling} condition if, for every hyperedge not in the image of $m$, every vertex of its source and target is either (i) not in the image of $m$ or (ii) in the image of $f ; m$.
% They satisfy the no-identification condition if any two vertices merged by $m$ are in the image of $f$.
% \end{definition}
% Intuitively, the no-dangling condition guarantees that sources and targets of an edge are not deleted if the edge is not deleted.
% The no-identification condition requires that if $m$ identifies two vertices, it must not be the case that one vertex should be removed and the other vertex should be preserved according to $\mathcal{K}$.

% \begin{figure*}[t!]
%     \begin{minipage}{0.1\textwidth}
%         $\;$
%     \end{minipage}
%     \hfill
% \begin{minipage}{0.7\textwidth}
%     \[
%         \scalebox{0.45}{
%             \tikzfig{../figures/combinatorial_semantics/structural_rule_pushout}
%         }
%     \]
%     \caption{The image of the top left equation from Figure~\ref{fig:string-equations} under $\llbracket - \rrbracket$}
%     \label{fig:structural_rule_pushout}
% \end{minipage}
% \hfill
% \begin{minipage}{0.1\textwidth}
%         $\;$
% \end{minipage}
% \end{figure*}

\begin{proposition}[Proposition~\ref{prop:boundary_unique} restatement]
\label{prop:boundary_complement}
    The boundary complement in~\ref{def:boundary_new} when exists is unique.
\end{proposition}
\begin{proof}
        The pushout in~\ref{def:boundary_new} satisfies the assumptions in Definition~\ref{pushout:assumptions} (note that since $i \to i' \to \mathcal{L} \xleftarrow{} j' \xleftarrow{} j$ is an MDA cospan, the image of $i + j$ in $\mathcal{L}$ consists of top-level vertices only) and therefore is isomorphic to the pushout constructed in~\ref{pushout:assumptions} and by construction should also be a pushout on the underlying sets of vertices and hyperedges.
        In particular,
    \[
        \input{./figures/pushout_verticies.tikz}
    \]
    \[
        \input{./figures/pushout_edges.tikz}
    \]
    where $V$s and $E$s are sets of vertices and hyperedges respectively, marked squares are pushouts and $\mathcal{L}^{\bot}$s are pushout complements and where $E_{0}$ is either empty or contains edges with no inputs and outputs only.
    % Because $i+j$ is discrete, $E_{\mathcal{G}}$ is the disjoint union of $E_{\mathcal{L}}$ and $E_{\mathcal{L}^{\bot}}$ and hence $E_{\mathcal{L}^{\bot}} = E_{\mathcal{G}} \setminus E_{\mathcal{L}}$.
    Since $m$ and $[c_1,c_2]$ are monos, the first square can be rewritten as follows.
% https://q.uiver.app/#q=WzAsNSxbMCwwLCJpICsgaiArIHggKyB5Il0sWzIsMCwiaStqICsgeCJdLFs0LDAsImkraiJdLFswLDIsImkgKyBqICsgeCArIHkgKyB6Il0sWzQsMiwiaSArIGogKyB3XFxcXFxcY29uZ1xcXFwgaSArIGogKyB6Il0sWzIsMSwiaF97VixleHR9IiwyXSxbMSwwLCJoX3tWLGludH0iLDJdLFswLDMsIm1fe1Z9IiwyXSxbNCwzLCJnX3tWfSJdLFsyLDQsImMiLDJdXQ==
\[\begin{tikzcd}
	{a + y} && {i + j + x} && {i+j} \\
	\\
	{a + y + z} &&&& {\begin{matrix} i + j + w \\ \cong \\ i + j + z \end{matrix}}
	\arrow["{h_{V,ext}}"', from=1-5, to=1-3]
	\arrow["{h_{V,int}}"', from=1-3, to=1-1]
	\arrow["{m_{V}}"', from=1-1, to=3-1]
	\arrow["{g_{V}}", from=3-5, to=3-1]
	\arrow["c"', from=1-5, to=3-5]
\end{tikzcd}\]
    Where $y$ is the image of $i + j$. Clearly, $i + j + z$ is a pushout complement, as computing the pushout identifies $i + j$ with its image within $y$ and leaves $z$ as is.
    Suppose there is another pushout complement $i + j + w$.
    Because $a + y + w$ also yields a pushout it must be the case that $a + y + w \cong a + y + z$ and therefore $w \cong z$ and pushout complement for vertices is $i + j + z$ up to isomorphism.
    The same reasoning applies to show that $E_{\mathcal{L}^{\bot}}$ is unique and hence the sets of edges and vertices are uniquely determined in the pushout complement.
    One can observe that $g_{V} = h_{V,ext};h_{V,int} + id_{n}$ and that $h_{V,ext}$ is mono and $h_{V,int} = [h_1,h_2]$ where $h_1,h_2$ are mono.
    To show that pushout complement is unique we next need to show that source and target maps are unique.
    Suppose, on contrary, that there exist
    \[
    \mathcal{L}_{1}^{\bot} = \langle V_{\mathcal{L}^{\bot}}, E_{\mathcal{L}^{\bot}}, s_{1} \rangle \qquad \mathcal{L}_{2}^{\bot} = \langle V_{\mathcal{L}^{\bot}}, E_{\mathcal{L}^{\bot}}, s_{2} \rangle
    \]

    such that there is $e \in E_{\mathcal{L}^{\bot}}$ and $s_{1}(e) \not = s_{2}(e)$.
    As $g$ is a homomorphism, it must be the case that $g^{*}_{V}(s_{1}(e)) = g^{*}_{V}(s_{2}(e)) = s_{\mathcal{G}}(g(e))$, or
    $(h_{V,ext};[h_{1},h_{2}] + id_{n})(s_{1}(e)) = (h_{V,ext};[h_{1},h_{2}] + id_{n})(s_{2}(e))$.
    For the latter to hold when $s_{1}(e) \not = s_{2}(e)$ it must be the case that $h_{V,ext};[h_{1},h_{2}]$ identifies $v_{1} \in s_{1}(e)$ with $v_{2} \in s_{2}(e)$ and necessarily $v_{1}$ and $v_{2}$ are in the image of $c$.
    This can only happen if $v_1 \in i$ and $v_2 \in j$, \textit{i.e.} they belong to input and output interfaces respectively.
    However, this contradicts the fact that $v_2 \in s_2(e)$ as a vertex which is a source of some edge can not be in the output interface as per the (6) and (7) conditions of the boundary complement~\ref{def:boundary_new} the image of $j$ in $\mathcal{L}^{\bot}$ should consist of vertices of in-degree 0.
    Similarly for $t_{\mathcal{L}^{\bot}}$.

    % Because $g$ is a homomorphism, it must be the case that $g_{V}^{*}(s_{\mathcal{L^{\bot}}}(e)) = s_{\mathcal{G}}(g_{E}(e))$ and $s_{\mathcal{L^{\bot}}}(e) \subseteq (g_{V}^{-1})^{*}(s_{\mathcal{G}}(g_{E}(e)))$.
    % Since $g_{V} = h_{V,ext};[h_1,h_2] + id_{n}$ the inverse image of $g_{V}$ contains at most two elements: one with the pre-image in $i$ and another with the pre-image in $j$.
    % If $g_{V}$ identified more than two vertices it would mean that $[h_1,h_2]$ identified more than two vertices which would mean that two of them come both either from $i$ or $j$ which would imply that either $h_1$ or $h_2$ is not mono.
    % If the inverse image contains exactly one vertex then the $s_{\mathcal{L^{\bot}}}(e)$ is uniquely fixed, otherwise suppose it contains $\{v_1,v_2\}$ such that $v_1 \in i$ and $v_2 \in j$.
    % However, $v_1$ can not be in $s_{\mathcal{L^{\bot}}}(e)$ as per the (6) and (7) conditions of the boundary complement~\ref{def:boundary_new} the image of $j$ in $\mathcal{L}^{\bot}$ should consist of vertices of in-degree 0.
    
    % Suppose that there exist
    % \[
    %     \mathcal{L}^{\bot}_{1} = (V_{\mathcal{L}^{\bot}}, E_{\mathcal{L}^{\bot}}, s_{1},t)
    % \qquad
    % \text{and}
    % \qquad
    %     \mathcal{L}^{\bot}_{2} = (V_{\mathcal{L}^{\bot}},E_{\mathcal{L}^{\bot}},s_{2},t)
    % \]
    % Because $g$ is a homomorphism, it must be the case that $g_{V}^{*}(s_{1}(e)) = s_{\mathcal{G}}^{*}(g_{E}(e)) = g_{V}^{*}(s_{2}(e))$.
    % Because $g$ is mono, $g_{V}^{*}(s_{1}(e)) = g_{V}^{*}(s_{2}(e))$ implies $s_{1}(e) = s_{2}(e)$.
    % Similarly for targets.
    We also need to show that $<$ and $\consistency$ are unique.
    Suppose that there exist
    \[
        \mathcal{L}^{\bot}_{1} = (V_{\mathcal{L}^{\bot}}, E_{\mathcal{L}^{\bot}}, s,t, <_{1})
    \qquad
    \text{and}
    \qquad
        \mathcal{L}^{\bot}_{2} = (V_{\mathcal{L}^{\bot}},E_{\mathcal{L}^{\bot}},s,t, <_{2})
    \]
    Because $g$ is homomorphism, it must be the case that 
    \[
        g_{E}(<_{1}(\iota_{V_{\mathcal{L}^{\bot}}}(v))) = <_{\mathcal{G}}(\iota_{V_{\mathcal{G}}}(g_{V}(v))) = g_{E}(<_{2}(\iota_{V_{\mathcal{L}^{\bot}}}(v)))
    \]
    which implies $<_{2}(\iota_{V_{\mathcal{L}^{\bot}}}(v)) = <_{1}(\iota_{V_{\mathcal{L}^{\bot}}}(v))$ for all $v$ such that $[v) \not = \varnothing$ since $g_{E}$ only identifies edges with no successors.
    Similarly for edges.
    We now show that if $[v) = \varnothing$ for $v \in V_{\mathcal{L}^{\bot}}$ then necessarily $[g_{V}(v)) = \varnothing$.
    Suppose $[v) = \varnothing$ and $<_{\mathcal{G}}(\iota_{V_{\mathcal{G}}}(g_{V}(v))) = e$ and there is a connected component that contains $v$ and such that for all $v'$ and $e'$ in this component $[v') = \varnothing = [e')$.
    Consider $\mathcal{G'}$ which is obtained from $\mathcal{G}$ by making $[g_{V}(v')) = [g_{E}(e')) = \varnothing$ for $v', e'$ above.
    Clearly there is a morphism from $\mathcal{L}^{\bot}$ to $\mathcal{G}'$, i.e. the morphism $g' = (g_{V},g_{E})$ that has the same action on edges and vertices as the morphism $g$ does, because $v'$ such that $[v') = \varnothing$ can be mapped to $g_{V}(v')$ such that $[g_{V}(v')) \not = \varnothing$.
    The only case when such a morphism may not exist is when $[v'') \not = \varnothing$, $[v') = \varnothing$ and $g_{V}(v'') = g_{V}(v')$.
    This implies that both $v''$ and $v'$ have a pre-image in $i + j$ and according to the definition of a boundary complement it must be $[[c_1,c_2](z_1)) = [[c_1,c_2](z_2))$ for all $z_1, z_2 \in V_{i + j}$.
    For the same reason there is also a morphism from $\mathcal{L}$ to $\mathcal{G}'$.
    The only case when such a morphism may not exist is when $v'' \in V_{\mathcal{L}}$ such that $[v'') \not = \varnothing$ is mapped to $m_{V}(v'') = g_{V}(v')$ such that $[v') = \varnothing$ which implies $v'$ and $v''$ share a pre-image in $i + j$.
    This leads to a contradiction as the image of $i + j$ in $\mathcal{L}$ should consist of top-level vertices exclusively.
    This construction results in a commutative diagram below
    
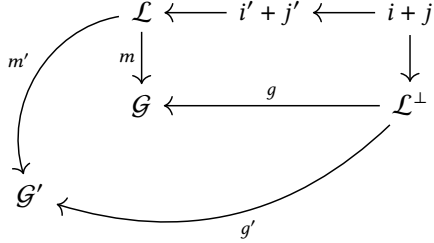
\[
        \begin{tikzcd}
            & \mathcal{L} \arrow[d, "m"'] \arrow[ldd, "m'"', bend right=49] & i'+j' \arrow[l] & i+j \arrow[d] \arrow[l]                                           \\
            & \mathcal{G}                                                   &                 & \mathcal{L}^{\bot} \arrow[ll, "g"'] \arrow[llld, "g'", bend left] \\
\mathcal{G}' &                                                               &                 &                                                                  
\end{tikzcd}
    \]
    that further results in a contradiction since there is no morphism from $\mathcal{G}$ to $\mathcal{G'}$ because for any such morphism $u$ it must be the case that $[u_{V}(g_{V}(v'))) \not = \varnothing = [g'_{V}(v')) = \varnothing$ which would mean that $\mathcal{G}$ is not a pushout.
    This means that either both $<_{1}(\iota_{*}(x))$ and $<_{2}(\iota_{*}(x))$ are undefined for a given $x$ or $<_{1}(\iota_{*}(x)) = <_{2}(\iota_{*}(x))$ which concludes that the relations are equal and $\mathcal{L}_{1}^{\bot} = \mathcal{L}_{2}^{\bot}$.
    The same reasoning applies when we consider $<_{1}(\iota_{E_{\mathcal{L}^{\bot}}}(e))$ and $<_{2}(\iota_{E_{\mathcal{L}^{\bot}}}(e))$.

    Finally, let's suppose there exist
    \[
        \mathcal{L}^{\bot}_{1} = (V_{\mathcal{L}^{\bot}}, E_{\mathcal{L}^{\bot}}, s,t, <, \consistency_{1})
    \quad
    \text{and}
    \quad
        \mathcal{L}^{\bot}_{2} = (V_{\mathcal{L}^{\bot}},V_{\mathcal{L}^{\bot}},s,t,<, \consistency_{2})
    \]
    Because $g$ is a homomorphism, we have
    \[
        [~g_{V};\iota_{V_{\mathcal{G}}}, g_{E};\iota_{E_{\mathcal{G}}}~]^{*}(\consistency_{\mathcal{L}^{\bot}_{1}}(\iota_{V_{\mathcal{L}^{\bot}}}(v))) \subseteq \consistency_{\mathcal{G}}(g_{V};\iota_{V_{\mathcal{G}}}(v))
    \]
    \[
        [~g_{V};\iota_{V_{G}}, g_{E};\iota_{E_{\mathcal{G}}}]^{*}(\consistency_{\mathcal{L}^{\bot}_{2}}(\iota_{V_{\mathcal{L}^{\bot}}}(v))) \subseteq \consistency_{\mathcal{G}}(g_{V};\iota_{V_{\mathcal{G}}}(v))
    \]
    \begin{itemize}
        \item If $\consistency_{\mathcal{G}}(g_{V};\iota_{V_{\mathcal{G}}}(v)) = \varnothing$ then both left-hand sides should be $\varnothing$ and therefore the relations are equal.
        \item Suppose that $\consistency_{\mathcal{G}}(g_{V};\iota_{V_{\mathcal{G}}}(v)) \not = \varnothing$ which further means that $[g_{V};\iota_{V_{\mathcal{G}}}(v)) \not = \varnothing$.
              $\consistency_{1} \not = \consistency_{2}$ implies there exist $v, v' \in \mathcal{L}^{\bot}$ such that $v \consistency_{1} v'$ and $v \not \consistency_{2} v'$ for $[v) = [v') \not = \varnothing$.
              Then there are connected components $C_1, C_2$ in $\mathcal{L}_{2}^{\bot}$ such that for all $x \in C_1$ and $y \in C_2$ $x \not \consistency_{2} y$ and $v \in C_1$, $v' \in C_2$.
              Consider $g_2 = (g_{V},g_{E}) : \mathcal{L}_2^{\bot} \to \mathcal{G}$ and construct $\mathcal{G'}$ that is obtained from $\mathcal{G}$ by making $g(x) \not \consistency g(y)$ for $x \in C_1$ and $y \in C_2$.
              There is a morphism $g_2' = g_2 : \mathcal{L}_2^{\bot} \to \mathcal{G}'$. Like above, the only case when such morphism may not exist is when for $v' \not \consistency v$ $g_{V}(v) = g_{V}(v')$.
              This implies $v, v'$ have a pre-image in $i + j$ and by definition it must be $v' \consistency v$.
              Similarly, there is a morphism $m' : \mathcal{L} \to \mathcal{G}'$. The only case when it may not exist is when for $v_1, v_1' \in \mathcal{L}$ such that $v_1 \consistency v_1'$, $m_{V}(v_1) = g_{V}(v_2)$ and $m_{V}(v_1') = g_{V}(v_2')$ where $v_2 \in C_1$ and $v_2' \in C_2$.
              It implies $v_1, v_1'$ have pre-image in $i + j$ and by definition of a boundary-complement the image of $i + j$ should consist of top-level vertices exclusively.
              This results in the diagram below there is no morphism $u : G \to G'$, as the existence of such a morphism implies $u(g(x)) = u(g'(x))$, $u(g(y)) = u(g'(y))$ for $x \in C_1$ and $y \in C_2$ and $u(g(x)) \consistency u(g(y))$ and $u(g'(x)) \not \consistency u(g'(y))$.
              This contradicts the fact that $\mathcal{G}$ is a pushout.
              Hence, $\consistency_{1} = \consistency_{2}$ and $\mathcal{L}_{2}^{\bot} = \mathcal{L}_{1}^{\bot}$.
    \end{itemize}
    Ultimately, we have shown that $\mathcal{L}^{\bot}$ is unique up to isomorphism.
    The same reasoning applies when we consider edges instead of vertices.
\end{proof}

Practically, the above proposition means that for any given match $m$ there is only one possible result of rewriting.

% \begin{remark}
%     Consider morphisms $i + j \xrightarrow{[f_{ext},g_{ext}];[f_{int},g_{int}]} \mathcal{L} \xrightarrow{m} \mathcal{G}$.
%     If both of these morphisms are mono, then the boundary complement exists. 
%     Intuitively, it exists for the same reason the corresponding boundary complement for plain hypergraphs exists: it is constructed by removing from $\mathcal{G}$ everything that has no pre-image in $i + j$.
% \end{remark}

% \begin{figure*}[t!]
%     \begin{minipage}{0.1\textwidth}
%         $\;$
%     \end{minipage}
%     \hfill
% \begin{minipage}{0.7\textwidth}
%     \[
%         \scalebox{0.45}{
%             \tikzfig{../figures/combinatorial_semantics/structural_rule_pushout}
%         }
%     \]
%     \caption{The image of the top left equation from Figure~\ref{fig:string-equations} under $\llbracket - \rrbracket$}
%     \label{fig:structural_rule_pushout}
% \end{minipage}
% \hfill
% \begin{minipage}{0.1\textwidth}
%         $\;$
% \end{minipage}
% \end{figure*}

\ifdefined\ONECOLUMN
\section{Proofs for Section \ref{sec:soundness-and-completeness}: Soundness and Completeness}
\else
\subsection{Proofs for Section \ref{sec:soundness-and-completeness}: Soundness and Completeness}
\fi
\label{sec:appendix:detailed_proofs}
In this section we will elaborate on soundness and completeness of our translation $\llbracket - \rrbracket : \catname{S}^{+}(\Sigma, \mathcal{E}) \to \WellTypedMdaEcospans / \mathcal{S}, \mathcal{E}$.
First we recite the interpretation function $[]$ on $\Sigma$-terms as defined in~\cite{bonchi_string_2022-2} that induces a functor $\catname{S}(\Sigma) \to \MdaCospans$ and eventually a functor $\catname{S}^{+}(\Sigma) \to \MdaCospans^{+}$.
The base cases for interpretation function $[-] : \textbf{SMT}(\Sigma) \to \MdaCospans$ is given in the Figure~\ref{fig:base_cases}.
Then it is defined inductively by letting $[f \otimes g] \Coloneqq [f] \otimes [g]$ and $[f;g] \Coloneqq [f];[g]$ and by freeness induces the corresponding functor.

\begin{remark}\label{remark:embedding_functor}
    There is a faithful identity-on-objects functor $E: \HypI{\Sigma} \to \Ecospans$ which maps a cospan $n \xrightarrow{f} \mathcal{G} \xleftarrow{g} m$ to an extended cospan $n \xrightarrow{f_{ext}} n \xrightarrow{f_{int}} \mathcal{G} \xleftarrow{g_{int}} m \xleftarrow{g_{ext}} m$ such that $f_{ext};f_{int} = f$ and $g_{ext};g_{int} = g$.
\end{remark}

\begin{lemma}
\label{lemma:normal_form}
For each cospan $f$ in ${\WellTypedMdaEcospans}/{\mathcal{S}}$ there is a \textit{normal form} such that 
\[
	f = f_1 + \ldots + f_n
\]
	such that each $f_i$ contains no hierarchical edges,  and for all $i \neq j$ we have $f_i \neq f_j$.
\end{lemma}

\begin{figure*}
\[
\adjustbox{width=\textwidth}{
\input{./figures/interpretation.tikz}
}
\]
\caption{Base cases for $[-] : \textbf{SMT}(\Sigma) \to \MdaCospans$}
\label{fig:base_cases}
\end{figure*}

\begin{lemma}[Lemma~\ref{lemma:cospans_plus_equiv} restatement]

    We have the following equivalence of $\catname{SLat}$-categories
    \[
    \MdaCospans^{+} \cong \MdaEcospans/\mathcal{S}~.
    \]
    \end{lemma}
    \begin{proof}
    We will denote the semilattice-enriched functor giving this equivalence as $M$.
    Observe that in $\MdaCospans^{+}$ we can consider only canonical morphisms of the form $f_{1} + \ldots + f_{n}$ where each of $f_{i}$ is an MDA cospan (in the sense of Definition~\ref{def:monogamy_hyp}) and for $i \not = j$, $f_{i} \not = f_{j}$.
    Then, $M$ is defined as follows.
    \begin{itemize}
    \item It maps each discrete hypergraph in $\MdaCospans^{+}$ to a corresponding discrete e-hypergraph in $\MdaEcospans/\mathcal{S}$.
    \item Every morphism $f_{1} + \ldots + f_{n}$ is mapped to $\text{E}(f_{1}) + \ldots + \text{E}(f_{n})$ where $\text{E}$ is the faithful functor induced by Remark~\ref{remark:embedding_functor}.
          This makes $M$ into $\catname{SLat}$-functor by construction.
    \end{itemize}
    To show the functor $M$ is faithful, suppose $f = f_{1} + \ldots + f_{n} \not = g_{1} + \ldots + g_{m} = g$.
    Then it must be that $Mf \not = Mg$.
    If $n \not = m$ then, immediately $Mf \not = Mg$ as $Mf$ and $Mg$ contain different number of non-isomorphic components.
    If $n = m$ then assume for a contradiction that $Mf = Mg$. 
    There exists a permutation $\sigma$, such that for every index $i$, $Mf_{i} = Mg_{\sigma(i)}$.
    Thus, by $\text{E}$ we have that $f_{i} = g_{\sigma(i)}$, implying $f = g$ which is a contradiction.
    Therefore, $Mf \not = Mg$ as required.
    
    To show that the functor is full, consider a morphism $f = f_{1} + \ldots + f_{n}$ in $\WellTypedMdaEcospans / \mathcal{S}$ (by Lemma~\ref{lemma:normal_form}). 
    As each of $f_{i}$ does not contain any hierarchical edges, they are in the image of $\text{E}$, which for each $f_{i}$ gives us $g_{i}$, such that $f_{i} = Mg_{i}$.
    This gives us a morphism $g = g_{1} + \ldots + {g_{n}}$ in $\MdaCospans^{+}$ such that $Mg = f$.
    This concludes that the functor is full.
    As it is also surjective on objects, it makes it into an equivalence.
\end{proof}

\begin{figure}
    \[
    \scalebox{0.75}{
        \input{./figures/contexts.tikz}
    }
    \]
    \caption{String diagrammatic contexts}
    \label{fig:string_contexts}
\end{figure}

\begin{remark}
    The construction in Figure~\ref{fig:A+B} naturally extends to $k$ operands.
    Given a join of $k$ cospans of e-hypegraphs
    \begin{align*}
        n_1 \xrightarrow{} n_1' \xrightarrow{} &\mathcal{F}_1 \xleftarrow{} m_1' \xleftarrow{} m_1\\
        &\;+\\
        &\vdotswithin{+}\\
        &\;+\\
        n_k \xrightarrow{} n_k' \xrightarrow{} &\mathcal{F}_k \xleftarrow{} m_k' \xleftarrow{} m_k
    \end{align*}
    we will denote the carrier of the resulting cospan as
    \[
    \mathcal{F}_1 \; \hat{+}\; \ldots\; \hat{+}\; \mathcal{F}_{k}
    \]
    to distinguish it from the coproduct.
\end{remark}

\begin{remark}
    Below we will also occasionally conflate the cospan resulting from $\llbracket - \rrbracket$ with its carrier which will be clear from the context.
    For example, if $\llbracket f \rrbracket = n \to n' \to \mathcal{F} \xleftarrow{} m' \xleftarrow{} m$, then we will also use $\llbracket f \rrbracket$ to denote $\mathcal{F}$.
\end{remark}

\begin{theorem}[Theorem~\ref{thm:full-completeness} restatement]
    \label{proof:appendix:soundness}
    Let $\mathcal{E}$ be a set of equations $l = r$ of $\Sigma$-terms, then we have the following equivalence of $\catname{SLat}$-categories
    \[
        \catname{S}^{+}(\Sigma, \mathcal {E} ) \cong \WellTypedMdaEcospans / \mathcal{S,E}~.
    \]
    where $\mathcal{E}$ on the right is overloaded as a set of EDPOI rewrite rules containing $\langle \llbracket l \rrbracket, \llbracket r \rrbracket  \rangle$ and $\langle \llbracket r \rrbracket, \llbracket l \rrbracket  \rangle$ for every equation $l = r$.
\end{theorem}
\begin{proof}
    The statement boils down to showing that the $\catname{SLat}$ functor defining the equivalence is well-defined, full and faithful as it is obviously surjective.
    Proof of fullness follows the same argument as in the case of $\llbracket - \rrbracket: \catname{S}^{+}(\Sigma) \to \WellTypedMdaEcospans / \mathcal{S}$, so we proceed with showing well-definedness and faithfulness.     
    The former is equivalent to showing that if $f = g$ in $\catname{S}^{+}(\Sigma)$ modulo $\mathcal{E}$ then $\llbracket f\rrbracket = \llbracket g \rrbracket$ in $\WellTypedMdaEcospans / \mathcal{S,E}~.$ or, equivalently, $\llbracket f \rrbracket \Rrightarrow_{\mathcal{E}}^{*} \llbracket g \rrbracket$ in $\WellTypedMdaEcospans / \mathcal{S}$.
    $f = g$ implies that either
    \begin{itemize}
        \item $f = \mathcal{C}[l] = g$ for some morphism $l \in \catname{S}^{+}(\Sigma)$;
        \item $f = \mathcal{C}[l]$ and $g = \mathcal{C}[r]$ for an equation $l = r$ or $r = l$ from $\mathcal{E}$;
        \item there is a sequence $w = (x_1, \ldots, x_n)$ such that $x_1 = f$ and $x_n = g$ and $x_i = x_{i+1}$ or $x_{i+1} = x_{i}$ in the sense above for $i < n$.
    \end{itemize}
    where $\mathcal{C}$ is a context defined as per the Figure~\ref{fig:string_contexts} and $\mathcal{C}[l]$ is defined by replacing the occurrence of
    \scalebox{0.5}{
    \begin{tikzpicture}[tikzfig]
        \begin{pgfonlayer}{nodelayer}
            \node [style=empty diag] (0) at (-0.5, 0.5) {};
        \end{pgfonlayer}
    \end{tikzpicture}
    }
    in $\mathcal{C}$ with a string diagram $l$.
    Let's prove the statement by induction on $\mathcal{C}$ when $|w| = 2$.
    \begin{itemize}
        \item If $\mathcal{C}$ is empty, then either $f = g$ syntactically and $\llbracket f \rrbracket = \llbracket g \rrbracket$ or $f = l = r = g$ for some $l = r$ or $r = l$.
              The latter trivially yields the EDPO diagram below
              \ifdefined\ONECOLUMN
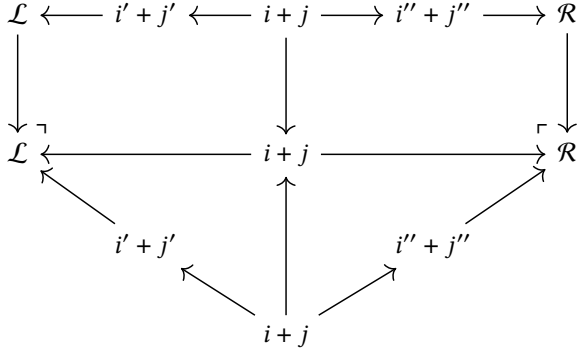
              \[
              \adjustbox{scale=0.8}{
                \begin{tikzcd}
                    {\mathcal{L}} & {i'+j'} & {i+j} & {i''+j''} & {\mathcal{R}} \\
                    \\
                    {\mathcal{L}} && {i+j} && {\mathcal{R}} \\
                    & {i'+j'} && {i''+j''} \\
                    && {i+j}
                    \arrow[from=1-1, to=3-1]
                    \arrow[from=1-2, to=1-1]
                    \arrow[from=1-3, to=1-2]
                    \arrow[from=1-3, to=1-4]
                    \arrow[from=1-3, to=3-3]
                    \arrow[from=1-4, to=1-5]
                    \arrow[from=1-5, to=3-5]
                    \arrow["\lrcorner"{pos=0.05, rotate=90, description}, phantom, from=3-1, to=1-2]
                    \arrow[from=3-3, to=3-1]
                    \arrow[from=3-3, to=3-5]
                    \arrow["\lrcorner"{pos=0.05, rotate=180, description}, phantom, from=3-5, to=1-4]
                    \arrow[from=4-2, to=3-1]
                    \arrow[from=4-4, to=3-5]
                    \arrow[from=5-3, to=3-3]
                    \arrow[from=5-3, to=4-2]
                    \arrow[from=5-3, to=4-4]
                \end{tikzcd}
              }
              \]
              \else
              \[
                \begin{tikzcd}
                    {\mathcal{L}} & {i'+j'} & {i+j} & {i''+j''} & {\mathcal{R}} \\
                    \\
                    {\mathcal{L}} && {i+j} && {\mathcal{R}} \\
                    & {i'+j'} && {i''+j''} \\
                    && {i+j}
                    \arrow[from=1-1, to=3-1]
                    \arrow[from=1-2, to=1-1]
                    \arrow[from=1-3, to=1-2]
                    \arrow[from=1-3, to=1-4]
                    \arrow[from=1-3, to=3-3]
                    \arrow[from=1-4, to=1-5]
                    \arrow[from=1-5, to=3-5]
                    \arrow["\lrcorner"{pos=0.05, rotate=90, description}, phantom, from=3-1, to=1-2]
                    \arrow[from=3-3, to=3-1]
                    \arrow[from=3-3, to=3-5]
                    \arrow["\lrcorner"{pos=0.05, rotate=180, description}, phantom, from=3-5, to=1-4]
                    \arrow[from=4-2, to=3-1]
                    \arrow[from=4-4, to=3-5]
                    \arrow[from=5-3, to=3-3]
                    \arrow[from=5-3, to=4-2]
                    \arrow[from=5-3, to=4-4]
                \end{tikzcd}
              \]
              \fi
              where $i \to i' \to \mathcal{L} \xleftarrow{} j' \xleftarrow{} j = \llbracket l \rrbracket = \llbracket f \rrbracket$ and $i \to i'' \to \mathcal{R} \xleftarrow{} j'' \xleftarrow{} j = \llbracket r \rrbracket = \llbracket g \rrbracket$ and $i + j \to i + j \to \mathcal{L}$ is a boundary complement as per~\ref{def:boundary_new}.
              Hence, $\llbracket f \rrbracket \Rrightarrow{}_{\langle \mathcal{L}, \mathcal{R} \rangle} \llbracket g \rrbracket$.
              \item Suppose 
              \[
                \scalebox{0.75}{\input{./figures/contexts_2.tikz}}
              \]
              and by inductive hypothesis $\llbracket \mathcal{C}[l] \rrbracket \Rrightarrow{} \llbracket \mathcal{C}[r] \rrbracket$ which results in the top half of the diagram below.
              
              \ifdefined\ONECOLUMN
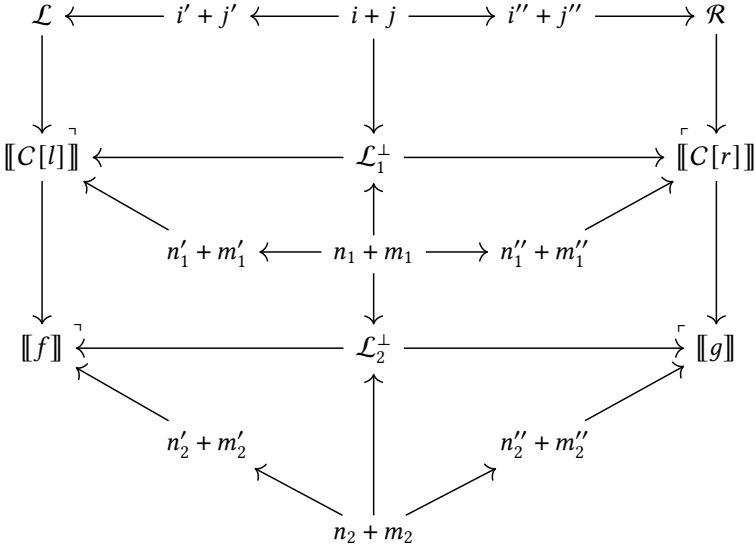
              \[
              \adjustbox{width=0.7\linewidth}{
              \begin{tikzcd}
                {\mathcal{L}} & {i'+j'} & {i+j} & {i''+j''} & {\mathcal{R}} \\
                \\
                {\llbracket \mathcal{C}[l]\rrbracket} && {\mathcal{L}_1^{\bot}} && {\llbracket \mathcal{C}[r]\rrbracket} \\
                & {n_1'+m_1'} & {n_1+m_1} & {n_1''+m_1''} \\
                {\llbracket f \rrbracket} && {\mathcal{L}_2^{\bot}} && {\llbracket g \rrbracket} \\
                & {n_2' + m_2'} && {n_2''+m_2''} \\
                && {n_2 + m_2}
                \arrow[from=1-1, to=3-1]
                \arrow[from=1-2, to=1-1]
                \arrow[from=1-3, to=1-2]
                \arrow[from=1-3, to=1-4]
                \arrow[from=1-3, to=3-3]
                \arrow[from=1-4, to=1-5]
                \arrow[from=1-5, to=3-5]
                \arrow["\lrcorner"{pos=0.025, rotate=90, description}, draw=none, from=3-1, to=1-2]
                \arrow[from=3-1, to=5-1]
                \arrow[from=3-3, to=3-1]
                \arrow[from=3-3, to=3-5]
                \arrow["\lrcorner"{pos=0.025, rotate=180, description}, draw=none, from=3-5, to=1-4]
                \arrow[from=3-5, to=5-5]
                \arrow[from=4-2, to=3-1]
                \arrow[from=4-3, to=3-3]
                \arrow[from=4-3, to=4-2]
                \arrow[from=4-3, to=4-4]
                \arrow[from=4-3, to=5-3]
                \arrow[from=4-4, to=3-5]
                \arrow["\lrcorner"{pos=0.025, rotate=90, description}, draw=none, from=5-1, to=4-2]
                \arrow[from=5-3, to=5-1]
                \arrow[from=5-3, to=5-5]
                \arrow["\lrcorner"{pos=0.025, rotate=180, description}, draw=none, from=5-5, to=4-4]
                \arrow[from=6-2, to=5-1]
                \arrow[from=6-4, to=5-5]
                \arrow[from=7-3, to=5-3]
                \arrow[from=7-3, to=6-2]
                \arrow[from=7-3, to=6-4]
            \end{tikzcd}}
            \]
              \else
              \[
              \adjustbox{width=\linewidth}{
              \begin{tikzcd}
                {\mathcal{L}} & {i'+j'} & {i+j} & {i''+j''} & {\mathcal{R}} \\
                \\
                {\llbracket \mathcal{C}[l]\rrbracket} && {\mathcal{L}_1^{\bot}} && {\llbracket \mathcal{C}[r]\rrbracket} \\
                & {n_1'+m_1'} & {n_1+m_1} & {n_1''+m_1''} \\
                {\llbracket f \rrbracket} && {\mathcal{L}_2^{\bot}} && {\llbracket g \rrbracket} \\
                & {n_2' + m_2'} && {n_2''+m_2''} \\
                && {n_2 + m_2}
                \arrow[from=1-1, to=3-1]
                \arrow[from=1-2, to=1-1]
                \arrow[from=1-3, to=1-2]
                \arrow[from=1-3, to=1-4]
                \arrow[from=1-3, to=3-3]
                \arrow[from=1-4, to=1-5]
                \arrow[from=1-5, to=3-5]
                \arrow["\lrcorner"{pos=0.025, rotate=90, description}, draw=none, from=3-1, to=1-2]
                \arrow[from=3-1, to=5-1]
                \arrow[from=3-3, to=3-1]
                \arrow[from=3-3, to=3-5]
                \arrow["\lrcorner"{pos=0.025, rotate=180, description}, draw=none, from=3-5, to=1-4]
                \arrow[from=3-5, to=5-5]
                \arrow[from=4-2, to=3-1]
                \arrow[from=4-3, to=3-3]
                \arrow[from=4-3, to=4-2]
                \arrow[from=4-3, to=4-4]
                \arrow[from=4-3, to=5-3]
                \arrow[from=4-4, to=3-5]
                \arrow["\lrcorner"{pos=0.025, rotate=90, description}, draw=none, from=5-1, to=4-2]
                \arrow[from=5-3, to=5-1]
                \arrow[from=5-3, to=5-5]
                \arrow["\lrcorner"{pos=0.025, rotate=180, description}, draw=none, from=5-5, to=4-4]
                \arrow[from=6-2, to=5-1]
                \arrow[from=6-4, to=5-5]
                \arrow[from=7-3, to=5-3]
                \arrow[from=7-3, to=6-2]
                \arrow[from=7-3, to=6-4]
            \end{tikzcd}}
            \]
            \fi
            Clearly the occurrence of $\llbracket \mathcal{C}[l] \rrbracket$ in $\llbracket f \rrbracket$ is convex and down-closed and by following the argument in Theorem 35~\cite{bonchi_string_2022-2}, there exists $n_1 + m_1 \xrightarrow{} n_1 + m_1 \xrightarrow{} \mathcal{L}_2^{\bot} \xleftarrow{} n_3 + m_3 \xleftarrow{} n_2 + m_2$ such that
            \ifdefined \ONECOLUMN
            \begin{align*}
                \llbracket f \rrbracket &= \;
                (0 \to 0 \to \llbracket C[l] \rrbracket \xleftarrow{} n_1' + m_1' \xleftarrow{} n_1 + m_1)
                ;
                (n_1 + m_1 \xrightarrow{} n_1 + m_1 \xrightarrow{} \mathcal{L}_2^{\bot} \xleftarrow{} n_3 + m_3 \xleftarrow{} n_2 + m_2)\\
                \llbracket g \rrbracket &= \;
                (0 \to 0 \to \llbracket C[r] \rrbracket \xleftarrow{} n_1' + m_1' \xleftarrow{} n_1 + m_1)
                ;
                (n_1 + m_1 \xrightarrow{} n_1 + m_1 \xrightarrow{} \mathcal{L}_2^{\bot} \xleftarrow{} n_3 + m_3 \xleftarrow{} n_2 + m_2)
                \end{align*}
            \else
            \begin{align*}
            \llbracket f \rrbracket = \;
            &(0 \to 0 \to \llbracket C[l] \rrbracket \xleftarrow{} n_1' + m_1' \xleftarrow{} n_1 + m_1)\\
            &;\\
            &(n_1 + m_1 \xrightarrow{} n_1 + m_1 \xrightarrow{} \mathcal{L}_2^{\bot} \xleftarrow{} n_3 + m_3 \xleftarrow{} n_2 + m_2)\\
            \llbracket g \rrbracket = \;
            &(0 \to 0 \to \llbracket C[r] \rrbracket \xleftarrow{} n_1' + m_1' \xleftarrow{} n_1 + m_1)\\
            &;\\
            &(n_1 + m_1 \xrightarrow{} n_1 + m_1 \xrightarrow{} \mathcal{L}_2^{\bot} \xleftarrow{} n_3 + m_3 \xleftarrow{} n_2 + m_2)
            \end{align*}
            \fi
            and
            % $n'_2 + m'_2 = n_3 + m_3 + (n_1' + m_1' \setminus (n_1 + m_1))$
            \[
                n_2' = n_3 + (n_1' \setminus n_1) \qquad m_2' = m_3 + (m_1' \setminus m_1)
            \]
            and
            % fix overflowing
            % \[
            % n_2 + m_1 \xrightarrow{} n_2' \setminus (n_1' \setminus n_1) + m_1 \xrightarrow{} \mathcal{L}_{2}^{\bot} \xleftarrow{} m'_2 \setminus (m_1' \setminus m_1) + n_1 \xleftarrow{} m_2 + n_1    
            % \]
            \tikzcdset{row sep/mysizerow/.initial=0.3ex}
            \tikzcdset{column sep/mysizecol/.initial=0.1ex}
            \[\
\begin{tikzcd}[cramped, row sep=mysizerow, column sep = mysizecol]
                n_2 + m_1\hspace{-2em} &&&&&&&& \hspace{-2em}m_2 + n_1\\
                & \hspace{-2em} n_2' \setminus (n_1' \setminus n_1) + m_1 &\; & & \mathcal{L}_{2}^{\bot} &\; && m'_2 \setminus (m_1' \setminus m_1) + n_1\hspace{-2em} &\\
                \arrow[from=1-1, to=2-2]
                \arrow[from=1-9, to=2-8]
                \arrow[from=2-2, to=2-5]
                \arrow[from=2-8, to=2-5]
            \end{tikzcd}
\]
            is an MDA cospan.
            By denoting the pushout of $\mathcal{L}_1^{\bot} \xleftarrow{} n_1 + m_1 \xrightarrow{} \mathcal{L}_{2}^{\bot}$ with $\mathcal{L}_3^{\bot}$,
            one can observe that $\llbracket f \rrbracket$ is the pushout object of $\mathcal{L} \xleftarrow{} i+j \xrightarrow{} \mathcal{L}_{3}^{\bot}$ as per the diagram below
            \[\begin{tikzcd}
                && {n_1+m_1} && {\mathcal{L}^{\bot}_2} \\
                \\
                {i+j} && {\mathcal{L}^{\bot}_1} && {\mathcal{L}^{\bot}_3} \\
                \\
                {\mathcal{L}} && {\mathcal{C}[l]} && {\llbracket f \rrbracket}
                \arrow[from=1-3, to=1-5]
                \arrow[from=1-3, to=3-3]
                \arrow[from=1-5, to=3-5]
                \arrow[from=3-1, to=3-3]
                \arrow[from=3-1, to=5-1]
                \arrow[from=3-3, to=3-5]
                \arrow[from=3-3, to=5-3]
                \arrow["\lrcorner"{description, pos=0.025, rotate=180}, draw=none, from=3-5, to=1-3]
                \arrow[from=3-5, to=5-5]
                \arrow[from=5-1, to=5-3]
                \arrow["\lrcorner"{description, pos=0.025, rotate=180}, draw=none, from=5-3, to=3-1]
                \arrow[from=5-3, to=5-5]
                \arrow["\lrcorner"{description, pos=0.025, rotate=180}, draw=none, from=5-5, to=3-3]
                \arrow["\dagger", draw=none, from=5-5, to=3-3]
            \end{tikzcd}
            \]
            where the square marked with $\dagger$ is a pushout because of pushout pasting law.
            The corestriction of $n_1 + m_1 \xrightarrow{[f_1,f_2]} \mathcal{L}_1^{\bot}$ to the image of $i + j \xrightarrow{[g_1,g_2]} \mathcal{L}_{1}^{\bot}$ is mono: suppose it is not mono, then there exist $z_1$, $z_2$, $z_3$ such that $f_1(z_1) = f_2(z_2) = g_1(z_3)$, then $[f_2,g_1]$ would not be mono, similarly if $f_1(z_1) = f_2(z_2) = g_2(z_3)$.\
            By construction of the pushout the restriction of $\mathcal{L}_{1}^{\bot} \to \mathcal{L}_{3}^{\bot}$ is mono and hence the arrow $i + j \to \mathcal{L}_{3}^{\bot}$ is mono.
            This implies that $\llbracket f \rrbracket \Rrightarrow_{\langle \mathcal{L},\mathcal{R} \rangle} \llbracket g \rrbracket$.
            \item Consider the case when
            \[
            \hspace{-1em}\scalebox{0.75}{\input{./figures/contexts_3.tikz}}    
            \]
            By hypothesis, there exists a diagram
            \ifdefined\ONECOLUMN
            \[
                \adjustbox{width=0.7\linewidth}{
                \begin{tikzcd}
                    {\mathcal{L}} & {i'+j'} & {i+j} & {i''+j''} & {\mathcal{R}} \\
                    \\
                    {\llbracket\mathcal{C}[l]\rrbracket} && {\mathcal{L}_1^{\bot}} && {\llbracket\mathcal{C}[r]\rrbracket} \\
                    & {n'+m'} && {n''+m''} \\
                    && {n+m}
                    \arrow[from=1-1, to=3-1]
                    \arrow[from=1-2, to=1-1]
                    \arrow[from=1-3, to=1-2]
                    \arrow[from=1-3, to=1-4]
                    \arrow[from=1-3, to=3-3]
                    \arrow[from=1-4, to=1-5]
                    \arrow[from=1-5, to=3-5]
                    \arrow[from=3-3, to=3-1]
                    \arrow[from=3-3, to=3-5]
                    \arrow[from=4-2, to=3-1]
                    \arrow[from=4-4, to=3-5]
                    \arrow[from=5-3, to=3-3]
                    \arrow[from=5-3, to=4-2]
                    \arrow[from=5-3, to=4-4]
                    \arrow["\lrcorner"{description, pos=0.025, rotate=90}, draw=none, from=3-1, to=1-2]
                    \arrow["\lrcorner"{description, pos=0.025, rotate=180}, draw=none, from=3-5, to=1-4]
                \end{tikzcd}}
                \]
            \else
                \[
                \adjustbox{width=\linewidth}{
                \begin{tikzcd}
                    {\mathcal{L}} & {i'+j'} & {i+j} & {i''+j''} & {\mathcal{R}} \\
                    \\
                    {\llbracket\mathcal{C}[l]\rrbracket} && {\mathcal{L}_1^{\bot}} && {\llbracket\mathcal{C}[r]\rrbracket} \\
                    & {n'+m'} && {n''+m''} \\
                    && {n+m}
                    \arrow[from=1-1, to=3-1]
                    \arrow[from=1-2, to=1-1]
                    \arrow[from=1-3, to=1-2]
                    \arrow[from=1-3, to=1-4]
                    \arrow[from=1-3, to=3-3]
                    \arrow[from=1-4, to=1-5]
                    \arrow[from=1-5, to=3-5]
                    \arrow[from=3-3, to=3-1]
                    \arrow[from=3-3, to=3-5]
                    \arrow[from=4-2, to=3-1]
                    \arrow[from=4-4, to=3-5]
                    \arrow[from=5-3, to=3-3]
                    \arrow[from=5-3, to=4-2]
                    \arrow[from=5-3, to=4-4]
                    \arrow["\lrcorner"{description, pos=0.025, rotate=90}, draw=none, from=3-1, to=1-2]
                    \arrow["\lrcorner"{description, pos=0.025, rotate=180}, draw=none, from=3-5, to=1-4]
                \end{tikzcd}}
                \]
                \fi
                which implies the existence of the following diagram where both squares are pushouts by construction and the matching is convex and down-closed.
                % \textcolor{red}{Consider cases when $\mathcal{L}$ has no-input and no-output sub-hypergraph and when it does not. Basically, repeat the above reasoning here as well, but say that if $n+m$ is empty, then there is no pushout.}
                \ifdefined \ONECOLUMN
                \[\adjustbox{width=0.95\linewidth}{
                    \begin{tikzcd}
                    {\mathcal{L}} & {i'+j'} & {i+j} & {i''+j''} & {\mathcal{R}} \\
                    \\
                    {\llbracket f_1 \rrbracket \; \hat{+}\; \ldots \hat{+} \; \llbracket\mathcal{C}[l]\rrbracket \; \hat{+} \; \ldots \; \hat{+} \; \llbracket f_n \rrbracket} && {\llbracket f_1 \rrbracket \; \hat{+} \; \ldots \; \hat{+} \; \mathcal{L}_{1}^{\bot} \; \hat{+} \; \ldots \; \hat{+} \; \llbracket f_n \rrbracket} && {\llbracket f_1 \rrbracket \; \hat{+} \; \ldots \; \hat{+} \; \llbracket\mathcal{C}[r]\rrbracket \; \hat{+} \; \ldots \; \hat{+} \; \llbracket f_n \rrbracket} \\
                    & {n_1' + m_1'} && {n_1'' + m_1''} \\
                    && {n + m}
                    \arrow[from=1-1, to=3-1]
                    \arrow[from=1-2, to=1-1]
                    \arrow[from=1-3, to=1-2]
                    \arrow[from=1-3, to=1-4]
                    \arrow[from=1-3, to=3-3]
                    \arrow[from=1-4, to=1-5]
                    \arrow[from=1-5, to=3-5]
                    \arrow["\lrcorner"{description, pos=0.025, rotate=90}, draw=none, from=3-1, to=1-2]
                    \arrow[from=3-3, to=3-1]
                    \arrow[from=3-3, to=3-5]
                    \arrow["\lrcorner"{description, pos=0.025, rotate=180}, draw=none, from=3-5, to=1-4]
                    \arrow[from=4-2, to=3-1]
                    \arrow[from=4-4, to=3-5]
                    \arrow[from=5-3, to=3-3]
                    \arrow[from=5-3, to=4-2]
                    \arrow[from=5-3, to=4-4]
                \end{tikzcd}
                }
                \]
                \else
                \begin{figure*}[hbt!]
                \[\adjustbox{width=\linewidth}{
                    \begin{tikzcd}
                    {\mathcal{L}} & {i'+j'} & {i+j} & {i''+j''} & {\mathcal{R}} \\
                    \\
                    {\llbracket f_1 \rrbracket \; \hat{+}\; \ldots \hat{+} \; \llbracket\mathcal{C}[l]\rrbracket \; \hat{+} \; \ldots \; \hat{+} \; \llbracket f_n \rrbracket} && {\llbracket f_1 \rrbracket \; \hat{+} \; \ldots \; \hat{+} \; \mathcal{L}_{1}^{\bot} \; \hat{+} \; \ldots \; \hat{+} \; \llbracket f_n \rrbracket} && {\llbracket f_1 \rrbracket \; \hat{+} \; \ldots \; \hat{+} \; \llbracket\mathcal{C}[r]\rrbracket \; \hat{+} \; \ldots \; \hat{+} \; \llbracket f_n \rrbracket} \\
                    & {n_1' + m_1'} && {n_1'' + m_1''} \\
                    && {n + m}
                    \arrow[from=1-1, to=3-1]
                    \arrow[from=1-2, to=1-1]
                    \arrow[from=1-3, to=1-2]
                    \arrow[from=1-3, to=1-4]
                    \arrow[from=1-3, to=3-3]
                    \arrow[from=1-4, to=1-5]
                    \arrow[from=1-5, to=3-5]
                    \arrow["\lrcorner"{description, pos=0.025, rotate=90}, draw=none, from=3-1, to=1-2]
                    \arrow[from=3-3, to=3-1]
                    \arrow[from=3-3, to=3-5]
                    \arrow["\lrcorner"{description, pos=0.025, rotate=180}, draw=none, from=3-5, to=1-4]
                    \arrow[from=4-2, to=3-1]
                    \arrow[from=4-4, to=3-5]
                    \arrow[from=5-3, to=3-3]
                    \arrow[from=5-3, to=4-2]
                    \arrow[from=5-3, to=4-4]
                \end{tikzcd}}
                \]
                \caption{$\llbracket f \rrbracket \Rrightarrow{}_{\langle \mathcal{L},\mathcal{R} \rangle} \llbracket g \rrbracket$}
                \label{fig:f_rewrites_to_g_under_plus}
            \end{figure*}
            \fi
                and hence we have $\llbracket f \rrbracket \Rrightarrow_{\langle \mathcal{L}, \mathcal{R} \rangle} \llbracket g \rrbracket$
    \end{itemize}
    The rest of the argument follows by induction on the length of $w$ and ultimately we have shown that if $f = g$ modulo $\mathcal{E}$ then $\llbracket f \rrbracket = \llbracket g \rrbracket$  in $\WellTypedMdaEcospans / \mathcal{S}, \mathcal{E}$.

    Then, faithfulness amounts to showing that if $\llbracket f \rrbracket \Rrightarrow_{\llbracket \mathcal{E} \rrbracket}^{*} \llbracket g \rrbracket$ in $\WellTypedMdaEcospans / \mathcal{S}$ then $f = g$ modulo $\mathcal{E}$ in $\catname{S}^{+}(\Sigma)$.
% By the theorem~\ref{thm:completeness_simple} if $\llbracket f \rrbracket = \llbracket g \rrbracket$ via rules from $\mathcal{S}$ then $f = g$, so we will show that the same holds when $\llbracket f \rrbracket = \llbracket g \rrbracket$ via rules from $\mathcal{E}$.
Consider a rewrite rule $\langle \llbracket l \rrbracket, \llbracket r \rrbracket \rangle = \langle \mathcal{L}, \mathcal{R} \rangle$ such that $\llbracket f \rrbracket \Rrightarrow_{\langle \mathcal{L}, \mathcal{R} \rangle} \llbracket g \rrbracket$ and suppose
\[
	\llbracket f \rrbracket = n \xrightarrow{} n' \xrightarrow{} \mathcal{F} \xleftarrow{} m' \xleftarrow{} m ~.
\]

\begin{figure*}
    \begin{subfigure}{\linewidth}
\begin{align*}
    (k \xrightarrow{id_{k}} k \xrightarrow{id_{k}} &k \xleftarrow{id_{k}} k \xleftarrow{id_{k}} k)\\
    n \xrightarrow{} n_{1}' \xrightarrow{} \mathcal{C}_{1} \xleftarrow{} m_{1}' + i + k \xleftarrow{} i + k; \qquad\qquad\qquad &\otimes \qquad\qquad\qquad ; j + k \xrightarrow{} j + k + n'_{2} \xrightarrow{} \mathcal{C}_{2} \xleftarrow{} m'_{2} \xleftarrow{} m\\
    (i \xrightarrow{} i' \xrightarrow{} &\mathcal{L} \xleftarrow{} j' \xleftarrow{} j)
\end{align*}
\subcaption{\;}
\end{subfigure}
\vspace{1ex}
\begin{subfigure}{\linewidth}
    \begin{align*}
        (k \xrightarrow{id_{k}} k \xrightarrow{id_{k}} &k \xleftarrow{id_{k}} k \xleftarrow{id_{k}} k)\\
        n \xrightarrow{} n_{1}' \xrightarrow{} \mathcal{C}_{1} \xleftarrow{} m_{1}' + i + k \xleftarrow{} i + k; \qquad\qquad\qquad &\otimes \qquad\qquad\qquad ; j + k \xrightarrow{} j + k + n'_{2} \xrightarrow{} \mathcal{C}_{2} \xleftarrow{} m'_{2} \xleftarrow{} m\\
        (i \xrightarrow{} i' \xrightarrow{} &\mathcal{F'} \xleftarrow{} j' \xleftarrow{} j)
    \end{align*}
\subcaption{\;}
\end{subfigure}
\caption{Decomposition}
\label{fig:decomposition_1}
\end{figure*}

This implies the existence of a diagram as in~\ref{def:dpoi-e}.
We will proceed by induction on the level of the occurrence of $\mathcal{L}$ within $\mathcal{F}$.
\begin{itemize}
	\item Suppose the image of $\mathcal{L}$ in $\mathcal{F}$ is top-level.
		  Then, by lemma~\ref{lemma:decomposition}, there exists a decomposition of $\llbracket f \rrbracket$ in terms of MDA cospans as in Figure~\ref{fig:decomposition_1}.

		% \begin{align*}
		% 	(k \xrightarrow{id_{k}} k \xrightarrow{id_{k}} &k \xleftarrow{id_{k}} k \xleftarrow{id_{k}} k)\\
		% 	n \xrightarrow{} n_{1}' \xrightarrow{} \mathcal{C}_{1} \xleftarrow{} m_{1}' + i + k \xleftarrow{} i + k; \qquad\qquad\qquad &\otimes \qquad\qquad\qquad ; j + k \xrightarrow{} j + k + n'_{2} \xrightarrow{} \mathcal{C}_{2} \xleftarrow{} m'_{2} \xleftarrow{} m\\
		% 	(i \xrightarrow{} i' \xrightarrow{} &\mathcal{L} \xleftarrow{} j' \xleftarrow{} j)
		% \end{align*}
		By fullness of $\llbracket - \rrbracket$ (from Theorem~\ref{thm:completeness_simple}) there exist $c_{1}$ and $c_{2}$ such that 
		\begin{align*}
		&\llbracket c_{1} \rrbracket = n \xrightarrow{} n_{1}' \xrightarrow{} \mathcal{C}_{1} \xleftarrow{} m_{1}' + i + k \xleftarrow{} i + k\\
        &\llbracket c_{2} \rrbracket = j + k \xrightarrow{} j + k + n'_{2} \xrightarrow{} \mathcal{C}_{2} \xleftarrow{} m'_{2} \xleftarrow{} m
        \end{align*}
		By functoriality, $\llbracket f \rrbracket = \llbracket c_{1};(id \otimes l);c_{2} \rrbracket$ and by faithfulness we have that $f = c_{1};(id \otimes l);c_{2}$ to which we can apply $\langle l, r \rangle$ and get $g' = c_{1};(id \otimes r);c_{2}$.
		We then have $\llbracket g' \rrbracket = \llbracket c_{1};(id \otimes r);c_{2} \rrbracket =  \llbracket g \rrbracket$ since the boundary complement is unique (Proposition ~\ref{prop:boundary_unique}) and $g' = g$ by faithfulness.

		\item If the image of $\mathcal{L}$ is not top-level, then either $\llbracket f \rrbracket$ is decomposable as in Figure~\ref{fig:decomposition_2}
		% \begin{align*}
		% 	(k \xrightarrow{id_{k}} k \xrightarrow{id_{k}} &k \xleftarrow{id_{k}} k \xleftarrow{id_{k}} k)\\
		% 	n \xrightarrow{} n_{1}' \xrightarrow{} \mathcal{C}_{1} \xleftarrow{} m_{1}' + i + k \xleftarrow{} i + k; \qquad\qquad\qquad &\otimes \qquad\qquad\qquad ; j + k \xrightarrow{} j + k + n'_{2} \xrightarrow{} \mathcal{C}_{2} \xleftarrow{} m'_{2} \xleftarrow{} m\\
		% 	(i \xrightarrow{} i' \xrightarrow{} &\mathcal{F'} \xleftarrow{} j' \xleftarrow{} j)
		% \end{align*}
		such that $\mathcal{L}$ is a convex down-closed sub-e-hypergraph of $F'$.
		Since the image of $\mathcal{L}$ lies within $\mathcal{F}'$, there is an EDPO rewrite step that turns
		\[
			(i \xrightarrow{} i' \xrightarrow{} \mathcal{F'} \xleftarrow{} j' \xleftarrow{} j)	
		\]
		into
		\[
			(i \xrightarrow{} i'' \xrightarrow{} \mathcal{G'} \xleftarrow{} j'' \xleftarrow{} j)	
		\]
		via $\langle \mathcal{L}, \mathcal{R} \rangle$.
		By fullness and inductive hypothesis, there exist $f'$ and $g'$ such that 
		\begin{align*}
			&\llbracket f' \rrbracket = i \xrightarrow{} i' \xrightarrow{} \mathcal{F'} \xleftarrow{} j' \xleftarrow{} j\\
            &\llbracket g' \rrbracket = (i \xrightarrow{} i'' \xrightarrow{} \mathcal{G'} \xleftarrow{} j'' \xleftarrow{} j)
        \end{align*}
		and that $f' = g'$.
		By decomposition and fullness we have $\llbracket f \rrbracket = \llbracket c_{1};(id \otimes f');c_{2} \rrbracket$ and $f = c_{1};(id \otimes f');c_{2}$ by faithfulness.
		By applying $f' = g'$ we get $f = g'' = c_{1};(id \otimes g');c_{2}$.
		Finally, $\llbracket g'' \rrbracket = \llbracket c_{1};(id \otimes g');c_{2} \rrbracket = \llbracket g \rrbracket$ (by the uniqueness of the boundary complement) and $g'' = g = f$ by faithfulness.

		\item or, $\llbracket f \rrbracket$ is decomposable as
		\begin{align*}
			n_1 \xrightarrow{} n_1' \xrightarrow{} &\mathcal{F}_1 \xleftarrow{} m_1' \xleftarrow{} m_1\\
			&\;+\\
			&\vdotswithin{+}\\
			&\;+\\
			n_p \xrightarrow{} n_p' \xrightarrow{} &\mathcal{F}_p \xleftarrow{} m_p' \xleftarrow{} m_p\\
			&\;+\\
			&\vdotswithin{+}\\
			&\;+\\
			n_k \xrightarrow{} n_k' \xrightarrow{} &\mathcal{F}_k \xleftarrow{} m_k' \xleftarrow{} m_k
		\end{align*}
		such that the image of $\mathcal{L}$ lies within $\mathcal{F}_{p}$.
		By inductive hypothesis we have that $\llbracket f_{p} \rrbracket = n_p \xrightarrow{} n_{p}' \xrightarrow{} \mathcal{F}_{p} \xleftarrow{} m_{p}' \xleftarrow{} m_{p}$ rewrites into $\llbracket g_{p} \rrbracket = n_p \xrightarrow{} n_{p}'' \xrightarrow{} \mathcal{G}_{p} \xleftarrow{} m_{p}'' \xleftarrow{} m_{p}$ such that $f_{p} = g_{p}$.
		By fullness and the fact that $\llbracket - \rrbracket$ is $\catname{SLat}$-enriched, there exist $\llbracket f_1 + \ldots + f_{p} + \ldots + f_{k} \rrbracket = \llbracket f_{1} \rrbracket + \ldots + \llbracket f_{p} \rrbracket + \ldots + \llbracket f_{k} \rrbracket = \llbracket f \rrbracket$.
		By applying $f_{p} = g_{p}$ we get $f = g''$ and $\llbracket g'' \rrbracket = \llbracket f_{1} + \ldots + g_{p} + \ldots f_{k} \rrbracket = \llbracket g \rrbracket$ and $f = g'' = g$.
\end{itemize}
\end{proof}

\begin{lemma}
\label{lemma:decomposition}
    Let $n \xrightarrow{f} n' \xrightarrow{f'} \mathcal{G} \xleftarrow{g'} m' \xleftarrow{g} m$ be an MDA extended cospan.
    If $\mathcal{L}$ is a convex down-closed sub-e-hypergraph of $\mathcal{G}$ such that the immediate predecessors of top-level connected components of $\mathcal{L}$ in $\mathcal{G}$ are either undefined or coincide and pair-wise consistent then either
    \begin{itemize}
    \item there exists a cospan $i \to i' \to \mathcal{G}' \xleftarrow{} j' \xleftarrow{} j$ and $k \in \mathbb{N}$ such that $\mathcal{G}$ is decomposable as
    \begin{align}
        \label{fig:decomposition:1}
    \scalebox{0.6}{
        \input{./figures/G_decomposition_1.tikz}
    }
    \end{align}
    \item or, $\mathcal{G}$ is decomposable as
    \begin{align}
        \label{fig:decomposition:2}
    \scalebox{0.6}{
        \input{./figures/G_decomposition_2.tikz}
    }
    \end{align}
    \end{itemize}
    such that $\mathcal{L}$ is a convex down-closed sub-e-hypergraph of $\mathcal{G}'$ and all the cospans are MDA.
\end{lemma}
\begin{proof}
    Suppose the image of top-level components of $\mathcal{L}$ consists of top-level edges and vertices.
    Then let $\mathcal{C}_1$ be the smallest e-hypergraph consisting of the input vertices of $\mathcal{G}$ and every hyperedge $h$ (and its successors) that is not in $\mathcal{L}$ but has a path to it.
    Let $\mathcal{C}_{2}$ be the smallest e-hypergraph such that $\mathcal{G} = \mathcal{C}_{1} \cup \mathcal{L} \cup \mathcal{C}_{2}$.
    These e-hypergraphs overlap only on vertices so we define
    \begin{align*}
        i = V_{\mathcal{C}_{1}} \cap V_{\mathcal{L}}\\
        j = V_{\mathcal{C}_{2}} \cap V_{\mathcal{L}}\\
        k = (V_{\mathcal{C}_{1}} \cap V_{\mathcal{C}_{2}}) \setminus V_{\mathcal{L}}
    \end{align*}
    and let $i'$ be the union of $i$ and the vertices of $\mathcal{L}$ in the image of $f'$ and $j'$ be the union of $j$ and the vertices of $\mathcal{L}$ in the image of $g'$.
    Let $n_{1}'$ be pre-image of the co-restriction of $f'$ to $\mathcal{C}_{1}$ and $n_{2}'$ be the pre-image of the co-restriction of $f'$ to $\mathcal{C}_{2}$.
    Similarly, for $m_{1}'$ and $m_{2}'$ with regard to $g'$.
    Then there are cospans
    \begin{align*}
        n \to n_{1}' \to &\mathcal{C}_{1} \xleftarrow{} m_{1}' + k + i \xleftarrow{} k + i\\
        i \to i' \to &\mathcal{L} \xleftarrow{} j' \xleftarrow{} j\\
        k + j \to n_{2}' + k + j \to &\mathcal{C}_{2} \xleftarrow{} m_{2}' \xleftarrow{} m\\
    \end{align*}
    and $n \xrightarrow{f} n' \xrightarrow{f'} \mathcal{G} \xleftarrow{g'} m' \xleftarrow{g} m$ is the co-limit of the following diagram
    % \[
    %     n \to n_{1}' \to \mathcal{C}_{1} \xleftarrow{} m_{1}' + k \xleftarrow{} i + k \to i' + k \to \mathcal{L} \xleftarrow{} k + j' \xleftarrow{} k + j \xrightarrow{} n_{2}' + k + j \to \mathcal{C}_{2} \xleftarrow{} m_{2}' \xleftarrow{} m
    % \]
    \ifdefined\ONECOLUMN
    \[
    n \xrightarrow{} n'_{1} \xrightarrow{} \mathcal{C}_{1} \xleftarrow{} m'_{1} + i + k \xleftarrow{} i+k \xrightarrow{} i' + k \xrightarrow{} \mathcal{L} \xleftarrow{} k + j' \xleftarrow{} k + j \xrightarrow{} n'_2 + k + j \xrightarrow{} \mathcal{C}_{2} \xleftarrow{} m'_2 \xleftarrow{} m
    \]
    \else
    \[
    \adjustbox{width=\linewidth}{
    \begin{tikzcd}
        n \arrow[r] & n'_{1} \arrow[r] & \mathcal{C}_{1} & m'_{1} + i + k \arrow[l]     & i+k \arrow[l] \arrow[r] & i'+k \arrow[rd] &             \\
                    &                  &                 &                          &                         &                 & \mathcal{L} \\
        m \arrow[r] & m'_{2} \arrow[r] & \mathcal{C}_{2} & n'_{2} + k + j \arrow[l] & k+j \arrow[r] \arrow[l] & k+j' \arrow[ru] &            
    \end{tikzcd}
    }
    \]
    \fi
    The two spans identify precisely those nodes from $\mathcal{G}$ that occur in more than one sub-hypergraph, so this amounts to simply taking the union
    \ifdefined\ONECOLUMN
    \begin{align*}
        n \to n_{1}' \cup n_{2}' \cup i' \to \mathcal{C}_{1} \cup \mathcal{C}_{2} \cup \mathcal{L} \xleftarrow{} m_{1}' \cup m_{2}' \cup j' \xleftarrow{} m  = n \xrightarrow{f} n' \xrightarrow{f'} \mathcal{G} \xleftarrow{g'} m' \xleftarrow{g} m
    \end{align*}
    \else
    \begin{align*}
    &n \to n_{1}' \cup n_{2}' \cup i' \to \mathcal{C}_{1} \cup \mathcal{C}_{2} \cup \mathcal{L} \xleftarrow{} m_{1}' \cup m_{2}' \cup j' \xleftarrow{} m  =\\
    &\;= n \xrightarrow{f} n' \xrightarrow{f'} \mathcal{G} \xleftarrow{g'} m' \xleftarrow{g} m
    \end{align*}
    \fi
    If we let $\mathcal{G}' = \mathcal{L}$ we get the decomposition as in~(\ref{fig:decomposition:1}).
    This follows the construction of a similar lemma in~\cite{bonchi_string_2022-1} (Lemma 24). 
    The cospans with $\mathcal{C}_{1}$ and $\mathcal{C}_{2}$ are MDAs because they are closed under successors and predecessors and the cospan with $\mathcal{L}$ is an MDA because it is convex.

    Now suppose that the image of top-level components of $\mathcal{L}$ consist of non top-level edges and vertices $e$ and $v$.
    Take a hyperedge $h$ such that $h \leq e$ and for all $e' \not = h$ such that $e' \leq e$, $h \leq e'$.
    Let $\mathcal{G}''$ be $h$ including all its sources and targets and successors.
    If $\mathcal{G} = \mathcal{G}''$ then it is decomposable as in~(\ref{fig:decomposition:2}) and let $\mathcal{G}'$ be the component that contains the image of $\mathcal{L}$ which exists by pair-wise consistency.
    All $\mathcal{G}_{i}$ (and their respective cospans) are MDA by definitions.
    Otherwise, if $\mathcal{G} \not = \mathcal{G}''$ then $\mathcal{G}''$ is a convex down-closed sub-e-hypergraph of $\mathcal{G}$ and 
    by applying the same construction as above we get~(\ref{fig:decomposition:1}).
\end{proof}

\ifdefined\ONECOLUMN
\section{Additional figures for Section~\ref{sec:introduction}:Introduction}
\else
\subsection{Additional figures for Section~\ref{sec:introduction}:Introduction}
\fi

This section contains the figures with more elaborate transformation of e-string diagrams from the Figure~\ref{fig:e-graph-example}.

\ifdefined\ONECOLUMN
\begin{figure}[htb!]
    \vspace{-3cm}
    \centering
    \[
        \hspace{1.3cm}
        \resizebox{0.8\textwidth}{!}{
        \input{./figures/egraph-translation-step-by-step-b-c.tikz}
        }
    \]
    \caption{Example translation from $(b)$ to $(c)$.}
    \label{fig:e-graph-example-b-c}
\end{figure}
\else
\begin{figure*}[htb!]
    \vspace{-3cm}
    \centering
    \[
        \hspace{1.3cm}
        \resizebox{0.8\textwidth}{!}{
        \input{./figures/egraph-translation-step-by-step-b-c.tikz}
        }
    \]
    \caption{Example translation from $(b)$ to $(c)$.}
    \label{fig:e-graph-example-b-c}
\end{figure*}
\fi

\ifdefined\ONECOLUMN
\begin{figure}[htb!]
    \[
        % \hspace{1.3cm}
        \scalebox{0.4}{
        \input{./figures/egraph-translation-step-by-step-a-b.tikz}
        }
    \]
    \caption{Example translation from $(a)$ to $(b)$.}
    \label{fig:e-graph-example-a-b}
\end{figure}
\else
\begin{figure*}[htb!]
    \[
        % \hspace{1.3cm}
        \scalebox{0.5}{
        \input{./figures/egraph-translation-step-by-step-a-b.tikz}
        }
    \]
    \caption{Example translation from $(a)$ to $(b)$.}
    \label{fig:e-graph-example-a-b}
\end{figure*}
\fi

\ifdefined\ONECOLUMN
\section{Additional figures for Section~\ref{sec:e-hypergraphs}: E-hypergraphs}
\else
\subsection{Additional figures for Section~\ref{sec:e-hypergraphs}: E-hypergraphs}
\fi
\label{sec:appendix:iso}

First we will explore the notion of isomorphism of two extended cospans in more detail.
Our example will be Figure~\ref{fig:appendix:non-isomorphic-cospans}.
Technically, we do not allow hierarchical edges with no delimiting dashed line, but here, for the purposes of illustration we make an exception to make the example compact.
\begin{example}
    Consider cospans from Figure~\ref{fig:appendix:non-isomorphic-cospans}(a) where the $f_{int}, f_{ext}$ (respectively, $g_{int},g_{ext}$) maps are shown with red arrows.
    The left cospan represents a symmetry placed inside a hierarchical edge while the right one depicts a tensor product of two identities put inside the hierarchical edge.
    The colours of the vertices show the induced partition of the interfaces.
    We will try to build an isomorphism step-by-step.
    To make the diagram commute, clearly, it must be the case that $u_0 \mapsto u_0$ and $u_1 \mapsto u_1$.
    Then, because the isomorphism should preserve the order within the subset of green vertices, it must be the case that $u_2 \mapsto u_2$ and $u_3 \mapsto u_3$.
    Similarly for the output interfaces.
    To further make the diagram commute, we need to map $v_2 \mapsto v_3$ and $v_3 \mapsto v_2$.
    This, however, contradicts the requirement that $g_{int};\gamma(w_2) = g_{int}(w_2)$.
    Hence, the two cospans are not isomorphic.
\end{example}

\begin{example}
    Now consider the cospans from Figure~\ref{fig:appendix:non-isomorphic-cospans}(b).
    Notice how the position of external input interface relative to strict internal interface is different in the second cospan.
    The isomorphism would map $u_0 \mapsto u_2, u_1 \mapsto u_1, u_2 \mapsto u_0, u_3 \mapsto u_3$.
    The other mappings are identities. 
    One can check that the corresponding diagram commutes.
\end{example}

Then we will elaborate on the notion of a boundary complement.

\begin{example}

Consider a pushout square in Figure~\ref{fig:not_boundary_complement} (a) where the labels suggest how arrows map vertices both from the interfaces and between the e-hypergraphs.
While this is a valid pushout square, such pushout complement does not correspond to a rewrite of $\Sigma^{+}$-terms --- as there is no redex for $f \otimes g$ within $(f + g);g$ even modulo SMC equations.
This pushout fails to be a boundary complement by violating conditions (3) and (4) of Definition~\ref{def:boundary_new}.
Now consider the square in subfigure (b): it fails to be a boundary complement by violating conditions (4) and (7).
\end{example}

\begin{example}

Now consider a full DPO square in Figure~\ref{fig:not_boundary_complement} (c).
First note that the pushout complement is a boundary complement as it satisfies all the conditions in Definition~\ref{def:boundary_new}.
In particular, observe that in order to make an extended MDA cospan with the boundary complement as a carrier, strict internal interfaces should be empty.
Indeed, strict internal interfaces of the left-hand side coincide with the ones of the source e-hypergraph.
Then, for the resulting e-hypergraph to form a valid well-typed MDA extended cospan we need to introduce the strict internal interfaces of the right-hand side of the rewrite rule.
\end{example}

\begin{figure*}
    \begin{subfigure}[T]{0.48\textwidth}
    \begin{subfigure}[T]{0.45\textwidth}
        \[
    \scalebox{0.55}{
        \input{./figures/sym_cospan.tikz}
    }    
    \]
    \end{subfigure}
    \hfill
    \begin{subfigure}[T]{0.45\textwidth}
        \[
            \scalebox{0.55}{
                \input{./figures/id_x_id_cospan.tikz}
            }    
            \]
    \end{subfigure}
    \subcaption{\;}
\end{subfigure}
    \hfill
\begin{subfigure}[T]{0.48\textwidth}
    \begin{subfigure}[T]{0.45\textwidth}
        \[
            \scalebox{0.55}{
                \input{./figures/id_x_id_cospan.tikz}
            }    
            \]
    \end{subfigure}
    \hfill
    \begin{subfigure}[T]{0.45\textwidth}
        \[
            \scalebox{0.55}{
                \input{./figures/id_x_id_cospan_2.tikz}
            }    
            \]
    \end{subfigure}
    \subcaption{\;}
\end{subfigure}
    \caption{Non-isomorphic cospans (a) and isomorphic cospans (b)}
    \label{fig:appendix:non-isomorphic-cospans}
\end{figure*}

\begin{figure*}
    \begin{subfigure}{0.48\textwidth}
\[
\adjustbox{scale=0.55}{
    \input{./figures/boundary_complement_non_example.tikz}
}
\]
\subcaption{Boundary complement non-example}
    \end{subfigure}
\hfill
    \begin{subfigure}{0.48\textwidth}
        \[
        \adjustbox{scale=0.55}{
            \input{./figures/boundary_complement_non_example_2.tikz}
        }
        \]
            \subcaption{Boundary complement non-example}
            \end{subfigure}
\begin{subfigure}{\textwidth}
    \[
    \adjustbox{scale=0.55}{
        \input{./figures/boundary_complement_example.tikz}
    }
    \]
        \subcaption{Complete DPO square}
        \end{subfigure}
\caption{DPO examples}
\label{fig:not_boundary_complement}
\end{figure*}

\ifdefined\ONECOLUMN
\section{DPO rewriting for arbitrary signatures}
\else
\subsection{DPO rewriting for arbitrary signatures}
\fi
\label{sec:dpo-fix}
Here we will discuss how one can formulate DPO rewrite rules where the left-hand of right-hand sides contain sub-hypergraphs with no inputs and outputs at the outermost level.
First we will illustrate where the restriction in the original formulation of the DPO rewriting as per the Definition~\ref{def:convex_dpo} comes from via several examples: one for Cartesian equations (which are important to express e-graphs as e-hypergraphs) and two for the structural rules we need for $\WellTypedMdaEcospans / \mathcal{S}$ to be a category enriched over $\catname{SLat}$.
Consider a rewrite rule induced by the Cartesian structure $a;! = id_{I}$ for $a : 0 \to 1$ which is rendered string diagrammatically as 
\[
\adjustbox{scale=0.6}{
\begin{tikzpicture}
	\begin{pgfonlayer}{nodelayer}
		\node [style=hedge] (0) at (-0.5, 4) {$a$};
		\node [style=hedge] (1) at (-0.5, 1) {$!$};
		\node [style=none] (2) at (1.5, 2.5) {\Large$=$};
		\node [style=empty diag] (3) at (4, 2.5) {};
	\end{pgfonlayer}
	\begin{pgfonlayer}{edgelayer}
		\draw (1) to (0);
	\end{pgfonlayer}
\end{tikzpicture}

}
\]
which becomes the following extended cospan of e-hypergraphs
\[
\adjustbox{scale=0.6}{
\input{./figures/cartesian_rewrite_hypergraph.tikz}
}
\]

Assume we want to perform a rewrite within the e-hypergraph in Figure~\ref{fig:dpo-stuck} (left).
However, we are immediately stuck as there is no matching of the left-hand side of the rule within the graph as there is no pushout complement: the only way we can impose a constraint that an edge shall have a particular predecessor in the pushout is through its adjacent vertices, but the sub-hypergraph for $a;!$ does not have any as demonstrated in the right subfigure.

\begin{figure*}[t!]
    \begin{subfigure}{0.4\linewidth}
    \[
    \adjustbox{scale=0.7}{
    \input{./figures/cartesian_rewrite_source_graph.tikz}
    }
    \]
    \end{subfigure}
    \hfill
    \begin{subfigure}{0.55\linewidth}
        \[
        \adjustbox{scale=0.7}{
        \input{./figures/cartesian_rewrite_dpo_stuck.tikz}
        }
        \]
    \end{subfigure}
    \caption{DPO rewriting stuck}
    \label{fig:dpo-stuck}
\end{figure*}

\begin{figure*}[!t]
    \centering
    \begin{minipage}[t]{0.48\linewidth}
        \centering
            \begin{subfigure}{\linewidth}
                \[
                \adjustbox{width=0.95\linewidth}{
                \input{./figures/idI_elimination_nested_hypergraph.tikz}
                }
                \]
            \end{subfigure}
            \vspace{1em}
            \begin{subfigure}{\linewidth}
                \[
                \adjustbox{width=0.95\linewidth}{
                \input{./figures/idI_introduction_nested_hypergraph.tikz}
                }
                \]
            \end{subfigure}
            \caption{Special edge nested introduction (top) and elimination (bottom) rules}
            \label{fig:nested_elimination}
    \end{minipage}%
    \hfill
    \begin{minipage}[t]{0.48\linewidth}
        \centering
            \begin{subfigure}{\linewidth}
                \[
                \adjustbox{width=0.95\linewidth}{
                \input{./figures/idI_introduction_new_component_hypergraph.tikz}
                }
                \]
            \end{subfigure}
            \vspace{1em}
                \begin{subfigure}{\linewidth}
                    \[
                    \adjustbox{width=0.95\linewidth}{
                    \input{./figures/idI_elimination_new_component_hypergraph.tikz}
                    }
                    \]
                \end{subfigure}
            \caption{New consistent component introduction (top) and elimination (bottom) rules}
            \label{fig:consistent_introduction}        
    \end{minipage}
\end{figure*}
\begin{figure}[t!]
\[
\adjustbox{width=\linewidth}{
    \input{./figures/structural_rewrite_dpo_stuck_fix.tikz}
}
\]
\caption{Complete DPO square for free enrichment structural rule}
\label{fig:completed_dpo}
\end{figure}

Another problematic rewrite arises when we want to utilise the structural rule of the free enrichment
\ifdefined\ONECOLUMN
\[
\adjustbox{width=0.6\linewidth}{
\input{./figures/structural_rewrite_problem.tikz}
}
\]
\else
\[
\adjustbox{width=\linewidth}{
\input{./figures/structural_rewrite_problem.tikz}
}
\]
\fi
if we consider the following DPO square
\ifdefined\ONECOLUMN
\[
\adjustbox{width=0.8\linewidth}{
    \input{./figures/structural_rewrite_dpo_stuck.tikz}
}
\]
\else
\[
\adjustbox{width=\linewidth}{
    \input{./figures/structural_rewrite_dpo_stuck.tikz}
}
\]
\fi

% we can not express id_I + id_I in such a way
% The last example is the rewrite that corresponds to the idempotence for the identity on the tensor unit $id_{I} = id_{I} + id_{I}$

% \[
% \adjustbox{scale=0.75}{
%     \tikzfig{../figures/appendix/structural_id_unit_problem}
% }
% \]

% with the DPO square as 

% \[
% \adjustbox{scale=0.75}{
% \tikzfig{../figures/appendix/structural_rewrite_id_unit_dpo_stuck}
% }
% \]

\begin{definition}

We will define the action of $+$ on two cospans where either of the carriers is the empty e-hypergraph as
\begin{align*}
0 \xrightarrow{} n \xrightarrow{} &\;\mathcal{F} \xleftarrow{} m \xleftarrow{} 0\\
&\;+ \hspace{6em} \Coloneqq \hspace{2em} 0  \xrightarrow{} n \xrightarrow{} \;\mathcal{F} \xleftarrow{} m \xleftarrow{} 0\\
0 \xrightarrow{} 0 \xrightarrow{} &\;\varnothing \xleftarrow{} 0 \xleftarrow{} 0.
\end{align*}
\end{definition}

This is so that we can absorb some equations and to seamlessly introduce what follows.

% \begin{remark}
% By putting no restrictions on the signature we need to enforce some restrictions on the equations in $\mathcal{E}$.
% Recall that $\llbracket id_{I} + id_{I}\rrbracket = \llbracket id_{I} \rrbracket$ as per the Remark~\ref{remark:f+g}.
% This implies that for all $f : 0 \to 0$, $\llbracket f + id_{I} \rrbracket$
% Consider a cospan of e-hypergraphs that corresponds to a morphism $f + g + (id_{I} + id_{I})$

% \[
% \adjustbox{scale=0.75}{
%     \tikzfig{../figures/appendix/remark_on_equations}
% }
% \]

% By using the idempotence equation for $id_{I} + id_{I}$ we can rewrite it to

% \[
%     \adjustbox{scale=0.75}{
%         \tikzfig{../figures/appendix/remark_on_equations_2}
%     }
% \]

% However, there is no rewrite rule to go back and equate these two cospans in $\MdaEcospans / \mathcal{S,E}$.
% To fix that, we need to introduce the following rewrite rule schema
% \[
% \adjustbox{scale=0.75}{
%     \tikzfig{../figures/appendix/remark_on_equations_3}
% }
% \]
% By introducing this equation we automatically require that every morphism $0 \to 0$ is essentially $id_{I}$.

% \end{remark}

We will start completing all the above DPO squares by first considering the structural rules.
We will include special unlabelled edges with no inputs and outputs to the interfaces to be able to impose the hierarchical relation on subgraphs with no inputs and outputs and will not impose any restrictions on labelling preservation when mapping these special edges.
That is, the previously mentioned structural rule will become

\ifdefined\ONECOLUMN
\[
\adjustbox{width=0.7\linewidth}{
    \input{./figures/structural_rewrite_problem_fix.tikz}
}
\]
\else
\[
\adjustbox{width=\linewidth}{
    \input{./figures/structural_rewrite_problem_fix.tikz}
}
\]
\fi

where the labels show how interface edges are mapped to the edges of the corresponding e-hypergraphs.
For example, the edge in the interface labelled $a_1$ is mapped to the outermost edge on the right.
We can now compute the pushout complement and complete the DPO square \ifdefined\ONECOLUMN below \else in Figure~\ref{fig:completed_dpo}. \fi

\ifdefined\ONECOLUMN
\[
\adjustbox{width=0.8\linewidth}{
    \input{./figures/structural_rewrite_dpo_stuck_fix.tikz}
}
\]
\else
% \[
% \adjustbox{width=\linewidth}{
%     \tikzfig{figures/structural_rewrite_dpo_stuck_fix}
% }
% \]
\fi

Handling deletion or introduction rewrite rules will require the introduction of auxiliary rules.
In particular, two rules for introduction and elimination of these special edges (Figure~\ref{fig:idI_intro_elim}),
\ifdefined\ONECOLUMN
\begin{figure}[h!]
    \begin{subfigure}{0.45\linewidth}
        \[
        \adjustbox{scale=0.7}{
        \input{./figures/idI_elimination_hypergraph.tikz}
        }
        \]
    \end{subfigure}
    \hfill
    \begin{subfigure}{0.45\linewidth}
        \[
        \adjustbox{scale=0.7}{
        \input{./figures/idI_introduction_hypergraph.tikz}
        }
        \]
    \end{subfigure}
    \end{figure}
\else
\begin{figure}[h!]
\begin{subfigure}{0.45\linewidth}
    \[
    \adjustbox{scale=0.55}{
    \input{./figures/idI_elimination_hypergraph.tikz}
    }
    \]
\end{subfigure}
\hfill
\begin{subfigure}{0.45\linewidth}
    \[
    \adjustbox{scale=0.55}{
    \input{./figures/idI_introduction_hypergraph.tikz}
    }
    \]
\end{subfigure}
\caption{Special edge flat introduction and elimination}
\label{fig:idI_intro_elim}
\end{figure}
\fi
similar rules to introduce and eliminate in nested contexts (Figure~\ref{fig:nested_elimination});
% \begin{figure}[t!]
%     \begin{subfigure}{\linewidth}
%         \[
%         \adjustbox{scale=0.55}{
%         \tikzfig{figures/idI_elimination_nested_hypergraph}
%         }
%         \]
%     \end{subfigure}
%     \vspace{1em}
%     \begin{subfigure}{\linewidth}
%         \[
%         \adjustbox{scale=0.55}{
%         \tikzfig{figures/idI_introduction_nested_hypergraph}
%         }
%         \]
%     \end{subfigure}
%     \caption{Special edge nested elimination}
%     \label{fig:nested_elimination}
% \end{figure}
and, finally, the rules to introduce such edges as new consistent components (Figure~\ref{fig:consistent_introduction}).
Intuitively such edge plays the role of a tensor unit identity that can be inserted in any context.

\ifdefined\ONECOLUMN
\begin{figure}[h!]
    \begin{subfigure}{0.6\linewidth}
        \[
        \adjustbox{width=\linewidth}{
        \input{./figures/idI_introduction_new_component_hypergraph.tikz}
        }
        \]
        \subcaption{}
    \end{subfigure}
    \begin{subfigure}{0.6\linewidth}
        \[
        \adjustbox{width=\linewidth}{
        \input{./figures/idI_elimination_new_component_hypergraph.tikz}
        }
        \]
        \subcaption{}
    \end{subfigure}
    \caption{New consistent component introduction (a) and elimination (b) rule}
    \label{fig:consistent_introduction}
\end{figure}
\else
% \begin{figure}[t!]
%     \begin{subfigure}{\linewidth}
%         \[
%         \adjustbox{width=\linewidth}{
%         \tikzfig{figures/idI_introduction_new_component_hypergraph}
%         }
%         \]
%     \end{subfigure}
%     \vspace{1em}
%         \begin{subfigure}{\linewidth}
%             \[
%             \adjustbox{width=\linewidth}{
%             \tikzfig{figures/idI_elimination_new_component_hypergraph}
%             }
%             \]
%         \end{subfigure}
%     \caption{New consistent component introduction (top) and elimination (bottom) rule}
%     \label{fig:consistent_introduction}
% \end{figure}
\fi

% \[
% \adjustbox{scale=0.75}{
% \tikzfig{../figures/appendix/auxiliary_schema}
% }
% \]

% This rule makes sense in any SMT where there is a unique morphism $I \to I$ which is $id_{I}$.
Then the deletion rule is reconstructed as

\ifdefined\ONECOLUMN
\[
\adjustbox{width=0.6\linewidth}{
    \input{./figures/cartesian_rewrite_fix.tikz}
}
\]
\else
\[
\adjustbox{width=0.8\linewidth}{
    \input{./figures/cartesian_rewrite_fix.tikz}
}
\]
\fi
and the full DPO square becomes

\ifdefined\ONECOLUMN
\[
\adjustbox{width=0.7\linewidth}{
    \input{./figures/cartesian_rewrite_dpo_fix.tikz}
}
\]
\else
\begin{figure}[h!]
\[
\adjustbox{width=\linewidth}{
    \input{./figures/cartesian_rewrite_dpo_fix.tikz}
}
\]
\caption{Complete DPO square for the deletion rule}
\label{fig:cartesian_deletion_dpo}
\end{figure}
\fi
The target e-hypergraph is not exactly what we wanted, however, but we can use the auxiliary rule to remove this unlabelled edge.

\ifdefined\ONECOLUMN
\[
\adjustbox{width=0.7\linewidth}{
    \input{./figures/cartesian_rewrite_dpo_fix_2.tikz}
}
\]
\else
\begin{figure}
\[
\adjustbox{width=\linewidth}{
    \input{./figures/cartesian_rewrite_dpo_fix_2.tikz}
}
\]
\caption{Complete DPO square for the deletion rule (after one more rule applied)}
\label{fig:cartesian_dpo_fix}
\end{figure}
\fi
% and then utilise the idempotence rule for the identity of the tensor unit

% \[
%     \adjustbox{scale=0.75}{
%         \tikzfig{figures/cartesian_rewrite_dpo_fix_3}
%     }
% \]

Generally, given an interpretation of a structural rule (where the right-hand side contains a single edge at the outermost level) $\llangle \llbracket l \rrbracket, \llbracket r \rrbracket \rrangle$, for each connected component $C_{i}$ at the outermost level with no inputs and outputs and for each hyperedge with no inputs \textit{or} outputs  in $C_{i}$  we introduce an empty hyperedge into the interfaces which are mapped to the single edge at the outermost level in $\llbracket r \rrbracket$.
Symmetric rules are handled by performing the above in the other direction.

Rules for deletion and introduction from Cartesian structure are interpreted similarly, by mapping all interface edges to a single empty hyperedge on the right (respectively, left) when deleting (respectively, introducing).
For an arbitrary rule $\llangle \llbracket l \rrbracket, \llbracket r \rrbracket \rrangle$ induced by $\mathcal{E}$ consider connected components with no inputs and outputs $C^{l}_{i}$ and $C^{r}_{j}$ for left-hand and right-hand sides respectively.
Let $c^{l}_{1,1}, \ldots, c^{l}_{k,n}$ and $c^{r}_{1,1}, \ldots, c^{r}_{p,m}$ where $c^{l}_{i,j}$ is $i$-th edge with no inputs or outputs in component $C^{l}_{j}$ and similarly for $c^{r}_{i,j}$ and let the interface contain the largest of the two sets.
And let the morphism be any valid morphism from $\{^{l}_{1,1}, \ldots, c^{l}_{k,n}\} \to c^{r}_{1,1}, \ldots, c^{r}_{p,m}$ is the former is larger as a set, or the other way around if the latter is larger.
Graphically, this is expressed as follows

\ifdefined\ONECOLUMN
\[
\adjustbox{width=0.7\linewidth}{
    \input{./figures/rule_transformation_example.tikz}
}
\]
\else
\[
\adjustbox{width=\linewidth}{
    \input{./figures/rule_transformation_example.tikz}
}
\]
\fi
assuming that $|c^{l}_{1,1}, \ldots, c^{l}_{k,n}| > |c^{r}_{1,1}, \ldots, c^{r}_{p,m}|$ and $f = \{ c^{l}_{1,1} \mapsto c^{l}_{1,1}, \ldots, c^{l}_{k,n} \mapsto c^{l}_{k,n} \}$ and $g$ is any valid mapping of $c^{l}_{1,1}, \ldots, c^{l}_{k,n}$ to $c^{r}_{1,1}, \ldots, c^{r}_{p,m}$.

The Definition~\ref{def:boundary_new} will then need a few additional constraints.
\begin{itemize}
\item For all top-level edges $e_{i}, e_{j}$ of type $0 \to 0$ in $\mathcal{L}$ it must be either $[m(e_{i})) = [m(e_{j})) = \varnothing$ or $m(e_{i}) \consistency m(e_{j})$.
\item For all edges $e_{i}, e_{j}$ of type $0 \to 0$ in $i + j$ it must be either $[[c_1,c_2](e_{i})) = [[c_1,c_2](e_{j})) = \varnothing$ or $[c_1,c_2](e_{i}) \consistency [c_1,c_2](e_{j})$.
\item For all edges $e$ of type $0 \to 0$ in $i + j$ it must be that $l([c_1,c_2](e)) = l(e) = \bot$ and $[c_1,c_2](e)$ must be maximal.
\end{itemize}
Then, when constructing the corresponding extended cospans, such auxiliary edges must be ignored in the interfaces.

\ifdefined\ONECOLUMN
\section{E-graph interpretation in $\catname{MEHypI}(\Sigma_{C}) / \mathcal{S}, \mathcal{E}_{C}$}
\else
\subsection{E-graph interpretation in $\catname{MEHypI}(\Sigma_{C}) / \mathcal{S}, \mathcal{E}_{C}$}
\fi
\label{sec:appendix:e-graph-translation}

In this section we will elaborate on the interpretation of e-graphs as morphisms of $\catname{MEHypI}(\Sigma_{C}) / \mathcal{S}, \mathcal{E}_{C}$.
We will first recite the formal definition of an e-graph.

\begin{figure}

\begin{align*}
    \text{functional symbols}& \hspace{2em} f,g\\
    \text{e-class id}& \hspace{2em} a,b\\
    \text{terms}& \hspace{2em} t \Coloneqq f \;|\; f (t_1, \ldots, t_m) \hspace{2em} &m \geq 1\\
    \text{e-nodes}& \hspace{2em} n \Coloneqq f \;|\; f (a_1, \ldots, a_m) \hspace{2em} &m \geq 1\\
    \text{e-classes}& \hspace{2em} c \Coloneqq \{n_1, \ldots, n_m\}   \hspace{2em} &m \geq 1
\end{align*}
\caption{E-graph components}
\label{fig:e-graph-components}
\end{figure}

\begin{definition}[E-graph~\cite{EggPaper}]
    E-graph is a data-structure that consists of
    \begin{itemize}
        \item components defined in a grammar in Figure~\ref{fig:e-graph-components}
        \item a union find structure $U$ that stores equivalence relation on e-class ids
        \item \textit{e-class} map $M$ that maps e-class ids to e-classes. All equivalent e-class ids map to the same e-class, \textit{i.e.}, if $a \equiv_{id} b$ then $M[a]$ and $M[b]$ are the same set.
              An e-class id $a$ is said to refer to e-class $M[\text{find}(a)]$. NB: $\text{find}()$ returns the canonical representative for the equivalence, i.e. the last predecessor of $a$ in $U$.
        \item The \textit{hashcons} map $H$ that maps e-nodes to e-class ids.
    \end{itemize}
\end{definition}

\begin{definition}[Congruence invariant~\cite{EggPaper}]
Two e-nodes $f(a_{1}, \ldots, a_{n})$ and $f_{b_{1}, \ldots, b_{n}}$ are \textit{congruent}, if $\text{find}(a_{i}) = \text{find}(b_{i})$ for all $i$.
The equivalence relation on e-nodes must be closed under this congruence, that is, two congruent e-nodes must reside in the same e-class.
\end{definition}

\begin{definition}[Hashcons invariant~\cite{EggPaper}]
The hashcons $H$ must map all canonical e-nodes to their e-class ids:
\[
\text{e-node } n \in M[a] \text{ iff } H[\text{canonicalize}(n)] = \text{find}(a)
\]
\end{definition}

\begin{definition}[~\cite{EggPaper}]
A rewriting over an e-graph is given by composition of the following operations
\begin{itemize}
    \item \textit{add} takes an e-node $n$ and:
    \begin{itemize}
        \item if $n$ is present in $H$, return $H[n]$;
        \item else create e-class id $a$ and set $H[n] = a$ such that $M[a] = \{n\}$.
    \end{itemize}
    \item \textit{merge} given two e-class ids $a$ and $b$ unions them in the union find $U$ such that $M[a] = M[b] = M[a] \cup M[b]$.
\end{itemize}
\end{definition}
Then, given a rewrite rule $l \to r$, as, for example, $x * 2 \to x \ll 1$ in Figure~\ref{fig:e-graph-example}, applying it to an e-graph amounts several steps.
First, finding a tuple $(\sigma, c)$ where e-class id $c$ contains $\sigma(l)$ where sigma is the mapping from variables to e-class ids, for example, $c = H[a * 2]$ and $\sigma(l) = (a * 2)$.
Then, adding $\sigma(r)$ to the e-graph and merging $c$ with $add(\sigma(r))$.
The latter step may require additional merges to maintain the invariant.
If two terms $f(a,b)$ and $f(a,c)$ reside in different e-classes, merging e-classes for $a$ and $b$ would require merging e-classes for $f(a,b)$ and $f(a,c)$ as they become congruent as congruent e-nodes must reside in the same e-class.
The process of maintaining the invariants by additional merging is called \textit{upward merging}~\cite{EggPaper}.
We will denote the process of obtaining e-graph $e_2$ from e-graph $e_1$ after applying a rule $l \to r$ (including upward merging) as $e_1 \leadsto e_2$.

\begin{definition}
The pseudocode in Algorithm~\ref{alg:translation} defines the translation function $\llbracket \rrbracket : \catname{E}\text{-}\catname{graph} \to \catname{MEHypI}(\Sigma_{C}) / \mathcal{S}, \mathcal{E}_{C}$ that turns a \textit{canonical} (the one for which the above invariants hold) \textit{acyclic} \textit{connected} e-graph into a cospan of e-hypergraphs.
\end{definition}

For the following e-graph with e-classes shown as dashed boxes it produces the following cospan of e-hypergraphs.
\texttt{connect} and \texttt{mkEdge} are responsible for making things shared using $\triangle$ and discarding unused inputs via $!$.

\begin{figure}
    \begin{subfigure}[C]{0.35\linewidth}
\[
\adjustbox{scale=0.6}{
\input{./figures/e-graph-example_1.tikz}
}
\]
\caption{\;}
\end{subfigure}
\hfill
\begin{subfigure}[C]{0.6\linewidth}
    \[
    \adjustbox{scale=0.6}{
    \input{./figures/e-hypergraph-example_1.tikz}
    }
    \]
\caption{\;}
\end{subfigure}
\hfill
\caption{E-graph (a) and a corresponding morphism of $\catname{MEHypI}(\Sigma_{C}) / \mathcal{S}, \mathcal{E}_{C}$ (b)}
\label{fig:e-graph-to-cospan}
\end{figure}

\begin{algorithm*}
    \caption{E-graph to e-hypergraph translation}
    \label{alg:translation}
    \begin{algorithmic}[1]
    \Function{e-class-to-HEdge}{e\_class\_id}
    \State e\_nodes $\gets$ [e\_node$_{1}$ $\ldots$ e\_node$_{n}$] \textbf{in} H.keys() \textbf{where} H[e\_node$_{i}$] == M[e\_class\_id]
    % \State \textbf{for} [$\text{e_node}_{1}, \ldots, \text{e_node}_{n}$] \textbf{in} H \textbf{where} H[$\text{e_node}_{i}$] == M[e-class-id] \texbf{do}
    \ForAll {e\_node$_{i}$ in e\_nodes}
    \State edge $\gets$ mkEdge(e\_class\_id)
    \Comment {Create a hierarchical edge that represents an e-class}
    \State \textbf{match} e\_node$_{i}$ \textbf{with}
    \State \qquad \textbf{case} f($e_1 \ldots e_{n}$) :
    \State \qquad \qquad \textbf{if} is\_singleton(e\_class\_id) \textbf{then}
    % \If{is\_singleton(e\_class\_id)}
        \State \qquad\qquad \hspace{2em} edge $\gets$ mkEdge(f)
        \Comment {Create an edge that represents the generator f $: k \to 1$}
        \State \qquad \qquad \hspace{2em} edges $\gets$ [\Call{e-class-to-HEdge}{$e_1$} $\ldots$ \Call{e-class-to-HEdge}{$e_n$}]
        \State \qquad \qquad \hspace{2em} edge.connect(edges)
        \State \qquad \qquad \hspace{2em} \Return edge
    % \Else
    \State \qquad \qquad \textbf{else}
    \State \qquad \qquad \hspace{2em} graph $\gets$ discard([e\_node$_1$$\ldots$e\_node$_{i-1}$]) $\cup$ mkEdge(f) $\cup$ discard([e\_node$_{i+1}$ $\ldots$ e\_node$_{n}$])
        \State \qquad \qquad \hspace{2em} edges $\gets$ [\Call{e-class-to-HEdge}{$e_1$} $\ldots$ \Call{e-class-to-HEdge}{$e_n$}]
        \State \qquad \qquad \hspace{2em} edge.add(graph)
        \Comment {Add \texttt{graph} as a consistent component into hierarchical edge \texttt{edge}}
        \State \qquad \qquad \hspace{2em} edge.connect(edges)
    % \EndIf
    \EndFor
    \State \Return edge
    \EndFunction
    \State
    \State graph $\gets \varnothing$ 
    \ForAll{e\_class\_id in H.values()}
    \State graph $\gets$ graph $\cup$ \Call {e-class-to-HEdge}{e\_class\_id})

    \EndFor
    \end{algorithmic}
\end{algorithm*}

\begin{lemma} 
    Two e-classes $e_1$ and $e_2$ are merged if and only if there are rewrite rules that turn every term representable by $e_1$ into $e_2$.
\end{lemma}

\begin{lemma}
Given a set of equations $\mathcal{E}$ between $\Sigma$-terms that induces a set of rewrite rules $l \to r$ and $r \to l$ for every $l = r \in \mathcal{E}$, $e_1 \leadsto e_2$ if and only if $\llbracket e_1 \rrbracket \Rrightarrow{}^{*} \llbracket e_2 \llbracket$.
\end{lemma}
\begin{proof}

We will prove this by example showing how an addition of an e-node followed by a series of merges can be expressed as a series of DPO rewrites.
Consider an example in Figure~\ref{fig:e-graph-rewrite}.
It shows the rewrite $b \to c$ applied to e-graph that initially encodes a single term $f(a,b) + f(a,c)$.
The first step is $\textit{merge}(b, \textit{add}(c))$ which unions e-classes for $b$ and $c$ by making the latter be the parent of the former, \textit{i.e.}, \text{find}$(b) = c$.
Then it recanonicalizes the e-node for $f(a,b)$ so that it points to the canonical e-class id for $b$.
This makes it possible to further merge two e-classes as they contain congruent e-nodes.
The interpretation of these steps as DPO rewrites is given in Figure~\ref{fig:e-string-rewrite}.
Subfigure $(a)$ is the direct translation of the initial e-graph by using $\llbracket \rrbracket$.
Then we utilize the idempotence twice as $b = b + b$ and $c = c + c$ (subfigure (b)) and use the fact that our rewrite rules are symmetric as they are induced by the set of equations and apply $b \to c$ and $c \to b$ (subfigure (c)) and then use the rules for sharing from the Cartesian structure (subfigures (d)-(e)).
The reverse direction follows by finding normal forms for $\llbracket e_{1} \rrbracket$ and $\llbracket e_{2} \rrbracket$.
\end{proof}

\begin{remark}

Technically, rewrite rules formed by $\Sigma$-terms do not contain any variables as they are not a part of the signature while the rewrite rules for a given e-graph do contain variables as, for example, $x * 2 \to x \ll 1$.
This is however not a limitation, as variables are merely interfaces and we can express the above-mentioned rule as a pair of $\Sigma$-terms $*;id \otimes 2 \to \ll ; id \otimes 1$.
\end{remark}

\begin{remark}
Technically, $\mathcal{E}_{C}$ is not induced by equations between $\Sigma$-terms, but rather between $\Sigma^{+}$-terms as we can apply naturality of copy and delete to boxes as well.
This does not pose any issues as the proof of Theorem~\ref{thm:full-completeness} works for any $\mathcal{E}$ as long as the conditions of Definition~\ref{def:boundary_new} are satisfied.
\end{remark}

\begin{figure}

\[
\adjustbox{width=\linewidth}{
    \input{./figures/e-graph-rewrite-example.tikz}
}
\]
\caption{E-graph rewrite example.}
\label{fig:e-graph-rewrite}
\end{figure}

\begin{figure}

    \[
        \adjustbox{width=\linewidth}{
            \input{./figures/e-string-rewrite-example.tikz}
        }
        \]
        \caption{String diagram rewrite example.}
        \label{fig:e-string-rewrite}
\end{figure}

% \begin{lemma}
% Every hierarchical edge can be decomposed as follows.

% \[
% \tikzfig{../figures/appendix/hyperedge_decomposition}
% \]
% \end{lemma}

% Then, given e-classes $e_1$ and $e_2$ for sub-hypergraphs $C$ and $D$ respectively, we introduce e-nodes into the e-graph $A(e_1)$ and $B(e_2)$ that reside in e-class $e_3$.

\end{document}